\documentclass[aps,prx,twocolumn,floatfix,superscriptaddress,longbibliography]{revtex4-2}

\usepackage{amsmath,amssymb,amsthm,mathtools}
\usepackage{bm}
\usepackage{braket}
\usepackage{graphicx}
\usepackage{xcolor}
\usepackage{hyperref}
\hypersetup{colorlinks=true,linkcolor=blue!55!black,citecolor=teal,urlcolor=teal}
\usepackage{enumitem}
\usepackage{booktabs}
\usepackage{orcidlink}
\usepackage[capitalize]{cleveref}
\usepackage{dsfont}
\usepackage[normalem]{ulem}

\newcommand{\g}{\mathfrak{g}}
\newcommand{\h}{\mathfrak{h}}

\newcommand{\adphi}{\Phi^{\mathrm{ad}}}
\newcommand{\Adphi}{\Phi^{\mathrm{Ad}}}

\newcommand{\Tr}{\operatorname{Tr}}
\newcommand{\ii}{\mathrm{i}}
\newcommand{\dd}{\mathrm{d}}
\newcommand{\R}{\mathbb{R}}

\newcommand{\Om}{\Omega}
\newcommand{\rvec}{\bm{r}}
\newcommand{\Rvec}{\bm{R}}
\newcommand{\gsim}{\ensuremath{\mathfrak g\text{-}\mathrm{sim}}}
\newcommand{\poly}{\mathrm{poly}}

\DeclareMathOperator{\Sp}{Sp}
\DeclareMathOperator{\spann}{span}

\theoremstyle{plain}
\newtheorem{theorem}{Theorem}
\newtheorem{proposition}{Proposition}
\newtheorem{lemma}{Lemma}
\newtheorem{corollary}{Corollary}
\theoremstyle{definition}
\newtheorem{definition}{Definition}

\newtheorem{example}{Example}

\begin{document}

\title{Lie-Algebraic Classical Simulation of Bosonic Systems Beyond Gaussian Dynamics}

\author{Adelina B\"arligea\,\orcidlink{0009-0008-5497-1941}}
\email{adelina.baerligea@uni-a.de}
\thanks{These authors contributed equally.}
\affiliation{Institute for Computer Science, University of Augsburg, 86159 Augsburg, Germany}

\author{Timothy Heightman\,\orcidlink{0000-0002-3314-0929}}
\email{theightman@icfo.net}
\thanks{These authors contributed equally.}
\affiliation{ICFO -- Institut de Ci\`encies Fot\`oniques, The Barcelona Institute of
Science and Technology, 08860 Castelldefels (Barcelona), Spain}
\affiliation{Quside Technologies SL, Carrer d'Esteve Terradas 1, 08860 Castelldefels
(Barcelona), Spain}

\author{Jakob~S.~Kottmann\,\orcidlink{0000-0002-4156-2048}}
\affiliation{Institute for Computer Science, University of Augsburg, 86159 Augsburg, Germany}
\affiliation{Center for Advanced Analytics and Predictive Sciences, University of Augsburg, 86159 Augsburg, Germany}

\author{Antonio Acín\,\orcidlink{0000-0002-1355-3435}}
\affiliation{ICFO -- Institut de Ci\`encies Fot\`oniques, The Barcelona Institute of
Science and Technology, 08860 Castelldefels (Barcelona), Spain}
\affiliation{ICREA -- Institucio Catalana de Recerca i Estudis Avançats, Lluis Companys 23, 08010 Barcelona, Spain}

\begin{abstract}
Classical simulability is ultimately determined by both the dynamics of a quantum system and the observables being evaluated. Lie-algebraic simulation exploits the latter to make exact polynomial-time classical simulations by propagating observables through low-dimensional invariant operator spaces. However, its conventional formulation in terms of polynomial-dimensional dynamical Lie algebras does not directly accommodate bosonic systems as their algebras are neither compact nor semisimple. 
In this contribution, we overcome this limitation, making bosonic systems accessible to the Lie-algebraic programme of exact polynomial-time classical simulation. We prove that expectation values, fixed-order correlation functions, including multi-time correlators and out-of-time-ordered correlators, and gradients are efficiently computable whenever their operator modules have polynomial dimension. This recovers Gaussian quantum optics and extends it to non-Gaussian input states, while identifying exact polynomial regimes of interacting non-Gaussian dynamics including bounded-photon Kerr and pair-hopping Hamiltonians and nilpotent polynomial phase dynamics. We show that unlike in the finite-dimensional spin and fermionic setting treated previously, a finite-dimensional bosonic generator algebra alone does not guarantee finite observable dynamics.
We further derive a controlled perturbative hierarchy for squeezing beyond exact sector confinement and confirm the predicted error orders numerically. We also evaluate operator spreading on interacting chains of up to $400$ modes and connect a topological doublon band with flux-reversed edge motion. These results provide a unified formalism for classifying, discovering, and systematically approximating tractable bosonic quantum dynamics with classical polynomial-time simulation.
\end{abstract}

\maketitle


\section{Introduction}

Claims of quantum advantage are meaningful only relative to the best available classical description. Beyond delineating the boundary between quantum and classical computation, classical simulation is valuable as a tool for validating quantum devices and understanding quantum many-body dynamics. Bosonic systems bring this question into focus since their  canonical commutation relations place them in an infinite-dimensional Fock space even at a fixed number of modes. Describing photons, phonons, and collective excitations in quantum optics and many-body systems, they hence combine direct physical relevance with a structural challenge absent from spin systems and finite-dimensional Hilbert spaces.

The classical simulation of bosonic systems is already a mature subject~\cite{Lloyd1999,Bartlett2002,Weedbrook2012,Mari2012}, but its conclusions depend crucially on the problem being considered. Gaussian states, operations, and measurements admit efficient descriptions in terms of first and second moments, forming the foundation of classical simulation in continuous-variable quantum information~\cite{Bartlett2002,Weedbrook2012}. Yet passive Gaussian dynamics acting on non-Gaussian inputs also underlies boson sampling, where sampling from the complete output distribution is believed to be classically intractable under standard complexity assumptions~\cite{Aaronson2013}. Gaussian boson sampling instead uses Gaussian inputs such as squeezed states; sampling is likewise believed to be classically intractable under standard complexity assumptions, while its photon-counting probabilities are governed by hafnians rather than permanents~\cite{Hamilton2017}. These examples highlight a broader distinction, that sampling, computing individual outcome probabilities, and evaluating expectation values are operationally different problems~\cite{Bravyi.2021-ClassicalAlgorithmsQuantum,Pashayan2020}. The same physical dynamics may be tractable for one output and hard for another.

This work addresses the quantum mean value problem, in which the exact evaluation of expectation values, fixed-order and multi-time correlators, out-of-time-ordered correlators (OTOCs)~\cite{Larkin.1969-QuasiclassicalMethodTheory}, and parameter gradients are to be simulated. These quantities provide direct access to transport, correlations, response, operator spreading, phase transitions, and variational objectives without requiring reconstruction or sampling of the full many-body state. The question motivating this contribution is thus whether physically informative observables can remain exactly tractable under dynamics that need not be Gaussian. 

Across successful classical simulation methods, tractability ultimately reflects the existence of a representation that scales polynomially in system size, such as stabilizer descriptions for Clifford circuits~\cite{Gottesman.1997-Stabilizer,aaronson2004improved}, a bounded bond dimension in tensor-network methods~\cite{Perez-Garcia.2007-MPSrepresentations,Markov.2008-SimulatingQCbyContractingTNs}, or a covariance matrix for Gaussian dynamics both for fermions~\cite{Knill.2001-FermionicLinearOptics,Valiant.2001-QuantumComputersThat,Terhal.2002-ClassicalSimulationNoninteractingfermion} and bosons~\cite{Wang.2007-QIGaussian,Weedbrook2012}. Lie-algebraic simulation follows the same principle in the Heisenberg picture. Instead of propagating a state vector or density matrix, it evolves the observables of interest within a low-dimensional invariant operator space.
Somma et al.~\cite{Somma.2005-QuantumComputationComplexity,Somma.2006-EfficientSolvabilityHamiltonians} established this approach for systems governed by polynomial-dimensional Lie algebras admitting suitable finite-dimensional faithful representations. More recently, this approach was recast in an adjoint-space formulation within the $\gsim$ framework~\cite{Goh.2025-LiealgebraicClassicalSimulations}, which propagates observables within invariant subspaces of the adjoint action. Later, Ref.~\cite{Barligea.2026-EnablingLieAlgebraicClassical} extended this beyond the standard free-fermion setting through symmetry-adapted operator bases, including permutation-invariant systems and dynamics on bounded-Hamming-weight sectors.

Bosons, however, appear to fall outside this paradigm. Their Hilbert space is infinite-dimensional. The canonical quadratures generate phase-space displacements and, together with the identity, form the non-semisimple Heisenberg--Weyl algebra. Quadratic Hamiltonians generate Gaussian transformations through the non-compact real symplectic algebra. Indeed, the original Lie-algebraic framework explicitly left the bosonic case outside its scope because the relevant algebra is not semisimple and its physical representations are infinite-dimensional~\cite{Somma.2006-EfficientSolvabilityHamiltonians}. These properties obstruct a direct application of the original argument through a finite-dimensional faithful representation of the generator algebra. Crucially, however, they do not preclude finite-dimensional Heisenberg evolution for the particular observable being evaluated. 

We exploit this distinction to extend Lie-algebraic classical simulation to bosonic systems. Starting from the invariant-subspace propagation principle underlying $\gsim$, we construct the relevant invariant operator space, which we call the \emph{reachable operator module}, directly from the dynamics--observable pair. For a fixed set of circuit generators and an observable, this module is the smallest linear operator space containing the observable that is invariant under commutation with every generator. In the finite-dimensional setting, this is the minimal invariant subspace required for the propagation developed in Ref.~\cite{Goh.2025-LiealgebraicClassicalSimulations}. In the bosonic setting, however, it becomes the central classifying object for efficient simulation; it can remain finite and polynomially accessible even though the Hilbert space and ambient operator algebra are infinite-dimensional and the associated Lie-algebra representation falls outside the conventional assumptions. \cref{fig:summary} previews this observable-seeded construction, whose formal definition and simulation criterion are developed in \cref{sec:finiteopertormodules}.

\begin{figure*}[htbp]
  \centering
  \includegraphics[width=\textwidth]{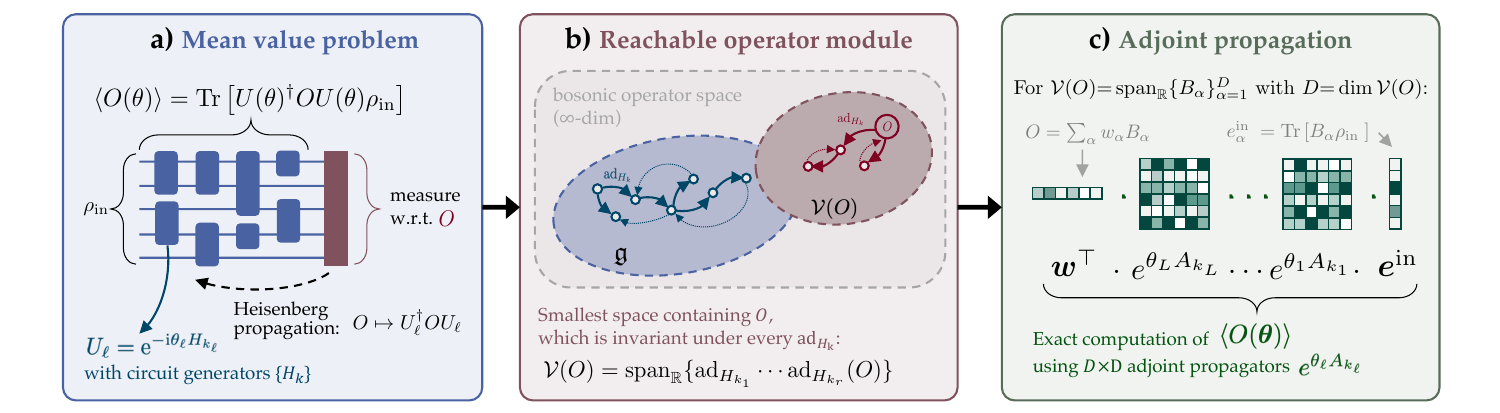}
  \caption{Reachable-operator-module simulation workflow.
  (a) For a circuit $U(\bm\theta)=U_L\cdots U_1$ with layers $U_\ell=e^{-\ii\theta_\ell H_{k_\ell}}$ as in \cref{eq:circuit-unitary}, the mean-value problem is evaluated by propagating the target observable backwards through the circuit in the Heisenberg picture, as in \cref{eq:mean-value-task}.
  (b) The Hermitian DLA $\g$ is generated from the Hamiltonians by repeated applications of $\mathrm{ad}_H(A)=\ii[H,A]$ (see \cref{eq:hermitian-ad,eq:DLA}), whereas the reachable operator module $\mathcal V(O)$ is seeded by the observable and is the smallest real operator space containing $O$ that is invariant under every generator action (see \cref{def:reach}). The overlap between $\g$ and $\mathcal V(O)$ is schematic: there is no general inclusion relation between them, and $\mathcal V(O)$ need not itself be closed under commutators. Rather, $\g$ acts on $\mathcal V(O)$ through the restricted adjoint maps.
  (c) If $D=\dim\mathcal V(O)<\infty$ and $\{B_\alpha\}_{\alpha=1}^{D}$ is a basis, these restricted actions are represented by matrices $A_k$ defined through $\mathrm{ad}_{H_k}(B_\alpha)=\sum_\beta(A_k)_{\alpha\beta}B_\beta$ as in \cref{eq:module-action-matrices}. Expanding $O=\sum_\alpha w_\alpha B_\alpha$ and defining $e_\alpha^{\mathrm{in}}=\Tr[B_\alpha\rho_{\mathrm{in}}]$ gives the exact finite-dimensional contraction in \cref{eq:module-eval}. 
  }
  \label{fig:summary}
\end{figure*}

In this work, we establish that exact bosonic Heisenberg simulation can be formulated directly on finite algebraic modules, identify the additional accessibility conditions required for an efficient simulation, and construct such modules for several nontrivial bosonic regimes. The reachable module need not contain the generators or be closed under commutators among its own elements, and it may be substantially smaller than the full dynamical Lie algebra (DLA). Conversely, a finite-dimensional generator algebra can act on an observable whose orbit is infinite-dimensional. Classical simulability is consequently not a property of Hamiltonians alone, but of the pair formed by the dynamics and the observable of interest. This viewpoint recovers Gaussian simulation as one instance of this principle, but it also exposes exact and polynomially tractable families with genuinely non-Gaussian Hamiltonian dynamics. Our main contributions are as follows.

\paragraph{Reachable-module simulability criterion.} 
We formulate exact Lie-algebraic simulation directly in terms of the reachable operator module. If this module has polynomial dimension and the required input overlaps are efficiently accessible, then the generators' actions can be constructed or applied efficiently, meaning that quantum mean values can be evaluated by finite-dimensional Heisenberg propagation. The same construction yields fixed-order correlators, multi-time observables, OTOCs, and reverse-mode gradients for machine learning. The conventional polynomial-dimensional DLA condition is recovered when the observable belongs to the generator algebra, but we show it is sufficient rather than necessary for forward simulation.

\paragraph{A unified construction of bosonic simulable families.} 
We identify three distinct mechanisms that produce polynomial reachable modules. First, quadratic dynamics preserves polynomial degree, yielding a hierarchy of finite moment modules. This embeds the usual displacement and covariance-matrix formalism as its first- and second-order sectors while extending it to arbitrary fixed-order correlators and non-Gaussian input states whenever the corresponding input moments are available. Second, adapting the $U(1)$ symmetry-sector strategy for bounded-Hamming-weight systems developed in~\cite{Barligea.2026-EnablingLieAlgebraicClassical}, we show that number-conserving bosonic dynamics are exactly simulable on any fixed sector, or any finite collection of sectors below a bounded maximum photon number. This includes interacting self-Kerr, cross-Kerr, and pair-hopping dynamics at arbitrary interaction strength and circuit depth, subject only to efficient access to the restricted generators. Third, motivated by finite-dimensional classifications of single-mode bosonic Lie algebras~\cite{Heib2025}, we construct a multi-mode nilpotent phase family containing cubic- and higher-order generators whose bounded-degree quadrature observables remain exactly tractable.

\paragraph{Controlled departure from exact sector confinement.} 
Number-changing Gaussian operations generally destroy exact sector confinement. We analyze this departure for squeezing, whose pair-creation and pair-annihilation terms couple photon-number sectors in parity-resolved steps. We show that weak squeezing can nevertheless be treated through parity-resolved photon-number bands that enlarge systematically with perturbative order. Under explicit analyticity assumptions, a band of depth $k$ reproduces the full dynamics through order $2k+1$ for number-conserving observables, such that the leading truncation error appears only at order $2(k+1)$ in the squeezing strength. Each enlargement of the band therefore suppresses the error by two additional powers of the squeezing amplitude. For standard squeezed vacua, this asymptotic control can further be sharpened using explicit photon-number-tail bounds. The resulting hierarchy provides a route between exact sector simulation and number-changing Gaussian dynamics, with a systematically improvable error governed by the retained band depth.

\paragraph{A sharp boundary of the method.} 
We demonstrate that a finite generator algebra alone does not guarantee simulability by constructing finite-dimensional non-Gaussian algebras with infinite observable orbits. In the opposite direction, suitable cubic-or-higher generators combined with Gaussian operations reach continuous-variable universality and generate modules of unbounded degree for generic observables~\cite{Lloyd1999,Calcluth2024}. The resulting taxonomy separates Gaussian degree preservation, bounded-photon sector confinement, and nilpotent phase closure from universal or otherwise unbounded non-Gaussian dynamics. It also makes explicit that the decisive resource is closure of the relevant observable dynamics, independently of whether the dynamics are Gaussian.

\paragraph{Numerical realization of the theory.} 
Our numerical studies include coherent and incoherent mixtures of bounded photon sectors, interacting Kerr and Bose--Hubbard dynamics, repulsively bound photon pairs, exact operator spreading through OTOCs, nilpotent cubic-phase propagation, and the predicted convergence orders of the squeezing-band approximation. Independent small-system and finite-sector calculations validate the corresponding propagations, while additional results illustrate differentiable control and symmetry-resolved interacting dynamics.

We emphasize that the outputs produced by our method are exact mean values and fixed-order correlators, not samples from the complete output distribution. This distinction is essential since the order of the considered observable is itself a computational resource. Fixed-order moment modules can remain polynomial, whereas the high-order coincidence observables encoding boson-sampling probabilities need not. Our results are therefore fully consistent with boson-sampling hardness~\cite{Aaronson2013,Hamilton2017}. Rather, they identify a complementary regime in which physically informative observables remain exactly accessible even when the input state or the Hamiltonian dynamics is non-Gaussian.

The remainder of this work is structured as follows.
\cref{sec:relwork} positions the mean-value problem within the literature on Gaussian simulation, non-Gaussian resources, sampling, and nonlinear bosonic dynamics. \cref{sec:finiteopertormodules} introduces the reachable-module criterion and its extension to correlators and gradients. \cref{sec:heisenbergweyl} specializes the construction to the Heisenberg--Weyl and Gaussian algebras, while \cref{sec:nongaussian} develops exact simulation at bounded photon number and its controlled extension under squeezing. \cref{sec:boundaries} establishes the limits of the criterion and constructs the nilpotent cubic-phase family. \cref{sec:numerics} presents the numerical demonstrations. Supporting proofs, gradient recursions, perturbative details, and extended numerical validations are collected in the appendices and referenced from the corresponding main-text sections.

\section{Related work}
\label{sec:relwork}

The consensus in the literature on classical simulation of bosonic systems is grounded in Gaussian quantum information, in which Gaussian states, dynamics, and measurements are fully described by first and second moments. These evolve through symplectic covariance-matrix methods~\cite{Bartlett2002,Wang.2007-QIGaussian,Weedbrook2012}. The boundary of this free theory is sharp, and adding a suitable non-Gaussian operation, such as a cubic phase gate, Kerr interaction, or a Gottesman--Kitaev--Preskill (GKP) resource state,  yields universal continuous-variable quantum computation \cite{Lloyd1999,Gottesman2001,Calcluth2024}. Beyond the Gaussian setting, existing simulation methods are naturally organized by two distinctions. First, we have the question of where non-Gaussianity enters, either through states, measurements or dynamics. Second, we have the question of whether the given problem is based on sampling, outcome probability estimation, or computing mean-values.

One line of work keeps the propagated dynamics Gaussian and controls the non-Gaussianity of states, measurements, or phase-space representations. Positive-quasiprobability methods give efficient weak simulation whenever the relevant states, processes, and measurements admit nonnegative representations~\cite{Mari2012,Veitch2012,Veitch2013,RahimiKeshari2016}. At the same time, negativity is not a complete obstruction, since structured GKP-type architectures can remain efficiently simulable despite large or unbounded Wigner negativity~\cite{GarciaAlvarez2020,Calcluth2022,Calcluth2023}. A complementary family of exact and approximate strong simulators decomposes non-Gaussian input states into Gaussian or coherent branches and propagates each branch through Gaussian dynamics, tracking their relative phases~\cite{Bourassa2021,Marshall2023,Dias.2024-ClassicalSimulationNongaussiana,Hahn2025}, with a fermionic analogue in~\cite{Dias.2024-ClassicalSimulationNonGaussian}. As in stabilizer-rank simulation for qubits~\cite{Bravyi2016,Bravyi.2019-SimulationQuantumCircuits}, the cost is then governed by resource measures such as Gaussian rank and extent~\cite{Hahn2025,Dias.2024-ClassicalSimulationNongaussiana}, stellar rank~\cite{Chabaud2020,Chabaud2021,Chabaud2023}, and symplectic rank \cite{Mele2026}. These methods are powerful, but their non-Gaussianity usually arises through state complexity, measurement complexity, or gate count, rather than by propagating a non-Gaussian Hamiltonian flow in a finite Heisenberg-picture operator space. 

Beyond the dynamics and inputs, it is also well-established that sampling from an output distribution, computing outcome probabilities, and evaluating expectation values are distinct notions of simulation~\cite{Bravyi.2021-ClassicalAlgorithmsQuantum,Pashayan2020}. Boson sampling hardness is located in high-order outcome probabilities, such as $n$-photon coincidences given by permanents~\cite{Aaronson2013}, as opposed to arbitrary fixed-order correlators. Conversely, even for Gaussian states, certain photon-number moments and cumulants are \#P-hard in general~\cite{Cardin2024}, showing that the observable class is part of the complexity landscape. Building on the qubit quantum mean-value problem~\cite{Bravyi.2021-ClassicalAlgorithmsQuantum}, recent work has isolated mean-value estimation as a problem for bosonic circuits~\cite{Lim.2025-ClassicalAlgorithmsEstimating}: Oh and Lim give additive-error algorithms for measurement-adaptive Gaussian circuits with non-Gaussian inputs~\cite{Oh2026}. We address the same simulation problem, however our method applies also to non-Gaussian dynamics as well as non-Gaussian inputs. 

Non-Gaussian bosonic dynamics have traditionally been treated by approximate numerical methods, including truncated Fock spaces \cite{Johansson.2012,Killoran.2019strawberryfields}, semiclassical phase-space methods~\cite{Polkovnikov2010,Deuar2002}, fixed-order cumulant truncations~\cite{Plankensteiner2022}, and tensor networks~\cite{Vinther2025}. 
Close in spirit to our non-Gaussian applications are two concurrent works on Gaussian-plus-Kerr circuits~\cite{Guseynov2026,upreti2026exponentiallyimproved}. The first of these works propagates coherent state superpositions  through displaced linear optics with Kerr gates, obtaining approximate quasi-polynomial costs for logarithmically many Kerr layers and polynomial costs at weak coupling \cite{Guseynov2026}. Next, the authors of Ref.~\cite{upreti2026exponentiallyimproved} compute outcome probabilities for circuits interleaving Gaussian layers with number-conserving non-Gaussian gates, identifying efficient regimes at logarithmic mode number or logarithmic depth, and exact descriptions for rational Kerr parameters. These Schr\"odinger-picture methods treat circuit families that include displacements and squeezing. Our framework instead evaluates expectation values, correlation functions, and gradients directly in the Heisenberg picture. At bounded photon number, it accommodates polynomial-depth number-conserving dynamics with unrestricted interaction strength and general sector-preserving nonlinearities, without reconstructing the evolved state or its output distribution.

On the passive side, very recent work~\cite{monbroussou2026classicalsimulationmodelconcentration} studies concentration and classical simulability by decomposing fixed-photon-number operator spaces into irreducible representations of $U(m)$. Fixed-degree number-preserving observables then occupy polynomial-dimensional low-irrep sectors, giving exact $\gsim$ when the required projections are efficiently accessible. Truncating higher irreps also yields average-case approximations over Haar-random interferometers. This provides a representation-theoretic account of the passive, fixed-order regime. Their analysis allows the photon number to grow but keeps the dynamics passive, whereas our reachable-module construction also covers active Gaussian transformations and non-Gaussian generators.

The role of photon number changes once the dynamics becomes interacting. 
The fixed-photon-number sector dimension $\binom{N+n-1}{N}$ is recognized as algorithmically important and as a route to universality on sectors~\cite{Childs.2013-UniversalComputationMultiparticle,Oszmaniec.2017,Gu.2021-FastforwardingQuantumEvolution}; here we use the same sector structure in the opposite direction, as an exact mean-value simulation mechanism when $N$ is bounded. This bounded-$N$ regime is complementary to the QMA-hard Bose--Hubbard setting in which the particle number grows with system size~\cite{Childs2014}.

\section{Classical simulability from finite reachable operator modules}
\label{sec:finiteopertormodules}
Lie-algebraic expectation-value simulation originates in the generalized coherent-state and finite-representation framework of Somma et al.~\cite{Somma.2005-QuantumComputationComplexity,Somma.2006-EfficientSolvabilityHamiltonians}. Goh et al.~\cite{Goh.2025-LiealgebraicClassicalSimulations} developed the practical adjoint-space $\gsim$ formulation, including propagation in invariant operator subspaces, product observables, and reverse-mode gradients. B\"arligea et al.~\cite{Barligea.2026-EnablingLieAlgebraicClassical} subsequently supplied symmetry- and subspace-adapted bases and efficient preprocessing primitives beyond the previously dominant free-fermionic setting. We first recall these ingredients before specializing the invariant space to the minimal module generated by the requested observable.

\subsection{Standard adjoint-space Lie-algebraic simulation}
Before turning to bosons, we recall the standard finite-dimensional adjoint-space construction familiar from qubit and spin systems, using a Hermitian operator convention throughout.
Let $H_1,\ldots,H_K$ be Hermitian generators and consider circuits of the form
\begin{equation}
U_\ell=e^{-\ii\theta_\ell H_{k_\ell}},
\quad
U(\bm\theta)=U_L\cdots U_1,
\label{eq:circuit-unitary}
\end{equation}
so that $U_1$ acts first on the input state. We now seek exact values for observables' mean values,
\begin{equation}
\begin{aligned}
  \langle O(\bm\theta)\rangle
  &=
  \Tr[
    O\,U(\bm\theta)\rho_{\mathrm{in}}U(\bm\theta)^\dagger
  ]\\
  &=
  \Tr[
    U(\bm\theta)^\dagger O U(\bm\theta)\rho_{\mathrm{in}}
  ].
\end{aligned}
  \label{eq:mean-value-task}
\end{equation}
For a Hermitian generator $H$, the adjoint action is the linear map that sends an operator $A$ to its commutator with $H$. We use the Hermitian convention
\begin{equation}
\mathrm{ad}_H(A):=\ii[H,A],
\label{eq:hermitian-ad}
\end{equation}
so that $\mathrm{ad}_H(A)$ remains Hermitian whenever $H$ and $A$ are Hermitian. Multiplication by $\ii$ identifies this convention with the usual skew-Hermitian formulation generated by $\ii H$.
In this convention, the dynamical Lie algebra (DLA) is obtained by closing the generators under repeated adjoint actions,
\begin{equation}
\begin{aligned}
  \g
  =
  \spann_{\mathbb R}
  \{
    & \mathrm{ad}_{H_{\alpha_1}}\cdots
    \mathrm{ad}_{H_{\alpha_\ell}}(H_\beta)
    : \\
    & \ell\ge 0,\;
    \alpha_1,\ldots,\alpha_\ell,\beta\in\{1,\ldots,K\}
  \}.
\end{aligned}
  \label{eq:DLA}
\end{equation}
Here $\ell$ is the depth of the nested commutator. In a finite-dimensional operator space, the span stabilizes after finitely many steps; if $d=\dim\mathfrak g$, words with $\ell\leq d-1$ suffice. For example, the one-qubit generators $H_1=\tfrac{X}{2}$ and $H_2=\tfrac{Z}{2}$ generate $\spann_{\mathbb R}\{X,Y,Z\}$ because $\ii[H_1,H_2]=\tfrac{Y}{2}$, and further commutators produce no additional directions. Thus the usual DLA is already a reachability construction, since it is the space reached from the generator set under repeated commutation with the generators. 

The original framework assumed compact semisimple Lie algebras to replace the physical Hilbert-space representation by finite algebraic data~\cite{Somma.2005-QuantumComputationComplexity,Somma.2006-EfficientSolvabilityHamiltonians}. Semisimplicity guarantees a faithful adjoint representation and provides the highest-weight structure used to encode input states efficiently, while compactness supplies a positive-definite invariant inner product. These assumptions, however, are stronger than what the propagation below requires, which is merely a finite-dimensional invariant operator space with computable generator actions and input overlaps.

For this finite adjoint propagation, let $\{B_\alpha\}_{\alpha=1}^d$ be a Hilbert--Schmidt orthonormal Hermitian basis of $\g$, where ${d=\dim\g}$, with inner product ${\langle B_\alpha,B_\beta\rangle=\Tr[B_\alpha^\dagger B_\beta]}$. More generally, the propagation formulas require only a real basis of the relevant invariant operator space; Hilbert--Schmidt orthonormality is used only to extract coefficients. The infinitesimal adjoint action of a basis element $B_\gamma$ is represented by the real $d\times d$ matrix $\adphi(B_\gamma)$ defined by
\begin{equation}
  \mathrm{ad}_{B_\gamma}(B_\alpha)
  =
  \ii[B_\gamma,B_\alpha]
  =
  \sum_{\beta=1}^d
  \left(\adphi(B_\gamma)\right)_{\alpha\beta}B_\beta .
  \label{eq:adphi-definition}
\end{equation}
Because the basis is Hilbert--Schmidt orthonormal, the expansion coefficients of the adjoint action are
\begin{equation}
f_{\alpha\beta}^{\ \ \gamma}:=\left\langle B_\gamma,\ii[B_\alpha,B_\beta]\right\rangle,\qquad \left(\adphi(B_\gamma)\right)_{\alpha\beta}=f_{\gamma\alpha}^{\ \ \beta}.
\label{eq:struct}
\end{equation}
The coefficients $f_{\alpha\beta}^{\ \ \gamma}$ are the \emph{structure constants} in this Hermitian convention.
For a general Hamiltonian $H=\sum_\gamma h_\gamma B_\gamma$, we may therefore write,
\begin{equation}
  \adphi(H)
  =
  \sum_\gamma h_\gamma\: \adphi(B_\gamma).
  \label{eq:adphi-linear}
\end{equation}
The finite Heisenberg (adjoint) action is then obtained by exponentiation,
\begin{equation}
  e^{\ii\theta H}B_\alpha e^{-\ii\theta H}
  =
  \sum_{\beta=1}^d
  \left(e^{\theta\:\adphi(H)}\right)_{\alpha\beta}B_\beta.
  \label{eq:finite-adjoint-one-layer}
\end{equation}
Thus $\adphi(H)$ and its exponential are $d\times d$ matrices, where $d=\dim\g$. The adjoint propagation is independent of the input state; $\rho_{\mathrm{in}}$ enters only through the $d$ input overlaps defined below, whose evaluation cost depends on how the state is represented.

Returning to \cref{eq:mean-value-task}, if the observable is expanded as $O=\sum_\alpha w_\alpha B_\alpha$ with $\bm w^{\top}=(w_1,\dots,w_d)$, and assuming that the input overlaps
\begin{equation}
e_\alpha^{\mathrm{in}}
=
\Tr[B_\alpha\rho_{\mathrm{in}}]
\end{equation}
exist and are efficiently available (i.e., can be computed classically in polynomial time), then the expectation value~\eqref{eq:mean-value-task} is the contraction of two $d$-dimensional objects through the adjoint propagator:
\begin{equation}
\begin{aligned}
  \langle O(\bm\theta)\rangle
  &=
  \Tr[
    U(\bm\theta)^\dagger O U(\bm\theta)\rho_{\mathrm{in}}
  ]
  \\
  &=
  \bm w^{\top}
  \Adphi(U(\bm\theta))
  \bm e^{\mathrm{in}}
  \\
  &=
\bm w^{\top}
e^{\theta_L\adphi(H_{k_L})}\cdots
e^{\theta_1\adphi(H_{k_1})}
\bm e^{\mathrm{in}}.
\end{aligned}
\label{eq:eval}
\end{equation}
This is the basic $\gsim$ evaluation formula~\cite{Goh.2025-LiealgebraicClassicalSimulations}. It gives an efficient exact algorithm when $d=\poly(n)$ and the basis admits efficient construction or application of the adjoint matrices, observable coordinates, and input overlaps~\cite{Barligea.2026-EnablingLieAlgebraicClassical}.

\subsection{Observable-seeded reachable modules}
For the forward mean-value problem, propagating the full DLA is stronger than necessary. Indeed, the same reachability operation can be seeded by the observable itself in lieu of the generators. This distinction is mostly hidden in finite systems, but it becomes essential for bosons. The Hilbert space is infinite-dimensional and the usual compactness and finite-representation assumptions may fail, yet the linear operator space generated from a particular observable by repeated commutation with the circuit generators can still be finite-dimensional.

\begin{definition}[Reachable operator module]
\label{def:reach}
Let $H_1,\ldots,H_K$ be Hermitian generators and let $O$ be an observable. The
reachable operator module of $O$ is
\begin{equation}
\begin{aligned}
  \mathcal V(O)
  :=
  \spann_{\R}
  \{
    & \mathrm{ad}_{H_{k_1}}\cdots\mathrm{ad}_{H_{k_\ell}}(O)
    :\\
    & \ell\ge 0,\;
    k_1,\ldots,k_\ell\in\{1,\dots,K\}\}.
\end{aligned}
  \label{eq:reachable-module}
\end{equation}
Equivalently, $\mathcal V(O)$ is the smallest real operator space containing $O$ and invariant under every $\mathrm{ad}_{H_k}$. It is a module for the DLA generated by the $H_k$, but it need not itself be a Lie algebra, as closure under commutators inside $\mathcal V(O)$ is not required.
\end{definition}

We review the module terminology, its distinction from Lie-algebra and associative-algebra closure, and its role in adjoint propagation in Appendix~\ref{app:module-background}.
Related notions of reachable operator spaces are long established in geometric control theory~\cite{Jurdjevic.1972-ControlSystemsLie}, in quantum control~\cite{Schirmer.2001-CompleteControllabilityQuantum,Albertini.2003-NotionsControllabilityBilinear,DAlessandro.2021-IntroductionQuantumControl}, in symmetric systems theory~\cite{Zeier.2011-Symmetryprinciplesquantumsystems,Allcock.2024-DynamicalLieAlgebras,Kazi.2025-AnalyzingQuantumApproximate,Diaz.2023-ShowcasingBarrenPlateau,Gargiulo.2026-PauliStringsQuantumUnified}, in the commutant formalism of Refs.~\cite{Moudgalya.2022-HilbertSpaceFragmentation,Moudgalya.2023-SymmetriesCommutantAlgebras,Lastres.2026-NonuniversalityConservedSuperoperators}, and on the simulation side~\cite{Somma.2005-QuantumComputationComplexity,Somma.2006-EfficientSolvabilityHamiltonians,Kokcu.2022-FixedDepthHamiltonian,Smith.2025-OptimallyGeneratingSu,Goh.2025-LiealgebraicClassicalSimulations,Anschuetz.2023-EfficientClassicalAlgorithms}. None of these works, however, isolate the reachable operator module $\mathcal V(O)$ generated from a fixed observable under nested adjoint action, nor tie it to a criterion for exact classical simulability of bounded-photon bosonic mean values.

For every \(A\in\mathcal{V}(O)\), assume that the circuit
unitaries and conjugated basis operators share a common
invariant domain on which
\begin{equation}
    e^{i\theta H_k} A e^{-i\theta H_k}
    =
    e^{\theta\operatorname{ad}_{H_k}}(A)
\end{equation}
holds. Suppose that
$
D:=\dim_{\mathbb{R}}\mathcal{V}(O)<\infty
$,
choose a basis $\{B_\alpha\}_{\alpha=1}^{D}$, and define
the matrices $A_k$ by
\begin{equation}
    \operatorname{ad}_{H_k}(B_\alpha)
    =
    \sum_{\beta=1}^{D}(A_k)_{\alpha\beta}B_\beta.
\label{eq:module-action-matrices}
\end{equation}
Let $\bm w$ and $\bm e^{\mathrm{in}}$ denote the observable-coordinate and input-overlap vectors in this basis, defined as in the standard construction above. We then obtain the following theorem.

\begin{theorem}[Reachable-module simulation criterion]
Every circuit of the form~\cref{eq:circuit-unitary} preserves \(\mathcal{V}(O)\), and
its expectation value satisfies
\begin{equation}
    \langle O(\boldsymbol{\theta})\rangle
    =
    w^\top
    e^{\theta_L A_{k_L}}
    \cdots
    e^{\theta_1 A_{k_1}}
    e^{\mathrm{in}} .
\label{eq:module-eval}
\end{equation}
If \(D=\operatorname{poly}(n)\), \(L=\operatorname{poly}(n)\),
and the basis actions, observable coordinates, input overlaps,
and matrix propagators are accessible in polynomial time, this
formula gives exact polynomial-time expectation values and
reverse-mode gradients.
\label{thm:reach}
\end{theorem}

\begin{proof}
See Appendix~\ref{prf:reach}
\end{proof}

The reduction replaces Hilbert-space propagation with exact
finite-dimensional linear algebra. The matrix
operations may still be evaluated to a prescribed numerical
precision, which supplies the only approximation in the
calculation; otherwise the calculation is free of any Fock-space cutoff, moment closure or perturbative truncation.

Reverse-mode gradients follow by differentiating the same product of finite-dimensional matrices and accumulating the derivatives in reverse order~\cite{Goh.2025-LiealgebraicClassicalSimulations}. A breakdown of the complete construction is captured in \cref{fig:summary}, from the mean-value simulation problem in \cref{eq:mean-value-task} through the observable-seeded closure in \cref{eq:reachable-module} to the finite-dimensional evaluation in \cref{eq:module-eval}.

Next, we turn our attention to fixed-order correlation functions. To that end, for fixed \(m=O(1)\), let
\[
    O_j^{(j)}
    :=
    \bigl(U^{(j)}\bigr)^\dagger O_j U^{(j)},
    \qquad j=1,\ldots,m,
\]
where the \(U^{(j)}\) may denote different circuits or circuit
prefixes. Suppose that \(O_j^{(j)}\) lies in a reachable module
\(\mathcal{V}_j\) with basis
\(\{B_{\alpha_j}^{(j)}\}_{\alpha_j=1}^{D_j}\), and expand in the basis,
\begin{equation}
    O_j^{(j)}
    =
    \sum_{\alpha_j}
    \widetilde w_{\alpha_j}^{(j)}
    B_{\alpha_j}^{(j)} .
\label{eq:reach-mod-expansion}
\end{equation}
This allows us to define the ordered input moments,
\[
    E^{\mathrm{in}}_{\alpha_1\cdots\alpha_m}
    :=
    \operatorname{Tr}\!\left[
        B_{\alpha_1}^{(1)}
        \cdots
        B_{\alpha_m}^{(m)}
        \rho_{\mathrm{in}}
    \right],
\]
for which we arrive at the following corollary.

\begin{corollary}[Fixed-order correlation functions]
\label{cor:module-correlators}
The ordered correlation function satisfies
\begin{equation}
    \left\langle
        O_1^{(1)}\cdots O_m^{(m)}
    \right\rangle
    =
    \sum_{\alpha_1,\ldots,\alpha_m}
    \left(
        \prod_{j=1}^{m}
        \widetilde w_{\alpha_j}^{(j)}
    \right)
    E^{\mathrm{in}}_{\alpha_1\cdots\alpha_m}.
\end{equation}
For fixed \(m\), this contraction is efficient whenever the
individual module actions and the required input-moment
contraction are efficiently accessible.
\end{corollary}

\begin{proof}
See Appendix~\ref{prf:module-correlators}.
\end{proof}

The same construction covers equal-time, unequal-time, and out-of-time-ordered products, as well as reverse-mode gradient recursion for automatic differentiation, which we detail in Appendix~\ref{app:module-products-gradients}. The resulting reverse-mode derivatives are validated against finite differences and used in a statistically replicated nonlinear two-photon control problem over depths $1$--$8$ and lattices as large as $9\times9$ in Appendix~\ref{app:num-gradients}, where success fractions distinguish ansatz expressivity from optimization robustness in the shallow-depth regime.

We see that the standard DLA criterion is recovered as the special case in which the observable is seeded \textit{inside} the generator algebra. Indeed, if $O\in\g$, then $\mathcal V(O)\subseteq\g$, so adjoint propagation in the full DLA is a valid choice. The converse need not hold, however, and $\mathcal V(O)$ can be much smaller than $\g$, meaning it may be unnecessary to represent the full algebra faithfully. For the forward simulation problem considered here, only the induced action on the reachable module and the input overlaps are used. 

This observation is the point at which bosons depart from the standard finite spin intuition. First, central directions and affine terms can be included in the module basis. More importantly, we can propagate the orbit of the observable instead of the full generator algebra which was previously done in spin and fermionic systems. For Gaussian bosonic dynamics and a linear observable $O\in\spann_{\mathbb R}\{\mathds{1},\hat x_1,\hat p_1,\ldots,\hat x_n,\hat p_n\}$, the reachable module has dimension at most $2n+1$, even though the ambient Gaussian algebra has dimension $\mathcal O(n^2)$. This shows that the DLA dimension can substantially overestimate the operator space required by a particular problem. Standard Gaussian optics already realizes this reduction through first-moment propagation; the significance here is that it follows directly from the reachable-module criterion and extends uniformly to the higher-order and non-Gaussian settings developed below.

\section{Heisenberg--Weyl algebra and Gaussian optics}
\label{sec:heisenbergweyl}

We now apply the reachable-module formalism of \cref{sec:finiteopertormodules} to bosonic modes. Gaussian dynamics provides the first example of the module criterion. Although the bosonic Fock space is infinite-dimensional, quadratic Hamiltonians preserve finite-dimensional spaces of observables of fixed polynomial degree. 
Throughout this section, we use the Hermitian operator convention of~\cref{eq:hermitian-ad}.

\subsection{Gaussian algebra and symplectic propagation}
\label{ssec:Gaussianpropagation}
A bosonic system of $n$ modes is described by creation and annihilation operators $\hat a_k,\hat a_k^\dagger$, with ${k=1,\dots,n}$, satisfying the canonical commutation relations (CCR),
\begin{equation}
  [\hat a_k,\hat a_l^\dagger]=\delta_{kl}\,\mathds{1},
  \qquad
  [\hat a_k,\hat a_l]=[\hat a_k^\dagger,\hat a_l^\dagger]=0 .
  \label{eq:ccr}
\end{equation}
Equivalently, in terms of quadratures
\begin{equation}
\hat x_k=(\hat a_k+\hat a_k^\dagger)/\sqrt2,
\qquad
\hat p_k=(\hat a_k-\hat a_k^\dagger)/(\sqrt2\,\ii),
\label{eq:bquadratures}
\end{equation}
one has $[\hat x_k,\hat p_l]=\ii\,\delta_{kl}\mathds{1}$. We collect the quadratures into the Hermitian phase-space vector
\begin{equation}
\Rvec=(\hat x_1,\hat p_1,\ldots,\hat x_n,\hat p_n)^{\top}.
\end{equation}
Its components are ordered as $\Rvec_{2k-1}=\hat x_k$ and $\Rvec_{2k}=\hat p_k$ for $k=1,\ldots,n$.
In these conventions, the CCR take the compact matrix form
\begin{equation}
[\Rvec_i,\Rvec_j]=\ii\Omega_{ij}\mathds{1},\qquad \Omega=\bigoplus_{k=1}^n\begin{pmatrix}0&1\\-1&0\end{pmatrix}.
\label{eq:Omega}
\end{equation}
Under the Hermitian bracket of \cref{eq:hermitian-ad},
$\mathrm{ad}_{\Rvec_i}(\Rvec_j)
=
-\Omega_{ij}\mathds{1}$,
so the nonzero structure coefficients of the linear Heisenberg--Weyl basis are read directly from $-\Omega$ in the index convention of \cref{eq:struct}. They are therefore available in closed form, with only one nonzero entry in each noncentral row. All higher commutators used below follow from \cref{eq:Omega} and the Leibniz rule.

The identity and linear quadratures therefore close on the Heisenberg--Weyl algebra,
\begin{equation}
  \h_n:=\spann_{\R}\{\, \mathds{1},\ \hat x_1,\dots,\hat x_n,\
        \hat p_1,\dots,\hat p_n\,\},
\label{eq:HWalgebra}
\end{equation}
with $\dim\h_n=2n+1$. By \cref{eq:Omega}, the Hermitian bracket of any two linear quadratures is proportional to the identity. Since the identity commutes with every operator, applying one further bracket gives zero. Hence $\h_n$ is two-step nilpotent.

Linear Hamiltonians in $\h_n$ generate displacements, and to obtain the usual Gaussian transformations while remaining finite-dimensional, we add quadratic Hamiltonians.
To describe quadratic Hamiltonians, let $\operatorname{Sym}$ denote Weyl symmetrization, namely the average over all orderings of its factors; in particular, $\operatorname{Sym}(\Rvec_i\Rvec_j)=\tfrac{1}{2}(\Rvec_i\Rvec_j+\Rvec_j\Rvec_i)$. The Hermitian quadratic span is
\begin{equation}
\mathfrak q_n:=\spann_{\mathbb R}\left\{\operatorname{Sym}(\Rvec_i\Rvec_j):1\leq i\leq j\leq2n\right\}.
\label{eq:quadraticspan}
\end{equation}

The Leibniz identity $[AB,C]=A[B,C]+[A,C]B$ reduces commutators between two such quadratic polynomials to the CCR. Since every elementary quadrature commutator is proportional to the identity, the bracket of a quadratic operator with a linear operator is linear, while the bracket of two quadratic operators remains quadratic. Thus $[\mathfrak q_n,\h_n]\subseteq\h_n$ and $[\mathfrak q_n,\mathfrak q_n]\subseteq\mathfrak q_n$.

The quadratic operators form the Lie algebra $\mathfrak q_n$, which through its action on $\Rvec$ is isomorphic to the real symplectic algebra $\mathfrak{sp}(2n,\mathbb R)$. Together with the Heisenberg--Weyl algebra $\h_n$, they form the finite-dimensional ambient algebra of Gaussian bosonic dynamics~\cite{Bartlett2002,Wang.2007-QIGaussian,Weedbrook2012},
\begin{equation}
\g_{\mathrm{Gauss}}=\mathfrak q_n\ltimes\h_n\cong\mathfrak{sp}(2n,\mathbb R)\ltimes\h_n,
\label{eq:gGauss}
\end{equation}
with $\dim\g_{\mathrm{Gauss}}=(n+1)(2n+1)$. The notation $\ltimes$ denotes the semidirect sum in which $\h_n$ is an ideal of $\g_{\mathrm{Gauss}}$ and $\mathfrak q_n$ acts on $\h_n$ by commutation. For example, $\ii[\hat x_k^2,\hat p_k]=-2\hat x_k$. Notice that the symplectic factor $\mathfrak q_n\cong\mathfrak{sp}(2n,\mathbb R)$ is simple but non-compact, whereas the full Gaussian algebra $\g_{\mathrm{Gauss}}$ is not semisimple because $\h_n$ is a nonzero nilpotent ideal.

A concrete Gaussian circuit may generate only a subalgebra of $\g_{\mathrm{Gauss}}$, but \cref{eq:gGauss} is the ambient algebra for arbitrary at-most-quadratic bosonic generators. 
To recover the corresponding action on linear observables, we can express an at-most-quadratic Hamiltonian as
\begin{equation}
  H=\frac12\,\Rvec^{\!\top}G\,\Rvec-\bm c^{\top}\Om\,\Rvec,
  \label{eq:Hquad}
\end{equation}
with ${G=G^{\top}\in\R^{2n\times 2n}}$ and ${\bm c\in\R^{2n}}$. 
For the time-independent Hamiltonian in \cref{eq:Hquad}, the Heisenberg equation is $\frac{\dd}{\dd t}\Rvec(t)=\ii[H,\Rvec(t)]$. Evaluating this commutator with the CCR and the Leibniz identity gives
\begin{equation}
  \frac{\dd}{\dd t}\Rvec(t)
  =
  \Om G\,\Rvec(t)-\bm c,
  \quad
  \Rvec(t)=S(t)\Rvec(0)-\bm d(t),
  \label{eq:heom}
\end{equation}
where
\begin{equation}
  S(t)=e^{t\Om G}\in\Sp(2n,\R),
  \quad
  \bm d(t)=\int_0^t e^{s\Om G}\bm c\,\dd s .
  \label{eq:symplectic}
\end{equation}
Since $G$ is symmetric, ${(\Om G)^{\top}\Om+\Om(\Om G)=0}$, and hence
$S(t)\in\Sp(2n,\R)$. Thus, the adjoint action of Gaussian dynamics on the linear module is the familiar affine symplectic transformation.

Equivalently, one may augment the quadrature vector by the central coordinate,
\begin{equation}
\begin{gathered}
  \widetilde{\Rvec}
  :=
  \begin{pmatrix}
    \mathds{1}\\
    \Rvec
  \end{pmatrix},
  \\
  \widetilde{\Rvec}(t)
  =
  \widetilde S(t)\widetilde{\Rvec}(0),
  \quad
  \widetilde S(t)
  =
  \begin{pmatrix}
    1&0\\
    -\bm d(t)&S(t)
  \end{pmatrix},
\end{gathered}
  \label{eq:augmentedS}
\end{equation}
where $S(t)$ and the displacement vector $\bm d(t)$ are defined in
\cref{eq:symplectic}.
Thus affine displacements remain within the finite-dimensional module because the inhomogeneous term is represented by the identity operator, which spans the center of $\h_n$.

To express the homogeneous part of the Gaussian evolution in the standard Bogoliubov form, define the creation--annihilation vector $\rvec=(\hat a_1,\ldots,\hat a_n,\hat a_1^\dagger,\ldots,\hat a_n^\dagger)^\top$. Let $C$ be the fixed change-of-basis matrix satisfying $\rvec=C\Rvec$. The corresponding transfer matrix is~\cite{Weedbrook2012}
\begin{equation}
C S(t) C^{-1}
=
\begin{pmatrix}
U(t)&V(t)\\
V(t)^*&U(t)^*
\end{pmatrix}.
\label{eq:bogoliubov}
\end{equation}
The quadrature and creation--annihilation representations therefore describe the same Gaussian transformation in different operator bases. Standard phase shifts, squeezers, two-mode squeezers, and beamsplitters are collected in Appendix~\ref{app:gaussian-gates}.

The important point for $\gsim$ is that none of these constructions relies on the infinite-dimensional bosonic Hilbert space. Linear observables evolve in a module of dimension $2n+1$, while the ambient Gaussian algebra has dimension $\mathcal O(n^2)$. The relevant finite object is therefore the reachable operator module, not the infinite bosonic Hilbert space and not necessarily the full ambient algebra.

\subsection{Fixed-degree moment modules}
\label{ssec:moment-modules}

The preceding affine transformation maps each quadrature to a linear combination of the identity and the quadratures. Consequently, a product of at most $m$ quadratures remains a polynomial of degree at most $m$, providing a closed operator space for higher-order expectation values and correlation functions.

To that end, let $O$ be a quadrature polynomial of degree at most $m$.
Using the augmented vector $\widetilde R$, we may write
\begin{equation}
    O
    =
    \sum_{\alpha_1,\ldots,\alpha_m}
    w_{\alpha_1\cdots\alpha_m}
    \widetilde R_{\alpha_1}
    \cdots
    \widetilde R_{\alpha_m},
\end{equation}
where identity components encode terms of degree below \(m\).
If we denote the corresponding input moment tensor by
\begin{equation}
    E^{\mathrm{in}}_{\alpha_1\cdots\alpha_m}
    :=
    \operatorname{Tr}\!\left[
        \widetilde R_{\alpha_1}
        \cdots
        \widetilde R_{\alpha_m}
        \rho_{\mathrm{in}}
    \right],
    \label{eq:input-moment-tensor}
\end{equation}
then we arrive at the following proposition~\footnote{For passive linear optics, this fixed-degree restriction also admits a representation-theoretic description via degree-$d$ number-preserving observables  that occupy at most the first $d+1$ irreducible components of the fixed-photon-number operator space~\cite{monbroussou2026classicalsimulationmodelconcentration}.}.

\begin{proposition}[Gaussian moment modules]
\label{prop:moment-tower}
Under quadratic dynamics,
\begin{equation}
    \langle O(t)\rangle
    =
    \sum_{\alpha_1,\ldots,\alpha_m}
    \sum_{\beta_1,\ldots,\beta_m}
    w_{\alpha_1\cdots\alpha_m}
    \left(
        \prod_{\ell=1}^{m}
        \widetilde S_{\alpha_\ell\beta_\ell}(t)
    \right)
    E^{\mathrm{in}}_{\beta_1\cdots\beta_m}.
\label{eq:moment-tower}
\end{equation}
Moreover, quadratic dynamics preserves the real polynomial
space
\begin{equation}
    \mathcal{P}_{\leq m}
    :=
    \operatorname{span}_{\mathbb{R}}
    \left\{
        \operatorname{Sym}(R^\alpha):|\alpha|\leq m
    \right\},
\end{equation}
and hence the dimension of the corresponding reachable operator module $\mathcal{V}(O)$ satisfies
\begin{equation}
    \dim_{\mathbb{R}}\mathcal{V}(O)
    \leq
    \dim_{\mathbb{R}}\mathcal{P}_{\leq m}
    =
    \binom{2n+m}{m}.
\end{equation}
\end{proposition}

\begin{proof}
See Appendix~\ref{prf:moment-tower}.
\end{proof}

Thus the unreduced contraction uses at most \((2n+1)^m\) tensor
entries. For \(m=1\), the proposition reproduces first-moment
propagation; for \(m=2\), it reproduces covariance propagation.
Higher-order input moments then encode non-Gaussian input states.

Indeed covariance matrices and Gaussian phase-space propagation follow as a corollary.  
For a state \(\rho\) with finite second moments, we may write
\begin{equation}
    \overline R_i
    :=
    \operatorname{Tr}[R_i\rho],
    \quad
    \sigma_{ij}
    :=
    \frac{1}{2}
    \operatorname{Tr}\!\left[
        \{R_i-\overline R_i,R_j-\overline R_j\}\rho
    \right],
\end{equation}
recovering the standard phase-space propagation equations.

\begin{corollary}[First and second moments]
\label{cor:1-2-moments}
Quadratic evolution propagates the mean and covariance as
\begin{equation}
    \overline{\boldsymbol R}(t)
    =
    S(t)\overline{\boldsymbol R}(0)-d(t),
    \quad
    \sigma(t)
    =
    S(t)\sigma(0)S(t)^\top .
\label{eq:1-2-moments}
\end{equation}
These identities hold without assuming that the input state is
Gaussian.
\end{corollary}

\begin{proof}
See Appendix~\ref{prf:1-2-moments}.
\end{proof}

Thus the module construction reproduces the conventional covariance-matrix method for
Gaussian inputs, and the same propagation framework can now naturally extend to fixed higher-order correlators, non-Gaussian inputs,
unequal-time correlators, and reverse-mode gradients. Appendix~\ref{app:num-gaussian} validates this propagation against analytic and
truncated-Fock references.

Given this result, we emphasize that there remains a clear boundary with boson sampling routines.  Under passive linear optics, every fixed-order photon-number correlator is contained in a fixed-degree moment module and is therefore covered by \cref{prop:moment-tower}. By contrast, collision-free transition probabilities between $N$-photon number states are squared permanents of the corresponding $N\times N$ interferometer submatrices~\cite{Aaronson2013}, with hardness arising when $N$ grows with system size. Similarly, evaluating photon-number moments and cumulants of Gaussian states becomes \#P-hard when their order is allowed to grow~\cite{Cardin2024}. Our construction therefore computes fixed-order expectation values and correlation functions, not complete output distributions. A concrete second-order formula, including the canonical example of the Hong--Ou--Mandel effect, is given in Appendix~\ref{app:bounded-order-boundary}.

The Gaussian family is the first bosonic instance of the reachable-module criterion. Here, degree truncation keeps $\mathcal V(O)$ polynomial for fixed observable degree. It does not, however, make generic non-Gaussian dynamics finite on the full Fock space. A Kerr Hamiltonian term, for example, has a small generator algebra, but the orbit of $\hat a$ under its adjoint action is infinite on the full bosonic Hilbert space. The next section therefore turns to a different finiteness mechanism, sector confinement, where number-conserving non-Gaussian Hamiltonians become finite matrices on fixed total-photon sectors. Symmetry-adapted refinements of the Gaussian module, which reduce the same adjoint map further under translation or permutation symmetry, are recorded in Appendix~\ref{app:gaussian-symmetry}.

\section{Exact bounded-photon dynamics and controlled sector leakage}
\label{sec:nongaussian}
In \cref{sec:heisenbergweyl}, Gaussian generators do not increase quadrature degree under commutation. Consequently, polynomials of degree at most $m$ form an invariant moment module for every fixed $m$.
We now introduce a mechanism that leverages confinement to a bounded set of photon-number sectors. This restriction enables non-Gaussian number-conserving interactions, including self-Kerr, cross-Kerr, and pair-hopping terms, while keeping the relevant photon-number subspace and reachable operator module finite. The construction is the bosonic analog of bounded-Hamming-weight simulation~\cite{Barligea.2026-EnablingLieAlgebraicClassical}. After establishing the exact result, we show how weak squeezing can be treated using systematically enlarged parity-resolved photon-number bands, with each increase in band depth suppressing the leading truncation error of number-conserving observables by two additional powers of the squeezing strength.

\subsection{Sector confinement and bounded-support simulation}

Let $\hat N=\sum_{k=1}^n \hat a_k^\dagger \hat a_k$ be the total photon-number operator. A Hamiltonian $H$ is number-conserving if $[H,\hat N]=0$. For the polynomial Hamiltonians considered here, this is equivalently the statement that, after normal ordering, every monomial contains the same number of creation and annihilation operators.

We consider finite linear combinations of the following number-conserving bosonic generators:
\begin{enumerate}[label=(\alph*),nosep]
  \item passive quadratic terms ${\sum_{k,l}h_{kl}\hat a_k^\dagger \hat a_l}$, with ${h=h^\dagger}$, so called because they transform annihilation operators only among themselves and do not create or annihilate photon pairs; these terms generate beamsplitters and phase shifters; 
  \item number-diagonal nonlinearities, including self-Kerr terms $\hat n_k(\hat n_k-1)$, cross-Kerr terms $\hat n_k\hat n_l$, and more generally polynomials $f(\hat n_1,\dots,\hat n_n)$ in the commuting number operators;
  \item interacting monomials of the form $\hat a_{k_1}^\dagger\cdots \hat a_{k_p}^\dagger \hat a_{l_1}\cdots \hat a_{l_p}$, such as pair hopping terms $\hat a_i^\dagger\hat a_j^\dagger\hat a_k\hat a_l$.
\end{enumerate}
The first family generates passive Gaussian dynamics, whereas the latter families include quartic and higher-order interaction terms outside $\g_{\mathrm{Gauss}}$.

Because $[H,\hat N]=0$, every such Hamiltonian preserves the eigenspaces of $\hat N$. Number conservation therefore decomposes the Fock space into invariant eigenspaces,
\begin{equation}
\begin{gathered}
  \mathcal H=\bigoplus_{N=0}^{\infty}\mathcal H_N,
  \\
  \mathcal H_N=\spann_{\mathbb C}\bigl\{\ket{\bm n}=\ket{n_1,\dots,n_n}:\sum_{k=1}^n n_k=N\bigr\}.
\end{gathered}
  \label{eq:sectordecomp}
\end{equation}

\begin{lemma}[Finite photon-number sectors]
\label{lem:finite-sector}
For every number-conserving Hamiltonian $\hat H$ and any $N\in\mathbb Z_{\ge0}$, the restriction $\hat H_N:=\hat H|_{\mathcal H_N}$ is a Hermitian matrix of size $d_{n,N}\times d_{n,N}$, where
\begin{equation}
  d_{n,N}=\dim\mathcal H_N=\binom{N+n-1}{N}.
  \label{eq:dimNsector}
\end{equation}
Thus, $e^{-\ii \hat Ht}$ restricts to a unitary ${e^{-\ii \hat H_Nt}\in U(d_{n,N})}$ on the $N$-photon sector $\mathcal H_N$.
\end{lemma}

\begin{proof}
See Appendix~\ref{prf:finite-sector}.
\end{proof}

The dimension $d_{n,N}$ of this sector scales polynomially in the number of modes $n$ whenever $N$ is fixed, 
\begin{equation}
  d_{n,N}
  =\frac{n^N}{N!}\left(1+\mathcal O(n^{-1})\right)
  =\Theta(n^N).
  \label{eq:polyN}
\end{equation}
Here and below, $f(n)=\Theta(g(n))$ denotes asymptotically tight scaling.
Exact confinement, however, is not restricted to a single sector. For $\mathcal S\subseteq\{0,\ldots,N_{\max}\}$, let
\begin{equation}
\begin{gathered}
  \mathcal H_{\mathcal S}
  :=\bigoplus_{N\in\mathcal S}\mathcal H_N,
  \\
  P_{\mathcal S}:=\sum_{N\in\mathcal S}\Pi_N,
  \quad
  d_{n,\mathcal S}
  :=\sum_{N\in\mathcal S}d_{n,N},
\end{gathered}
  \label{eq:mixedcarrier}
\end{equation}
where $\Pi_N$ is the orthogonal projector onto the $N$-photon sector.
In particular,
\begin{equation}
  \dim\mathcal H_{\leq N_{\max}}
  =\sum_{N=0}^{N_{\max}}d_{n,N}
  =\binom{n+N_{\max}}{N_{\max}}
  =\Theta(n^{N_{\max}})
  \label{eq:boundedcarrier}
\end{equation}
for fixed $N_{\max}$. 
Assume henceforth that $N_{\max}=\mathcal O(1)$, and let $\{\hat H_k\}_{k=1}^{K}$ be number-conserving Hamiltonians whose restrictions to $\mathcal H_{\mathcal S}$ have efficiently computable matrix elements. We also assume that $\rho_{\mathrm{in}}=P_{\mathcal S}\rho_{\mathrm{in}}P_{\mathcal S}$.
This leads us to the following theorem.

\begin{theorem}[Exact simulation at bounded-photon number]
For every circuit $U$ generated by the $H_k$ and every observable $O$, 
\begin{equation}
\Tr[O\,U\rho_{\mathrm{in}}U^\dagger]
=
\Tr[P_{\mathcal S}OP_{\mathcal S}\,
U\rho_{\mathrm{in}}U^\dagger].
\label{eq:exact-sector-compression}
\end{equation}
Because $O_{\mathcal S}=P_{\mathcal S}OP_{\mathcal S}$ is Hermitian and the Hermitian bracket preserves Hermiticity, its reachable module satisfies $\mathcal V(O_{\mathcal S})\subseteq\operatorname{Herm}(\mathcal H_{\mathcal S})$, where $\operatorname{Herm}(\mathcal H_{\mathcal S})$ denotes the real vector space of Hermitian operators on $\mathcal H_{\mathcal S}$, with 
\begin{equation}
  \dim_{\mathbb R}\mathcal V(O_{\mathcal S})
  \leq d_{n,\mathcal S}^2
  =\mathcal O(n^{2N_{\max}}).
  \label{eq:boundedmodule}
\end{equation}
\label{thm:boundedN}
\end{theorem}

\begin{proof}
See Appendix~\ref{prf:boundedN}.
\end{proof}

If $[O,\hat N]=0$, every operator in the reachable module is block diagonal in photon number, and hence we may write
\begin{equation}
\begin{gathered}
  \mathcal V(O_{\mathcal S})
  \subseteq
  \bigoplus_{N\in\mathcal S}\operatorname{Herm}(\mathcal H_N),
  \\
  \dim\mathcal V(O_{\mathcal S})
  \leq \sum_{N\in\mathcal S}d_{n,N}^2
  \leq d_{n,\mathcal S}^2,
\end{gathered}
  \label{eq:diagonalmodule}
\end{equation}
by counting the dimensionality of each block along the diagonal.

The support condition $\rho_{\mathrm{in}}=P_{\mathcal S}\rho_{\mathrm{in}}P_{\mathcal S}$ of \cref{thm:boundedN} allows both block-diagonal mixtures and states with coherences between different photon-number sectors in $\mathcal S$.
For number-conserving observables, cross-sector coherences are invisible, and the diagonal modules of \cref{eq:diagonalmodule} suffice. Number-changing observables can instead probe these coherences through off-diagonal operator spaces $\operatorname{Hom}(\mathcal H_N,\mathcal H_M)$, whose elements map $\mathcal H_N$ to $\mathcal H_M$. For fixed-order products of such observables, the exact sector restriction must additionally contain any intermediate sectors visited between successive insertions. The construction remains polynomial whenever the largest sector reached along these nonzero paths has bounded photon number; if every insertion is number conserving, no enlargement beyond $\mathcal S$ is required. Thus, the relevant operator space remains the observable-dependent module rather than automatically the full space $\operatorname{End}(\mathcal H_{\mathcal S})$. The explicit sector-path decomposition, including coherent cross-sector inputs, is given in Appendix~\ref{app:bounded-sector-proofs}.

We demonstrate the distinction between diagonal and off-diagonal sector modules numerically in \cref{ssec:num-bounded}, together with the measured scaling of state propagation on bounded sector unions through $N_{\max}=5$. The same finite-sector construction enables the exact evaluation of a state-dependent squared commutator, independently validated at small system sizes and extended to $400$ modes in \cref{ssec:num-otoc}.

The Gaussian and bounded-photon constructions are complementary. This is because gaussian dynamics preserve operator degree while allowing arbitrary input states on the full Fock space, while sector confinement permits arbitrary number-conserving nonlinearities at the cost of bounded photon-number support. Interleaving such nonlinear dynamics with squeezing breaks exact sector confinement; the following subsection develops a systematically improvable approximation for this intermediate regime.

\subsection{Sector representation and nonlinear interactions}

Let us enumerate the Fock states in $\mathcal H_N$ as $\{\ket a\}_{a=1}^{d_{n,N}}$, with the matrix units $\hat E^{ab}=\ket a\!\bra b$ forming a basis of $\mathrm{End}(\mathcal H_N)$. Adapting the modified generalized Gell--Mann (MGGM) representation introduced for bounded-Hamming-weight, $U(1)$-equivariant dynamics in Ref.~\cite{Barligea.2026-EnablingLieAlgebraicClassical}, their symmetric and antisymmetric off-diagonal combinations, together with the diagonal projectors, form a Hermitian basis with sparse closed-form commutation rules. The explicit bosonic adaptation and its extension to sector-pair modules are given in Appendix~\ref{app:bounded-sector-proofs}.

The bosonic operators then become restricted generator matrix elements. For example,
\begin{equation}
\begin{gathered}
  \hat a_j^\dagger\hat a_l
  \ket{\ldots,n_j,\ldots,n_l,\ldots}
  \\=
  \sqrt{n_l(n_j+1)}
  \ket{\ldots,n_j+1,\ldots,n_l-1,\ldots},
\end{gathered}
  \label{eq:hopelt}
\end{equation}
with the result understood as zero when $n_l=0$. These bosonic factors are in direct analogy to the unit hopping amplitudes of Hamming-weight sectors. Number-diagonal nonlinearities are diagonal in the same basis, so Kerr terms preserve the closure property of the operators.

The distinction from passive Gaussian dynamics first becomes nontrivial at $N=2$. In the one-photon sector, passive hopping acts through the defining $n$-dimensional representation and Kerr terms contribute at most phases. With two or more photons, nonlinearities distinguish double occupancy from separated photons and generate interacting dynamics inaccessible to passive linear quantum optics~\cite{Knill.2001-SchemeEfficientQuantum}. We can understand this through the following example.

\begin{example}[Two-photon cross-Kerr sector]
\label{ex:crosskerr}
For $N=2$,
\begin{equation}
  d_{n,2}=\frac{n(n+1)}{2},
\end{equation}
with basis
\begin{equation}
  \left\{\ket{2_j}:1\leq j\leq n\right\}
  \cup
  \left\{\ket{1_j1_l}:1\leq j<l\leq n\right\}.
\end{equation}
The cross-Kerr Hamiltonian
\begin{equation}
  \hat H_{\chi}
  =\sum_{j<l}\chi_{jl}\hat n_j\hat n_l
\end{equation}
imprints the phase $e^{-\ii\chi_{jl}t}$ on $\ket{1_j1_l}$ while leaving double-occupancy components unaffected. Interleaving this interaction with passive hopping produces non-Gaussian dynamics within the two-photon sector $\mathcal H_2$ of dimension $\mathcal O(n^2)$, from which observables such as $\langle\hat n_i\hat n_j\rangle(t)$ can be evaluated exactly. This departure from the passive quadratic baseline is demonstrated in \cref{ssec:num-bounded}, while the formation and asymptotic co-tunnelling scale of repulsively bound pairs are shown in \cref{ssec:num-doublon}.
\end{example}

\subsection{Controlled leakage under squeezing}
\label{ssec:squeezing}

A squeezer is quadratic and therefore Gaussian, but it is not number conserving, since each pair-creation or pair-annihilation term changes the total photon number by $\pm2$. Interleaving squeezing with Kerr dynamics therefore destroys exact finite-sector confinement, however the resulting leakage is nevertheless banded and admits a controlled fixed-order perturbative treatment. To that end, consider
\begin{equation}
  \hat H(r,t)=\hat H_0(t)+r\hat V(t),
  \quad
  [\hat H_0(t),\hat N]=0,
  \label{eq:squeezedH}
\end{equation}
where $\hat V(t)$ is Hermitian and
\begin{equation}
  \Pi_M\hat V(t)\Pi_N=0
  \quad\text{unless}\quad
  |M-N|=2.
  \label{eq:squeeze-selection}
\end{equation}
For initial support $\mathcal S$, we define the parity-resolved depth-$k$ band
\begin{equation}
\begin{aligned}
  \mathcal S^{(k)}
  :=
  \{
    M\in\mathbb Z_{\geq0}:\ &\exists N\in\mathcal S,\ 
    M\equiv N\!\!\!\pmod 2,\\ & 
    |M-N|\leq2k
  \}.
\end{aligned}
  \label{eq:squeezeband}
\end{equation}
Let $U_r(t)$ and $U_{r,k}(t)$ denote the propagators generated by
$\hat H(r,t)$ and $\hat H_k(r,t)=P_k\hat H(r,t)P_k$, respectively. Here, $P_k:=\sum_{M\in\mathcal S^{(k)}}\Pi_M$ is the projector onto the retained sector space $P_k\mathcal H = \bigoplus_{M\in\mathcal S^{(k)}}\mathcal H_M$, while its reachable operator module is a subspace of $\operatorname{End}(P_k\mathcal H)$. For an input state supported on $\mathcal H_{\mathcal S}$ and a number-conserving observable $O$, we can define the following time-traces,
\begin{equation}
\begin{aligned}
F(r,t)
&:=
\Tr\!\left[
O\,U_r(t)\rho_{\mathrm{in}}U_r(t)^\dagger
\right],
\\
F_k(r,t)
&:=
\Tr\!\left[
P_kOP_k\,U_{r,k}(t)\rho_{\mathrm{in}}U_{r,k}(t)^\dagger
\right],
\end{aligned}
\label{eq:squeezing-expectations}
\end{equation}
and compare their weak-squeezing behavior by expanding the corresponding interaction-picture propagators perturbatively in $r$. To do so, we let $\max\mathcal S=N_{\max}=\mathcal O(1)$ and fix $k=\mathcal O(1)$, and let $\rho_{\mathrm{in}}$ be supported on $\mathcal H_{\mathcal S}$. Then under the sector-growth conditions stated in Appendix~\ref{app:squeezing}, the full and projected expectation values of \cref{eq:squeezing-expectations} have identical coefficients through order $2k+1$,
\begin{equation}
\left.\partial_r^qF(r,t)\right|_{r=0}
=
\left.\partial_r^qF_k(r,t)\right|_{r=0},
\quad
0\le q\le2k+1.
\label{eq:coefficientmatching}
\end{equation}
Consequently, for every finite time interval $[0,T]$, there are constants $r_\ast>0$ and $C_{T,k}<\infty$ such that the finite-time error is bounded
\begin{equation}
\sup_{0\leq t\leq T}
\left|F(r,t)-F_k(r,t)\right|
\leq
C_{T,k}|r|^{2(k+1)}
\quad
\text{for }|r|\leq r_\ast.
\label{eq:squeezeerror}
\end{equation}
Moreover, the projected module stays polynomial in $n$ at every fixed band depth,
\begin{equation}
\begin{aligned}
\dim(P_k\mathcal H)&=\mathcal O(n^{N_{\max}+2k}),
\\
\dim\operatorname{End}(P_k\mathcal H)
&=\mathcal O(n^{2(N_{\max}+2k)}).
\end{aligned}
\label{eq:dim-bound}
\end{equation}

We defer a detailed account of this perturbative expansion to Appendix~\ref{app:squeezing}, where we prove the coefficient matching under explicit sector-growth assumptions and derive the corresponding uniform error bound on finite time intervals.

In the expectation value $F(r,t)$, the squeezing interaction appears on both sides of the input state, through $U_r(t)$ and $U_r(t)^\dagger$. Reaching a sector outside the depth-$k$ band requires at least $k+1$ squeezing insertions. If $O$ conserves photon number, a nonzero contribution must connect equal photon-number sectors on the two sides. Any departure of the system outside the band must therefore either return or be matched by a departure on the other side, requiring at least $2(k+1)$ insertions in total. The first possible error consequently occurs at order $\mathcal O(r^{2(k+1)})$. A number-changing observable can itself connect different final sectors, so a single departure may contribute and the error can begin at order $\mathcal O(r^{k+1})$. At finite $k$, this calculation exactly solves the projected model, and increasing $k$ by one raises the first possible error order by two for number-conserving readouts.

For every fixed band depth $k$, the retained sector space remains polynomial in $n$; see \cref{eq:dim-bound}. 
For the explicit single- and two-mode squeezed-vacuum states considered in Appendix~\ref{app:squeezing}, the photon-number tails provide an additional non-perturbative estimate for bounded number-conserving observables. Denoting the squeezing parameter of these states by $s$, the two-mode squeezed vacuum has discarded probability $\tau_k^{(2)}(s)=|\tanh s|^{2(k+1)}$ above $k$ generated pairs~\cite{Weedbrook2012}. Fixed $s$ and fixed target accuracy therefore require a band depth independent of $n$, whereas an accuracy $\tau_k^{(2)}(s)\leq n^{-c}$ requires $k=\Theta(\log n)$ and produces a quasi-polynomial retained sector space of dimension $n^{\mathcal O(\log n)}$. These conclusions do not follow from coefficient matching for generic multimode or interleaved Kerr--squeezing dynamics.

In \cref{ssec:num-squeezing}, we test the parity-resolved construction, the successive-band reference certificate, and the predicted $\mathcal O(r^{2(k+1)})$ error orders for band depths through $k=3$ and systems of up to $n=12$ modes.

Finally, we emphasize that while number conservation gives an exact sector restriction at any photon number, polynomial scaling requires a bounded maximum photon number of occupied sectors. A single sector with $N=\nu n$ already has dimension $\binom{(1+\nu)n-1}{\nu n}=\exp(\Theta(n))$, so the exact family consists of bounded-photon, number-conserving dynamics. Squeezing therefore provides a controlled approximation around this family, but not an additional exact finite module. The infinite-orbit construction in \cref{sec:boundaries} makes this distinction explicit.

\section{Boundaries and further polynomial regimes}
\label{sec:boundaries}

The Gaussian construction of \cref{sec:heisenbergweyl} rests on degree truncation, in which quadratic Hamiltonians preserve the finite-degree polynomial modules of the quadratures. Meanwhile the bounded-photon construction of \cref{sec:nongaussian} rests on a different mechanism, sector confinement, where number-conserving Hamiltonians preserve each fixed total-photon sector. We now explain why these mechanisms are needed. In bosonic systems, neither a finite generator algebra nor a small set of non-Gaussian generators are enough to guarantee exact and efficient simulability. The controlling object remains the reachable operator module $\mathcal V(O)$ of~\cref{thm:reach}. Indeed, when we recognize this, it is apparent that there are other sets of observables whose reachable operator module remains polynomial. In this section, we identify these extra families after understanding why finite generator algebras are not sufficient.

\subsection{Finite generator algebras are not sufficient}
A first pitfall is to identify simulability with the dimension of the Lie algebra generated by the Hamiltonians. This is sufficient in the standard $\gsim$ setting when the observable lies inside the propagated algebra, but it is not necessary and, in bosonic systems, not even the right object to analyze for arbitrary observables. We can see why this is the case through the following example.

\begin{example}[A finite cubic algebra with an infinite observable orbit]
\label{ex:projective-cubic}
Consider the one-mode projective generators~\cite{Turbiner.1988,GonzalezLopez.1991-QuasiExactly}
\begin{equation}
\begin{gathered}
  K_-=\hat p,\qquad K_0=\frac12(\hat x\hat p+\hat p\hat x),\\ K_+=\frac13(\hat x^2\hat p+\hat x\hat p\hat x+\hat p\hat x^2).
\end{gathered}
  \label{eq:projective-generators}
\end{equation}
A direct calculation using $[\hat x,\hat p]=\ii$ gives
\begin{equation}
\begin{gathered}
\mathrm{ad}_{K_0}(K_-)=-K_-,
\qquad
\mathrm{ad}_{K_0}(K_+)=K_+,
\\
\mathrm{ad}_{K_-}(K_+)=2K_0.
\end{gathered}
\label{eq:projective-sl2}
\end{equation}
Thus the generators close exactly on an algebra isomorphic to
$\mathfrak{sl}(2,\mathbb R)$, the three-dimensional Lie algebra of real traceless $2\times2$ matrices. This is the standard projective realization by first-order differential operators. Nevertheless,
\begin{equation}
\mathrm{ad}_{K_+}^{\,r}(\hat x)
=
r!\,\hat x^{r+1},
\qquad
r\geq0.
\label{eq:projective-flow}
\end{equation}
Hence $\mathcal V(\hat x)$ contains linearly independent powers of arbitrarily high degree and is infinite-dimensional, even though the generator algebra is finite-dimensional. The calculation is detailed in \cref{app:projective-cubic-flow}.
\end{example}

The example is the converse of the Gaussian situation. For Gaussian dynamics, the ambient algebra has dimension $\mathcal O(n^2)$ while a linear observable lives in a much smaller $(2n+1)$-dimensional module. Here the generator algebra is tiny, but the observable orbit is infinite. Simulability is therefore a property of the pair consisting of the dynamics and the measured observable, through $\mathcal V(O)$, not of $\dim\g$.

The same situation appears in an even more familiar form for non-Gaussian generators. Indeed, a single Kerr Hamiltonian $H=\chi \hat n^2$ generates a one-dimensional Lie algebra, but the annihilation operator $\hat a$ has an infinite adjoint orbit on the full Fock space, since repeated commutators generate $\hat a$ times arbitrarily high powers of $\hat n$. Thus, a Kerr Hamiltonian becomes finite in our framework only after imposing an additional structure, namely fixed total photon number as in \cref{sec:nongaussian}. The elementary commutator calculation is recorded in \cref{app:kerr-full-fock-orbit}.

\subsection{Generic cubic dynamics and universality}
The Gaussian algebra is special because commutation with quadratic Hamiltonians preserves polynomial degree.  Suitable non-Gaussian resources added to Gaussian operations instead yield universal continuous-variable computation, with standard realizations including cubic or higher-order nonlinearities and GKP resources~\cite{Lloyd1999,Gottesman2001,Calcluth2024}. Correspondingly, the degree filtration used in \cref{prop:moment-tower} is no longer invariant under such general generator sets.
Universality does not imply that every individual observable has an infinite reachable module. In fact, special observables may remain tractable even under otherwise universal dynamics. Rather, under standard complexity assumptions, one should not expect a uniformly constructible family of polynomial-dimensional modules satisfying all accessibility conditions for a computationally complete collection of inputs and readouts. The families identified here should consequently be understood as sufficient polynomially tractable regimes.

\subsection{Nilpotent cubic-phase dynamics}
\label{ssec:nilpotentalg}
There is nevertheless a second non-Gaussian polynomial regime, distinct from bounded photon number. It contains cubic and higher phase gates, but only in a nilpotent, kinetic-free setting. Let us fix $q\ge3$ and define
\begin{equation}
  \g_{\mathrm{cph}}^{(q)}:=\spann_{\mathbb R}\{\hat p_1,\ldots,\hat p_n\}\oplus\mathcal P_{\le q}(\hat x_1,\ldots,\hat x_n),
  \label{eq:cph}
\end{equation}
where $\mathcal P_{\le q}(\hat x_1,\ldots,\hat x_n)$ denotes the real vector space of position polynomials of total degree at most $q$, including the constant.

\begin{proposition}[Nilpotent phase closure]
\label{prop:cph}
The space $\g_{\mathrm{cph}}^{(q)}$ is a nilpotent Lie algebra under the Hermitian bracket $\mathrm{ad}_H(A)=\ii[H,A]$. Its dimension is
\begin{equation}
  \dim\g_{\mathrm{cph}}^{(q)}=n+\binom{n+q}{q}=\mathcal O(n^q),
  \label{eq:cph-dimension}
\end{equation}
which is polynomial in $n$ for fixed $q$. In particular, it contains non-Gaussian generators such as $\hat x_k^3$ and $\hat x_k\hat x_l\hat x_m$.
\end{proposition}

\begin{proof}
    See Appendix~\ref{prf:cph}.
\end{proof}

The intuition is that position polynomials commute with one another, the momenta commute with one another, and $\ii[\hat p_k,\hat f(\bm x)]=\partial_{x_k}\hat f(\bm x)$ lowers polynomial degree. Hence no commutator produces products such as $\hat x\hat p$, $\hat p^2$, or higher momentum powers.

Hence the dynamics generated by $\g_{\mathrm{cph}}^{(q)}$ consist of position-diagonal phase gates, shears, and displacements. They include cubic phase gates of the form $e^{\ii\gamma\hat x_k\hat x_l\hat x_m}$, but the nilpotent algebra itself excludes kinetic terms $\hat p_k^2$ and quadratic dilation terms. Adding kinetic terms to cubic position polynomials can destroy the triangular closure and produce unbounded degree growth. Some restricted quadratic additions, such as a pure dilation, instead yield finite solvable extensions, so the nilpotent family above is sufficient but not maximal. Its tractability follows from a third closure mechanism, distinct from Gaussian degree preservation and photon-number confinement, namely a weighted polynomial filtration adapted to the triangular action of the phase dynamics.

To define this filtration, for multi-indices $\bm\alpha,\bm\beta\in\mathbb Z_{\ge0}^n$, fix the normal ordering
\begin{equation}
\hat{\bm x}^{\bm\alpha}\hat{\bm p}^{\bm\beta}
:=
\prod_{j=1}^n\hat x_j^{\alpha_j}
\prod_{j=1}^n\hat p_j^{\beta_j}.
\label{eq:cph-normal-order}
\end{equation}
For $q\ge3$, define
\begin{equation}
\begin{aligned}
\mathcal F_{m,q}
:=
\spann_{\mathbb C}
\{
\hat{\bm x}^{\bm\alpha}\hat{\bm p}^{\bm\beta}:
|\bm\alpha|+(q-1)|\bm\beta|
\le m(q-1)
\}.
\end{aligned}
\label{eq:cph-filtration}
\end{equation}

\begin{proposition}[Reachable observables for nilpotent phase dynamics]
\label{prop:cph-reach}
Let $H_1,\ldots,H_K\in\g_{\mathrm{cph}}^{(q)}$, and let $O$ be a polynomial of total degree at most $m$, written in the fixed normal ordering of \cref{eq:cph-normal-order}. Then
\begin{equation}
\dim_{\mathbb R}\mathcal V(O)
\leq
\dim_{\mathbb C}\mathcal F_{m,q}
=
\Theta\!\left(n^{m(q-1)}\right)
\label{eq:cph-reachable-dimension}
\end{equation}
for fixed $q$ and $m$.
\end{proposition}

\begin{proof}
See Appendix~\ref{prf:cph-reach}.
\end{proof}
Hence expectation values, correlation functions, and gradients are exactly computable whenever the required generator actions and input moments are efficiently available. Using the fixed-order product construction of \cref{cor:module-correlators}, we demonstrate exact single- and multimode cubic-phase propagation numerically in Appendix~\ref{app:num-cph} through a non-Gaussian fourth momentum cumulant and a cross-mode momentum correlator, both validated against independent wavefunction-grid calculations in the one- and two-mode cases. We benchmark coefficient propagation separately for systems of up to several hundred modes.

For fixed $q$ and fixed observable degree, \cref{prop:cph-reach} covers arbitrary polynomial quadrature observables. These include fixed-degree photon-number observables such as ${\hat n_k=\frac12(\hat x_k^2+\hat p_k^2-\mathds{1})}$ and fixed-order products of number operators. Although their ordinary polynomial degree may increase under cubic-phase evolution, their weighted degree remains bounded within the filtration of \cref{eq:cph-filtration}. The construction does not cover observables whose degree grows with the number of modes, complete photon-counting distributions, or generator sets that destroy this filtration.

The underlying finite-dimensional algebras are related to known mathematical structures. For $n=1$, the nilpotent phase algebra lies within the single-mode classification of finite-dimensional bosonic Lie-algebras in Ref.~\cite{Heib2025}. Closely related finite-dimensional algebras of first-order differential operators in one and two variables were classified in the quasi-exact-solvability literature~\cite{Turbiner.1988,GonzalezLopez.1991-QuasiExactly}, including realizations related to the nilpotent and projective structures considered here. These works do not formulate an arbitrary-mode reachable-module simulability theorem. Likewise, the cubic phase gate and its role as a non-Gaussian resource for universal continuous-variable computation are standard~\cite{Lloyd1999,Gottesman2001,Calcluth2024}. Our contribution is therefore not a new classification of the underlying algebras, but their simulability interpretation and the explicit separation between finite generator algebras with polynomial observable modules and finite algebras with infinite observable orbits.

Our results identify three distinct mechanisms that produce finite, polynomially accessible reachable modules, i.e., preservation of polynomial degree under Gaussian dynamics, confinement to bounded photon-number sectors under number-conserving nonlinear dynamics, and weighted degree-lowering nilpotence under phase dynamics. Outside such protected regimes, non-Gaussian evolution on the full Fock space can generate infinite-dimensional observable orbits, even when the generator algebra itself is finite-dimensional. Bosonic simulability is therefore governed by the closure and accessibility of the observable's reachable module.

\section{Numerical demonstrations}
\label{sec:numerics}
We next numerically test the simulation guarantees identified above and demonstrate the interaction quantities they make accessible on five illustrative problems. The calculations in \cref{ssec:num-bounded,ssec:num-squeezing} test exact sector-union propagation, its fixed-$N_{\max}$ scaling, and the predicted leakage and readout-error orders under squeezing. In \cref{ssec:num-doublon,ssec:num-otoc}, we apply the same finite-sector construction to repulsively bound-pair dynamics and a state-dependent OTOC on chains of up to $400$ modes. Finally, \cref{ssec:num-edge} combines a non-Abelian Chern calculation with flux-reversed edge propagation for an interacting doublon multiplet. Appendices~\ref{app:num-gaussian}--\ref{app:num-cph} provide Gaussian and non-Gaussian-input benchmarks, finite-difference gradient validation and Kerr-assisted control, and independent tests of nilpotent polynomial-phase propagation.

For the following numerical results, we set $\hbar=1$ throughout. Whenever a hopping scale $J$ is present, we express energies in units of $J$ and times in units of $1/J$. All calculations were implemented in double-precision Python using NumPy~\cite{Harris.2020-ArrayProgrammingNumPy} and SciPy~\cite{Virtanen.2020-SciPyFundamentalAlgorithms}. 

\subsection{Exact dynamics on bounded sector unions}
\label{ssec:num-bounded}

We first test the exact bounded-sector construction on four modes in the two-photon sector based on the number-conserving Kerr Hamiltonian,
\begin{equation}
\hat H
=
\sum_{j,k=0}^{3}h_{jk}\hat a_j^\dagger\hat a_k
+\chi_{01}\hat n_0\hat n_1
+\chi_{23}\hat n_2\hat n_3,
\label{eq:num-cross-kerr}
\end{equation}
where the first term describes passive mode mixing generated by the one-particle matrix $h$, while the cross-Kerr terms introduce quartic interactions. We generate the complex Hermitian one-particle matrix ${h=(A+A^\dagger)/2}$ by drawing ${A_{jk}=0.6(X_{jk}+\ii Y_{jk})}$, with independent
$X_{jk},Y_{jk}\sim\mathcal N(0,1)$. We choose $\chi_{01}=1$ and $\chi_{23}=0.8$, and initialize the system in $|1,1,0,0\rangle$. Setting both Kerr couplings to zero while retaining $h$ defines the quadratic baseline.

The observable we follow is the fourth-order correlator $C_{01}(t)=\langle\hat n_0(t)\hat n_1(t)\rangle$ over $0\leq t\leq2.5$ at $80$ equally spaced times, shown in
Panel~\ref{fig:num-bounded}(a). Exact propagation in the $d_{4,2}=10$ dimensional sector $\mathcal H_2$ agrees with an independent tensor-product Fock calculation using three local Fock levels to floating point precision.
We see also that the cross-Kerr curve separates visibly from the baseline, demonstrating exact quartic dynamics beyond the closed first- and second-moment evolution of the covariance formalism~\cite{Weedbrook2012} (see also \cref{ssec:moment-modules}). 

Next, we give the diagonal and off-diagonal sector blocks an operational test by comparing states with identical photon-number populations but different cross-sector coherences. For this, we evolve an open five-site Bose--Hubbard chain~\cite{Fisher.1989-BosonLocalizationSuperfluid,Jaksch.1998-ColdBosonicAtoms} in the self-Kerr convention, 
\begin{equation}
\hat H_{\mathrm{mix}}
=
-J\sum_{j=0}^{3}
\left(
\hat a_j^\dagger\hat a_{j+1}
+\mathrm{h.c.}
\right)
+\frac{U}{2}\sum_{j=0}^{4}\hat n_j^2
\label{eq:num-mixed-sector-hamiltonian}
\end{equation}
with $J=1$ and $U=1.5$. Since $\hat n_j^2=\hat n_j(\hat n_j-1)+\hat n_j$, this convention differs from the standard Bose--Hubbard interaction by the term $U\hat N/2$. It therefore leaves the dynamics within each fixed-photon sector unchanged up to a global phase, while producing relative phases between different total-photon sectors. We label the sites of the five-mode chain by $j=0,\ldots, 4$, and place the photons at the central site $c=2$. The two input states are
\begin{equation}
\begin{aligned}
|\psi_{\mathrm{coh}}\rangle
&=
\frac{|0\rangle+|1_c\rangle+|2_c\rangle}{\sqrt3},
\\
\rho_{\mathrm{deph}}
&=
\frac13
\left(
|0\rangle\langle0|
+|1_c\rangle\langle1_c|
+|2_c\rangle\langle2_c|
\right),
\end{aligned}
\label{eq:num-mixed-inputs}
\end{equation}
where $|N_c\rangle$ denotes $N$ photons localized at site $c$. Both inputs have identical populations in the sectors $N=0,1,2$, but only $|\psi_{\mathrm{coh}}\rangle$ contains coherences between them, as $\rho_{\mathrm{deph}}$ is block diagonal in the photon-number decomposition. Their comparison therefore isolates which observables can detect these cross-sector coherences.

The invariant sector union $\mathcal H_0\oplus\mathcal H_1\oplus\mathcal H_2$ has dimension $1+5+15=21$. We evaluate the dynamics at $100$ times over $0\leq t\leq4$. As shown in \cref{fig:num-bounded}(b), the local photon number density curves $\langle\hat n_c(t)\rangle$ agree to numerical precision, whereas the quadrature expectations $\langle\hat x_c(t)\rangle$ differ by as much as $1.14$. This distinction follows directly from the sector-block structure, as number-conserving observables evolve only within the diagonal blocks
$\bigoplus_N\operatorname{End}(\mathcal H_N)$ and therefore cannot detect cross-sector coherences~\cite{Bartlett.2007-ReferenceFramesSuperselection,Descamps.2024-SuperselectionRulesBosonic}. By contrast, a phase-referenced measurement of $\hat x_c$, which changes the photon number by one, probes the adjacent off-diagonal blocks
$\operatorname{Hom}(\mathcal H_N,\mathcal H_{N\pm1})$.

Panel~\ref{fig:num-bounded}(c) measures the cost of sparse state propagation on
${
\mathcal H_{\leq N_{\max}}
=
\bigoplus_{N=0}^{N_{\max}}\mathcal H_N}
$
under the open-chain Hamiltonian in \cref{eq:num-mixed-sector-hamiltonian}, now with $J=1$ and $U=0.7$. At each pair $(n,N_{\max})$, we draw coefficients $z_\alpha=X_\alpha+\ii Y_\alpha$ with independent $X_\alpha,Y_\alpha\sim\mathcal N(0,1)$ and propagate the normalized state $\ket{\psi}=z/\|z\|$. This construction samples unit vectors uniformly with respect to Haar measure.
The largest mode counts we test for $N_{\max}=1,\ldots,5$ are $1024, 512, 256, 64, 40$, respectively. Each timing measures the application of $e^{-\ii\hat Ht}$ at $t=0.5$ and excludes construction of the basis and Hamiltonian.

The polynomial reduction compared with a tensor-product Fock representation that does not resolve total photon number is substantial. For example, a local cutoff capable of representing five photons on each of $40$ modes requires $6^{40}\simeq 2^{103}$ basis states, whereas the exact subspace with at most five photons in total has dimension $\dim\mathcal H_{\leq5}=\binom{45}{5}=1\,221\,759$. This comparison illustrates the gains we make from resolving total photon number.
\cref{fig:num-bounded}(c) likewise benchmarks sparse state propagation on this invariant Hilbert space of dimension $\binom{n+N_{\max}}{N_{\max}}$. We emphasize the distinction of this with propagation in the full endomorphism space, whose dimension is the square of the Hilbert-space dimension, or with a potentially smaller observable-dependent reachable module. Indeed, our comparison demonstrates polynomial scaling in $n$, and establishes that the polynomial Hilbert spaces entering \cref{thm:boundedN} remain computationally accessible across the reported mode ranges.

\begin{figure*}[htbp]
\centering
\includegraphics[width=\textwidth]{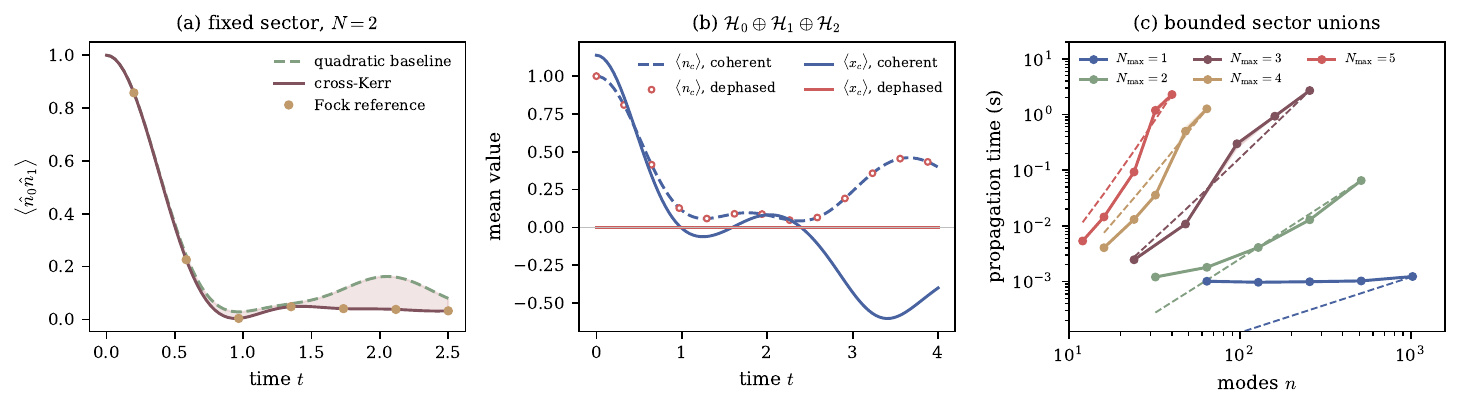}
\caption{Exact dynamics on bounded sector unions. (a) Two-photon correlator under a quadratic Hamiltonian and the same Hamiltonian after adding cross-Kerr interactions; markers are independent truncated-Fock calculations for the interacting dynamics. (b) A coherent superposition across $N=0,1,2$ and its sector-dephased counterpart give identical number readouts but different quadrature readouts. (c) Median sparse-propagation time, with interquartile ranges over seven repetitions after warm-up, for states in $\mathcal H_{\leq N_{\max}}$ with $N_{\max}=1,\ldots,5$. The sector union Hilbert space has exact dimension $\binom{n+N_{\max}}{N_{\max}}$. Dashed curves show this exact
finite-$n$ dependence, normalized at the largest mode count for each
$N_{\max}$; the $N_{\max}=1$ timings are dominated by fixed propagation
overhead.}
\label{fig:num-bounded}
\end{figure*}

\subsection{Repulsively bound photon pairs}
\label{ssec:num-doublon}

We next use exact two-photon propagation of~\cref{sec:nongaussian} to reproduce the formation and strong-coupling motion of repulsively bound pairs in an open Bose--Hubbard chain~\cite{Hartmann.2006-StronglyInteractingPolaritons,Greentree.2006-QuantumPhaseTransitions},
\begin{equation}
H_{\mathrm{BH}}=-J\sum_j\left(\hat a_j^\dagger\hat a_{j+1}+\mathrm{h.c.}\right)+\frac{U}{2}\sum_j\hat n_j(\hat n_j-1).
\label{eq:num-bose-hubbard}
\end{equation}
This is the conventional Bose--Hubbard form of the same quartic on-site interaction used in \cref{eq:num-mixed-sector-hamiltonian}. Indeed, $\sum_j\hat n_j^2=\sum_j\hat n_j(\hat n_j-1)+\hat N$, and the additional term is relevant to relative phases in the mixed-sector calculation above, but contributes only a global phase within the fixed two-photon sector considered here. 
For $n=41$, we place both photons on the central site $c=20$, preparing the doublon $|2_c\rangle$, and propagate the state exactly within $\mathcal H_2$.

Figures~\ref{fig:num-doublon}(a,b) show the local density $\langle\hat n_j(t)\rangle$ for $U=0J$ and $U=8J$ over $0\leq t\leq10$ at $120$ times, and panel (c) shows the doublon fraction correlator,
\begin{equation}
D(t)
=
\frac12\sum_j
\left\langle
\hat n_j(t)\bigl(\hat n_j(t)-1\bigr)
\right\rangle
\label{eq:num-doublon-fraction}
\end{equation}
for $U/J=0,2,8$. We see the noninteracting pair spreads rapidly, whereas strong repulsion retains the photons in a slowly propagating composite state, consistent with repulsively bound-pair physics~\cite{Winkler2006,Morvan.2022-FormationRobustBound}.
A separate $n=4$, $U=8J$ calculation at six times up to $t=2$ agrees with an independent tensor-product Fock-space calculation to within $7\times10^{-14}$. The local basis
$\{|0\rangle,|1\rangle,|2\rangle\}$ represents the invariant
two-photon sector exactly, so the residual is consistent with
accumulated floating-point round-off across the two implementations.

We quantify the slow pair motion by placing the same model on periodic chains with $n=101,\,201$, and $401$ modes. For $U/J=4,6,8,12,16,24,32,48$, we isolate the repulsive bound band to eigenvalue tolerance $10^{-10}$ and define $J_{\mathrm{eff}}=W/4$ from its bandwidth $W$. The factor four is the bandwidth of a nearest-neighbor cosine band with hopping $J_{\mathrm{eff}}$. The results for all three system sizes are indistinguishable at the resolution of \cref{fig:num-doublon}(d). 
A log--log fit over $U/J\geq8$ at $n=401$ gives
\begin{equation}
\frac{J_{\mathrm{eff}}}{J}
\propto
\left(\frac{U}{J}\right)^{-0.9711\pm0.0068}.
\end{equation}
At the largest interaction, $U/J=48$, we find $UJ_{\mathrm{eff}}/J^2=1.9965$, approaching the strong-coupling
limit $2$. This agrees with the prediction $J_{\mathrm{eff}}\simeq2J^2/U$ for the magnitude of the effective
dimer hopping~\cite{Valiente2008,Folling.2007-DirectObservationSecondorder}.

\begin{figure*}[htbp]
\centering
\includegraphics[width=\textwidth]{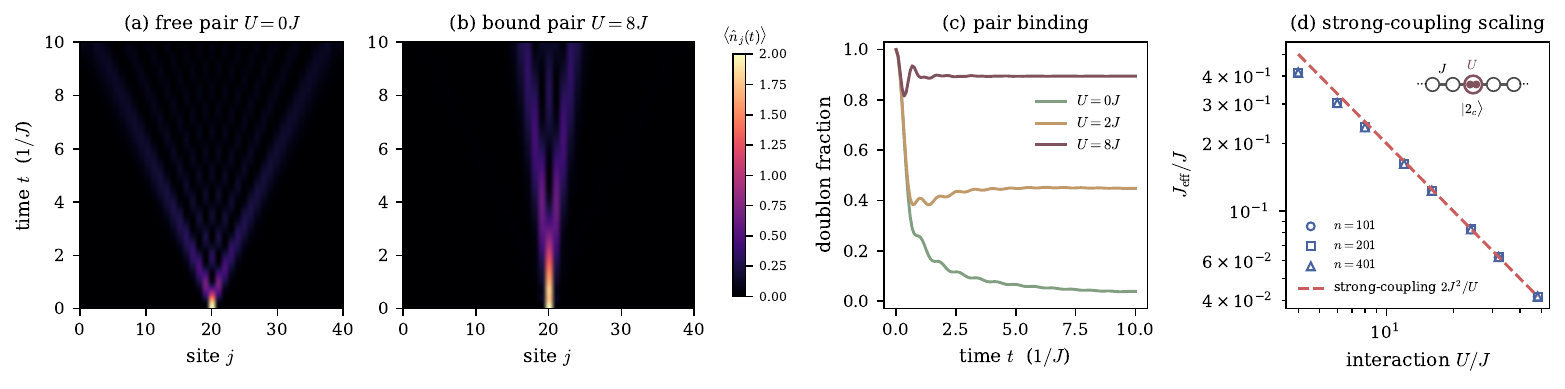}
\caption{Repulsively bound photon pairs in the two-photon Bose--Hubbard chain. (a,b) Local density $\langle\hat n_j(t)\rangle$ after initializing both photons at the central site of an open $n=41$ chain, shown for $U=0J$ and $U=8J$ on a shared color scale. (c) Doublon fraction, equal to the probability that both photons occupy the same site, for $U/J=0,2,8$. (d) Effective doublon hopping $J_{\mathrm{eff}}=W/4$, extracted from the width $W$ of the isolated repulsive bound band on periodic chains with $n=101,201,401$. The dashed line shows the strong-coupling prediction $2J^2/U$~\cite{Valiente2008}. A log--log fit for $n=401$ over $U/J\geq8$ gives the exponent $-0.9711\pm0.0068$; at $U/J=48$, $UJ_{\mathrm{eff}}/J^2=1.9965$. The inset sketches the local hopping $J$, onsite interaction $U$, and initial doublon $|2_c\rangle$.}
\label{fig:num-doublon}
\end{figure*}

\subsection{Controlled photon-sector leakage under squeezing}
\label{ssec:num-squeezing}

Squeezing breaks photon-number conservation but preserves photon-number parity. We test the corresponding leakage and readout-error orders derived in \cref{ssec:squeezing} using
\begin{equation}
\begin{aligned}
\hat H(r)
=
&-J\sum_{j=0}^{n-2}
\left(
\hat a_j^\dagger\hat a_{j+1}
+\mathrm{h.c.}
\right)\\
&+U\hat n_0\hat n_1
+\frac{r}{2}
\left(
\hat a_0\hat a_1
+\hat a_0^\dagger\hat a_1^\dagger
\right).
\end{aligned}
\label{eq:num-squeezing-hamiltonian}
\end{equation}
The cross-Kerr interaction makes the dynamics non-Gaussian, while the squeezing term changes total photon number by $\pm2$. Starting from $\mathcal H_2$, the depth-$k$ calculation evolves under the projected Hamiltonian $\hat H_k(r)=P_k\hat H(r)P_k$, where $P_k$ retains the parity-compatible sectors reachable in at most $k$ steps. For fixed $k$, the dimension of this subspace scales as $d_k=\mathcal O(n^{2+2k})$.

To test convergence, we use $n=2$, $U/J=2$, and the input $\ket{\psi_{\mathrm{in}}}=\ket{1,1}$, so that the system allows for tightly converged numerical reference values. \cref{fig:num-squeezing}(a) shows
$F_k(r,t)
=
\langle\hat n_0\hat n_1\rangle_k(t)$
for $k=0,1,2$ at $r/J=0.8$ and $0\leq Jt\leq 2$. The $k=8$ reference agrees with the next band within $1.48\times10^{-14}$ over the reported parameter range. It thus provides a converged finite-band reference for this observable, with an exact infinite-Fock-space result being naturally out of reach due to its infinite-dimensional basis.

We define the maximum readout error by
\begin{equation}
\epsilon_k(r)
=
\max_{0\leq t\leq2.5}
\left|
F_k(r,t)-F_{\mathrm{ref}}(r,t)
\right|.
\label{eq:num-squeezing-error}
\end{equation}
Panel \ref{fig:num-squeezing}(b) evaluates this error at 14 logarithmically spaced squeezing strengths between $r/J=0.02$ and $0.89$. 
Using the nominal window $0.04\leq r/J\leq0.50$ and excluding errors below $10^{-13}$ gives exponents $2.002$, $3.991$, and $5.984$ for $k=0,1,2$. Allowing each endpoint of the fit window to move by one sampled squeezing strength gives the respective ranges $2.001$--$2.004$, $3.985$--$3.995$, and $5.971$--$5.991$. These values agree with the predicted exponents $2$, $4$, and $6$, supporting $\epsilon_k(r)=\mathcal O((r/J)^{2(k+1)})$ in the weak-squeezing regime.

Panel \ref{fig:num-squeezing}(c) benchmarks propagation times at $r/J=0.4$ and $Jt=0.25$, where the calculations reach $n=96$, $24$, and $12$ for $k=0$, $1$, and $2$, respectively, and follow the increasing projected-subspace dimensions $d_k=\Theta(n^{2+2k})$. 

Panel \ref{fig:num-squeezing}(d) compares the two perturbative scales directly for $n=2,4,6,8,12$ at $Jt=0.5$. We compute both the readout error and the norm of the reference state outside the depth-$k$ subspace over ten squeezing strengths between $r/J=0.02$ and $0.70$. The $k=6$ and $k=7$ references agree in the correlator within $3.61\times10^{-15}$, while their phase-insensitive state distance remains below $1.03\times10^{-7}$, beneath the $10^{-6}$ leakage-fit floor. For $k=0,1,2,3$, the fitted leakage and readout exponents differ from $k+1$ and $2(k+1)$ by at most $0.018$ and $0.027$, respectively. At $k=4$, too few values remain above the respective numerical floors to support either fit.

Hence each calculation solves its projected finite-dimensional model exactly up to numerical precision. The comparison with the converged references quantifies the approximation to the unbounded squeezed dynamics: at fixed band depth the cost remains polynomial in $n$, while increasing $k$ systematically postpones both state leakage and number-conserving readout errors. 

\begin{figure*}[t]
\centering
\includegraphics[width=\textwidth]{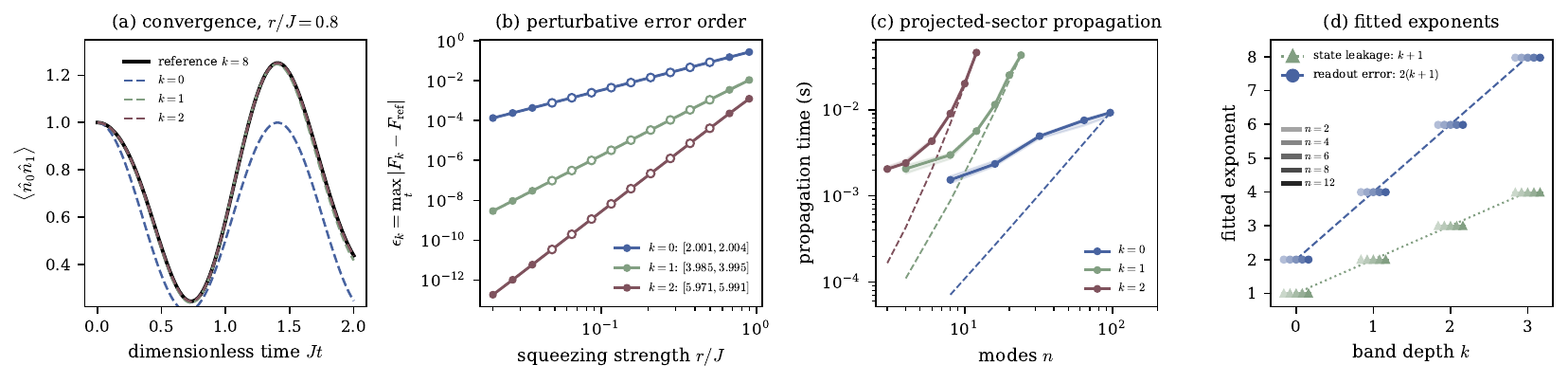}
\caption{Controlled photon-sector leakage under squeezing. (a) Convergence of the number-conserving correlator $\langle\hat n_0\hat n_1\rangle$ with band depth $k$ for $n=2$, $U/J=2$, and $r/J=0.8$. The $k=8$ reference agrees with the next band within $1.48\times10^{-14}$. (b) Maximum readout error versus squeezing strength; open markers indicate the points included in the power-law fits. (c) Median propagation times and interquartile ranges over seven repetitions. Dashed curves show the exact retained-sector dimensions $d_k(n)$,
normalized at the largest mode count for each $k$; their asymptotic
degrees are $2+2k$. (d) Fitted exponents for the leakage norm (triangles) and number-conserving readout error (circles), compared with the predictions $k+1$ and $2(k+1)$, respectively. Shading encodes system size from $n=2$ (lightest) to $n=12$ (darkest); the $k=4$ fits are omitted because the corresponding errors fall below the numerical floor.}
\label{fig:num-squeezing}
\end{figure*}

\subsection{Exact operator spreading}
\label{ssec:num-otoc}
Operator spreading is the standard analysis for how local disturbances propagate under interacting dynamics, and OTOCs have become its usual probe~\cite{Larkin.1969-QuasiclassicalMethodTheory,Maldacena.2016-BoundChaos}, with front structure and its relation to Lieb--Robinson bounds studied extensively in lattice models~\cite{Roberts.2016-LiebRobinsonBoundButterfly,Nahum.2018-OperatorSpreadingRandom,vonKeyserlingk.2018-OperatorHydrodynamicsOTOCs} and measured on quantum hardware~\cite{Mi.2021-InformationScramblingQuantum}. The fixed-order product construction of \cref{cor:module-correlators} covers these operator-level quantities, while the bounded two-photon sector allows us to evaluate their action without constructing the full $d_{n,2}^2$-dimensional operator representation.

For the open two-photon Bose--Hubbard chain of \cref{eq:num-bose-hubbard}, we compute the state-dependent squared commutator
\begin{equation}
C_{\psi}(i,t)=\left\|[\hat n_i(t),\hat n_c]|\psi\rangle\right\|^2,
\quad
\hat n_i(t)=e^{\ii Ht}\hat n_i\, e^{-\ii Ht},
\label{eq:num-otoc}
\end{equation}
where $c=\lfloor n/2\rfloor$ is the reference site for sites labelled ${0,\ldots,n-1}$. Expanding the squared commutator gives an out-of-time-ordered four-point correlation function~\cite{Larkin.1969-QuasiclassicalMethodTheory,Maldacena.2016-BoundChaos}.

Here, we draw $|\psi\rangle\in\mathcal H_2$ by normalizing a complex Gaussian vector, the resulting states of which are Haar distributed on the two-photon sector and satisfy
\begin{equation}
\mathbb E_\psi C_\psi(i,t)
=
\frac{1}{d_{n,2}}
\Tr_{\mathcal H_2}
\left(
[\hat n_i(t),\hat n_c]^\dagger
[\hat n_i(t),\hat n_c]
\right).
\end{equation}
This means that each realization is a stochastic estimate of the normalized sector trace~\cite{Steinigeweg.2014-SpincurrentAutocorrelations,Goldstein.2006-CanonicalTypicality,Popescu.2006-EntanglementFoundationsStatistical}, while retaining the state dependence needed to assess sample-to-sample variation. We evaluate the action of the commutator on $|\psi\rangle$ through sparse Schr\"odinger propagation in $\mathcal H_2$ and never construct the corresponding $d_{n,2}^2$ operator matrix as its size rapidly becomes infeasible with $n$.

For an independent small-system validation, we take $n=5$, $U=8J$, and evaluate nine times over $0\leq t\leq4$. The fixed-sector calculation agrees with an exact tensor-product Fock calculation on the local basis $\{|0\rangle,|1\rangle,|2\rangle\}$ to within $2.30\times10^{-13}$. 
For the deterministic realization shown in \cref{fig:num-otoc}(a), the $U=0J$ quadratic baseline differs from the interacting Fock-space reference by as much as $1.24$. This order-one difference reflects the omitted quartic interaction, but its magnitude is state dependent.

Panels~\ref{fig:num-otoc}(b,c) show $\log_{10}C_\psi(i,t)$ for a single seeded random state on $n=400$ sites, where $d_{400,2}=80\,200$ and the corresponding dense operator representation would contain approximately $6.43\times10^9$ entries. We test $U=0J$ and $U=8J$, and sample $48$ times over $0\leq t\leq48$. In both cases, the threshold advances approximately linearly over the displayed interval~\cite{Nahum.2018-OperatorSpreadingRandom,vonKeyserlingk.2018-OperatorHydrodynamicsOTOCs}. The interaction thus has little effect on this threshold but substantially redistributes the commutator weight behind it. 

Panel~\ref{fig:num-otoc}(d) examines the finite-size and state dependence of this observation for $n=120,200,300,400$, and interactions $U/J=0,2,4,8,12,16$, for three independent trials of random initial states. The corresponding final times are $20,32,40,48$. At each time, we define the threshold-front position by 
\begin{equation}
r_\eta(t)
=
\max\left\{
|i-c|:C_\psi(i,t)\geq\eta
\right\},
\qquad
\eta=10^{-5},
\end{equation}
and extract its slope from the points satisfying
$0<r_\eta(t)<0.4n$, thereby excluding the initial plateau and boundary-affected data. The resulting slope depends only weakly on $U$ but varies with system size and with the threshold used to define the front. We therefore use it only as a numerical diagnostic of the leading front and do not identify it with a Lieb--Robinson velocity~\cite{Lieb.1972-FiniteGroupVelocity}, whose existence and value for lattice bosons with unbounded on-site occupation is itself delicate and has been established only under restrictions such as bounded
density~\cite{Schuch.2011-InformationPropagationInteracting,Yin.2022-FiniteSpeedQuantum,Faupin.2022-MaximalSpeedMacroscopic}.

\begin{figure*}[htbp]
\centering
\includegraphics[width=\textwidth]{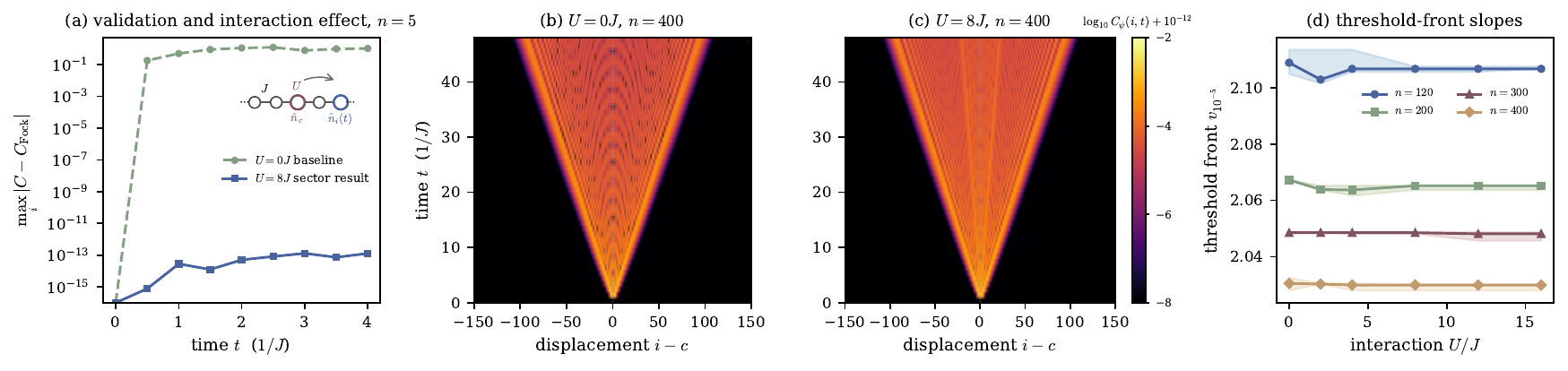}
\caption{State-dependent operator spreading in the two-photon Bose--Hubbard chain. Time is measured in units of $1/J$. (a) Maximum sitewise deviation from an independent truncated-Fock calculation at $n=5$ and $U=8J$. The interacting fixed-sector result agrees to numerical precision, whereas the $U=0J$ noninteracting quadratic baseline shows the effect of omitting the interaction. The inset sketches the hopping $J$, onsite interaction $U$, and operators $\hat n_c$ and $\hat n_i(t)$ entering the commutator. (b,c) $\log_{10}[C_\psi(i,t)+10^{-12}]$ for $n=400$ at $U=0J$ and $U=8J$, respectively, with the central site $c$ marked by the dotted line and a shared color scale. (d) Fitted velocity of the $C_\psi=10^{-5}$ threshold front versus interaction strength for $n=120,200,300,400$. Curves show medians and shaded ranges over three seeded states.}
\label{fig:num-otoc}
\end{figure*}

\subsection{Topological certification and edge dynamics of bound pairs}
\label{ssec:num-edge}

Interacting photons in synthetic gauge fields are now experimentally accessible~\cite{Roushan.2017-ChiralGroundstateCurrents}, and Hofstadter bands together with their Chern numbers have been realized and measured with ultracold atoms~\cite{Aidelsburger.2013-RealizationHofstadterHamiltonian,Aidelsburger.2015-MeasuringChernNumber}. Two-body bound states can themselves inherit nontrivial topology~\cite{Gorlach.2017-TopologicalEdgeStates,Okuma.2023-RelationshipTwoparticleTopology}, and the two-particle Hofstadter problem has been studied since~\cite{Barelli.1996-DoubleButterflySpectrum,Barelli.1997-TwoInteractingHofstadter}. Here we extend these studies to the real-time dynamics of bound pairs under open boundary conditions. Our framework gives direct access to observables of boundary-localized states, allowing us to connect the topological bound-pair band to its edge dynamics.

As a specialized two-dimensional application, we therefore consider two photons on a Hofstadter--Hubbard lattice~\cite{Hofstadter.1976-EnergyLevelsWave,Barelli.1996-DoubleButterflySpectrum,Barelli.1997-TwoInteractingHofstadter}. In the Landau gauge used here,
\begin{equation}
\begin{aligned}
\hat H
=&
-J\sum_{x,y}
\left(
\hat a_{x,y}^\dagger\hat a_{x+1,y}
+
e^{2\pi\ii\phi x}
\hat a_{x,y}^\dagger\hat a_{x,y+1}
+\mathrm{h.c.}
\right)\\
& +\frac{U}{2}\sum_j\hat n_j^2 ,
\end{aligned}
\label{eq:num-hofstadter-hubbard}
\end{equation}
where $\phi$ denotes the signed magnetic flux parameter in units of the flux quantum, with its positive direction fixed by the hopping convention in \cref{eq:num-hofstadter-hubbard}. Within $\mathcal H_2$, replacing $\hat n_j^2$ by $\hat n_j(\hat n_j-1)$ changes only
\begin{equation}
\frac{U}{2}\sum_j\hat n_j^2
=
\frac{U}{2}\sum_j\hat n_j(\hat n_j-1)
+U\mathds{1},
\end{equation}
so the interaction differs from the conventional Bose--Hubbard form only by an overall energy shift.

We first examine the spectrum on an open $12\times12$ lattice with $M=144$ modes, dimension $d_{144,2}=10\,440$, and $U=10J$. Sparse shift-invert diagonalization samples $120$ eigenpairs nearest $E=2U$ for fluxes $\phi=0$ and $\phi=1/4$. Writing $\ket{2_j}
=
\frac{(\hat a_j^\dagger)^2}{\sqrt{2}}\ket{0}$,
we characterize an eigenstate $|\psi_\alpha\rangle$ by its total doublon fraction and the fraction of that doublon weight carried by boundary sites,
\begin{equation}
D_\alpha
=
\sum_j
\left|
\langle2_j|\psi_\alpha\rangle
\right|^2,
\qquad
\eta_\alpha
=
\frac{
\sum_{j\in\partial\Lambda}
\left|
\langle2_j|\psi_\alpha\rangle
\right|^2
}{
D_\alpha
}.
\label{eq:num-edge-measures}
\end{equation}
Panel~\ref{fig:num-edge-transport}(a) shows the sampled doublon-like states with $D_\alpha>0.5$, which we classify as edge-localized when $\eta_\alpha>0.6$~\cite{Winkler2006,Valiente2008,Strohmaier.2010-ObservationElasticDoublon}. At nonzero flux, the boundary-weight distribution develops a pronounced upper tail, where $31$ of the $120$ sampled states satisfy $D_\alpha>0.5$ and $\eta_\alpha>0.6$, compared with one at zero flux, even though the median $\eta_\alpha$ decreases from $0.343$ to $0.236$. This concentration of boundary weight into a subset of the doublon spectrum provides finite-size evidence of boundary localization~\cite{Hatsugai.1993-ChernNumberEdge}.

We also determine the topology of the interacting multiplet on an $L\times L$ torus using twisted boundary conditions~\cite{Niu.1985-QuantizedHallConductance},
\begin{equation}
\hat a_{x+L,y}
=
e^{\ii\theta_x}\hat a_{x,y},
\qquad
\hat a_{x,y+L}
=
e^{\ii\theta_y}\hat a_{x,y}.
\end{equation}
At each twist $\bm\theta=(\theta_x,\theta_y)$, we identify the $M$-state doublon multiplet from $M+8$ eigenvectors near $2U$, order the selected states by energy, and retain the lower subband of rank $M/2$. Let $\Psi(\bm\theta)$ be an orthonormal frame for the selected subspace. We evaluate its non-Abelian lattice Chern number \cite{Thouless.1982-QuantizedHallConductance} using the determinant links
\begin{equation}
\mathcal U_\mu(\bm\theta)
=
\frac{
\det[
\Psi(\bm\theta)^\dagger
\Psi(\bm\theta+\Delta_\mu)
]
}{
\left|
\det[
\Psi(\bm\theta)^\dagger
\Psi(\bm\theta+\Delta_\mu)
]
\right|
},
\end{equation}
and the plaquette curvature
\begin{equation}
\begin{aligned}
\mathcal F_{xy}(\bm\theta)
=
\arg\!\left[
\frac{
\mathcal U_x(\bm\theta)
\mathcal U_y(\bm\theta+\Delta_x)
}{
\mathcal U_x(\bm\theta+\Delta_y)
\mathcal U_y(\bm\theta)
}
\right],\quad
C
=
\frac{1}{2\pi}
\sum_{\bm\theta}
\mathcal F_{xy}(\bm\theta).
\end{aligned}
\label{eq:num-doublon-chern}
\end{equation}
This is the determinant-link construction of Ref.~\cite{Fukui2005}, applied to an interacting two-body band as in Ref.~\cite{Iskin.2021-TwobodyProblemMultiband,Iskin.2023-TopologicalTwobodyBands,Alyuruk.2024-ChernNumbersTwoBody}. The ordering in \cref{eq:num-doublon-chern} fixes the positive orientation of the $(\theta_x,\theta_y)$ grid.

Panel~\ref{fig:num-edge-transport}(b) shows the results at $U=10J$ for $L=8,12,16$, $\phi=\pm1/4$, and an $8\times8$ twist grid. Every displayed lower subband remains isolated with minimum gap $\Delta_{\min}\geq0.088J$ and has $|C|=4$, with opposite signs under flux reversal. Calculations on $6\times6$ and $10\times10$ twist grids at $L=8$ give the same integer, and a broader sweep over $U/J=8,10,16$ produces the same result whenever the selected subband remains isolated; we assign no invariant when the separating gap closes.

We finally test whether the topological multiplet supports directional boundary motion~\cite{Bello.2017-SublatticeDynamicsQuantum}. On the open $12\times12$ lattice at $U=10J$ and $\phi=\pm1/4$, we retain eigenstates satisfying $D_\alpha>0.5$ and $\eta_\alpha>0.7$ within $|E_\alpha-E_{\mathrm{med}}|<0.12J$, where $E_{\mathrm{med}}$ is the median energy of the selected edge-localized states. Projecting the local doublon at the midpoint of the left boundary onto this subspace and normalizing gives the initial wave packet. We propagate it over $0\leq Jt\leq260$ and record the doublon density
$\rho_j^{(d)}(t)=|\langle2_j|\psi(t)\rangle|^2$. Its center of mass is
\begin{equation}
\bm r_{\mathrm{cm}}(t)
=
\frac{
\sum_j\bm r_j\rho_j^{(d)}(t)
}{
\sum_j\rho_j^{(d)}(t)
}.
\end{equation}
Defining $z_{\mathrm{cm}}(t)=[x_{\mathrm{cm}}(t)-x_0]+\ii[y_{\mathrm{cm}}(t)-y_0]$ relative to the lattice center and denoting the unwrapped phase of $z_{\mathrm{cm}}(t)$ by $\widetilde\vartheta(t)$, we measure the accumulated winding
\begin{equation}
\nu(t)
=
\frac{
\widetilde\vartheta(t)-\widetilde\vartheta(0)
}{
2\pi
}.
\label{eq:num-edge-winding}
\end{equation}

The doublon fraction remains above $0.917$ for both flux orientations. Figures~\ref{fig:num-edge-transport}(c,d) show that reversing the flux reverses the boundary motion. The final windings are approximately $+1.77$ and $-0.99$ turns for $\phi=+1/4$ and $\phi=-1/4$, respectively. Their magnitudes differ because we reverse the flux while holding the localized initial packet fixed. Thus flux reversal fixes the direction of motion, but not the winding magnitude. Together with the independently computed Chern number, this associates the directional finite-size edge motion with the isolated topological doublon band. Edge-confined doublon dynamics under a magnetic flux has been demonstrated in driven two-dimensional lattices~\cite{Bello.2017-SublatticeDynamicsQuantum}, where chirality was inferred from a ribbon dispersion. Here the invariant is computed directly for the isolated multiplet on the twist torus, and the circulation of an open-lattice wave packet is shown to reverse with the sign of the flux. 

\begin{figure*}[htbp]
\centering
\includegraphics[width=\textwidth]{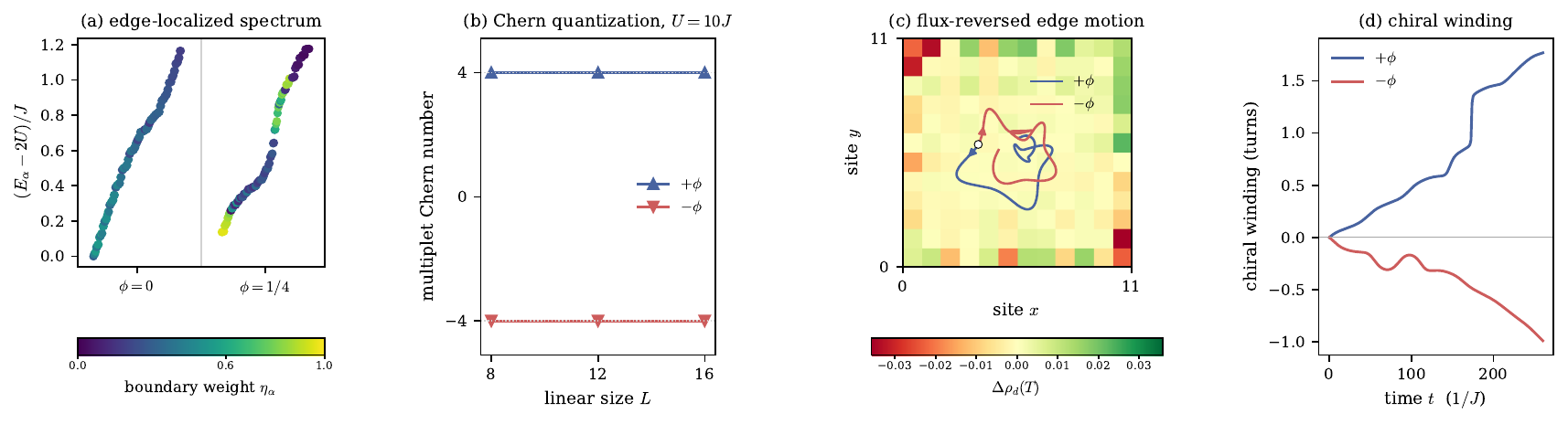}
\caption{ Topological doublon edge dynamics in the two-photon Hofstadter--Hubbard model. (a) Energy-ordered doublon spectra at $\phi=0$ and $\phi=1/4$, colored by their boundary weight $\eta_\alpha$. (b) Chern number of the isolated doublon multiplet for $L=8,12,16$ and $\phi=\pm1/4$. All displayed multiplets remain gapped, with $\Delta_{\min}\geq0.088J$. (c) Difference $\Delta\rho_d(T)=\rho_d^{(+)}(T)-\rho_d^{(-)}(T)$ between the final doublon densities for opposite flux orientations. The curves trace the corresponding center-of-mass motion, the arrows indicate propagation direction, and the white marker denotes the initial packet. (d) Accumulated winding for the two flux orientations.}
\label{fig:num-edge-transport}
\end{figure*}

\section{Discussion and outlook}
\label{sec:discussion}
This work extends Lie-algebraic classical simulation to bosonic systems by identifying the object that actually governs forward mean-value simulation: the reachable operator module of the requested observable. The resulting criterion separates the complexity of observable propagation from both the dimension of the physical Hilbert space and the dimension of the full dynamical Lie algebra. This distinction is particularly consequential for bosons, where the physical representation is infinite-dimensional, yet the Heisenberg orbit of a relevant observable may remain finite and polynomially tractable. Conversely, our boundary examples show that even a finite generator algebra may induce an infinite observable orbit. Exact simulability is therefore a property of the dynamics-observable pair rather than of the Hamiltonians alone.

Several individual ingredients used here are established techniques. Symplectic propagation is the standard language of Gaussian optics~\cite{Bartlett2002,Wang.2007-QIGaussian,Weedbrook2012}, and restriction to a conserved few-particle sector is familiar from many-body physics. Their role in the present work is not as isolated constructions, but as instances of a common simulation principle. Gaussian dynamics closes because it preserves polynomial degree, whilst number-conserving non-Gaussian dynamics closes because it preserves bounded photon-number sectors. Meanwhile nilpotent phase dynamics closes on suitable bounded-degree quadrature observables. The reachable-module criterion places these mechanisms on the same footing, states their computational requirements explicitly, and places their expectation values, correlation functions, and gradients within a common observable-module criterion. In particular, it identifies exact polynomial-time mean-value simulation under non-Gaussian Hamiltonian flow at bounded photon number, including arbitrary-strength Kerr and pair-hopping interactions and coherent support across finitely many bounded sectors.

Numerically, we verify the exact reductions against independent references, confirm the squeezing-band error orders, and apply the finite-sector construction to doublon co-tunnelling, state-dependent OTOCs on up to $400$ modes, and Chern-certified bound-pair edge dynamics.

Several directions follow naturally from these results. 
The perturbative squeezing construction suggests extending exact closure to \emph{approximately closed reachable modules}, equipped with computable a priori or a posteriori leakage estimates and adaptive truncation rules determined by a target accuracy, evolution time, or interaction strength. Recent work has shown how Lie-algebraic commutator structure controls truncation errors in finite-order Magnus approximations to time-dependent Liouville-space dynamics~\cite{ding2026liealgebraicapproachnonmarkovianquantum}. Adapting such algebra-aware error analyses to this projection error, particularly for unbounded bosonic generators, could turn successive-band convergence into a certified accuracy criterion for broader classes of non-Gaussian dynamics without relying on arbitrary local Fock-space cutoffs.

The efficient gradient propagation available within this framework also motivates applications to variational bosonic control, circuit and pulse optimization, and Hamiltonian learning. In these settings, the module structure may help determine which parameters are identifiable from a chosen set of observables and which models remain efficiently differentiable and classically verifiable. Together with recent results linking model concentration and classical simulability in passive linear optics~\cite{monbroussou2026classicalsimulationmodelconcentration}, our reachable-module and gradient constructions provide a natural starting point for asking whether analogous trainability--simulability relations, including barren-plateau behavior, persist for active Gaussian circuits and the non-Gaussian families identified here.

A central theoretical challenge is to classify generator-observable pairs with polynomial-dimensional reachable modules. A concrete step in this direction is to determine the maximal solvable extensions of the nilpotent phase family, including which additional quadratic generators preserve polynomial-dimensional reachable modules. More generally, such a classification should encompass the mechanisms identified here, including degree preservation, conserved-sector decompositions, nilpotent adjoint actions, and symmetry-adapted representations, while also accounting for the computational accessibility of module bases, generator actions, and input overlaps. This could lead to algorithms that discover tractable operator spaces directly from generators and observables without first constructing the full DLA. Finally, the same perspective may extend beyond unitary dynamics by replacing commutator actions with adjoint Liouvillian or quantum-channel actions, opening a route toward exact or controlled simulation of driven-dissipative bosonic systems.

To conclude, we have shown that reachable operator modules provide a problem-adapted criterion for exact bosonic mean-value simulation. They resolve the obstruction posed by infinite-dimensional representations, unify Gaussian and sector-based simulation mechanisms, and identify non-Gaussian regimes with exact correlator dynamics. Beyond the specific families constructed here, this work establishes a basis for discovering, differentiating, and systematically approximating further classes of classically tractable quantum dynamics.

\section*{Code and data availability}
The reference implementation for reproducing all experiment configurations and figures in this paper is openly available at \url{https://github.com/TimothyHeightman/bosonic-gsim}.

\section*{Acknowledgements}
The authors thank the organizers and participants of the 7th Seefeld Workshop on Quantum Information (in particular Isabel Moreno), where some of the central ideas of this project were developed. We would also like to thank Jose Martinez for his help proof-reading the manuscript.
T.H.\ and A.A.\ acknowledge support from the Government of Spain (Severo Ochoa CEX2019-000910-S, FUNQIP, and European Union NextGenerationEU PRTR-C17.I1), the European Union (PASQuanS2.1), Fundaci\'o Cellex, Fundaci\'o Mir-Puig, Generalitat de Catalunya (CERCA program), and the AXA Chair in Quantum Information Science.
A.B.\ acknowledges funding from the  Federal Ministry of Research, Technology and Space of Germany (BMFTR) under grant 13N17231 (HoliQC2), and generous support from the Munich Quantum Valley, which is funded by the Bavarian state government through the Hightech Agenda Bavaria.
Part of the numerical experiments and runtime benchmarks were carried out on the LICCA HPC cluster of the University of Augsburg, co-funded by the Deutsche Forschungsgemeinschaft (DFG, German Research Foundation), Project ID 499211671.


\clearpage
\onecolumngrid
\appendix

\section{Lie-algebra modules and reachable operator spaces}
\label{app:module-background}

The adjoint-space formulation of $\gsim$ represents Heisenberg evolution as a linear action of the dynamical Lie algebra on a space of observables~\cite{Goh.2025-LiealgebraicClassicalSimulations}. We briefly recall the corresponding module terminology and explain why the reachable operator module need not itself be an operator algebra.

Let $\mathfrak g$ be the real DLA generated by the Hermitian operators $H_1,\ldots,H_K$, with the Hermitian bracket $ [X,Y]_{\mathfrak g}:=\ii[X,Y]$ of \cref{eq:hermitian-ad}.
Let $\mathcal A$ denote a real vector space of Hermitian operators on which the relevant commutators are defined. The adjoint action of \cref{eq:adphi-definition} associates with every $X\in\mathfrak g$ the linear map
\begin{equation}
\rho(X):\mathcal A\longrightarrow\mathcal A,
\qquad
\rho(X)(A):=\mathrm{ad}_X(A)=\ii[X,A].
\label{eq:module-adjoint-action}
\end{equation}
The Jacobi identity gives
\begin{equation}
[\rho(X),\rho(Y)]
=
\rho\!\left([X,Y]_{\mathfrak g}\right),
\label{eq:module-representation}
\end{equation}
where the bracket on the left is the commutator of linear maps. Thus, $\rho$ is a representation of $\mathfrak g$ on $\mathcal A$. In representation-theoretic language, $\mathcal A$ is a $\mathfrak g$-module.
A subspace $\mathcal W\subseteq\mathcal A$ is a $\mathfrak g$-submodule if $\mathrm{ad}_X(\mathcal W)\subseteq\mathcal W$ for every $X\in\mathfrak g$.
It suffices to verify this condition for the circuit generators $H_k$. Indeed, invariance under their adjoint maps implies invariance under linear combinations and commutators of these maps, and hence under the action of every element of the generated Lie algebra. The reachable operator module $\mathcal V(O)$ of \cref{def:reach} is therefore the smallest $\mathfrak g$-submodule containing $O$. It is sometimes called the cyclic submodule generated by $O$.

This construction resembles the nested-commutator definition of the DLA~\eqref{eq:DLA}, but the two objects have different seeds and closure requirements. In \cref{eq:DLA}, the adjoint map is applied to all circuit generators $H_\beta$. By contrast, $\mathcal V(O)$ starts from the target observable and need not contain any generator. If $O\in\mathfrak g$, then $\mathcal V(O)$ is the ideal of $\mathfrak g$ generated by $O$; it equals the full DLA only under additional conditions. If $O\notin\mathfrak g$, the reachable module is instead an invariant subspace of the ambient operator representation.

A Lie-algebra module is not generally a Lie algebra. Module invariance requires $\ii[H_k,A]\in\mathcal V(O)$ for $A\in\mathcal V(O)$, but it imposes no condition on $\ii[A,B]$ for two elements $A,B\in\mathcal V(O)$. Consequently, $\mathcal V(O)$ need not possess internal structure constants. For a basis $\{B_\alpha\}_{\alpha=1}^{D}$, the required data are instead the representation matrices of \cref{eq:module-action-matrices}.
These matrices describe how the generators act on the module and provide the finite-dimensional propagators used in the main text.

A Lie-algebra module should also not be confused with an associative bimodule. If $\mathcal B$ is the associative algebra generated by the Hamiltonians, a $\mathcal B$-bimodule must be closed under separate left and right multiplication,
\begin{equation}
\mathcal B\mathcal W\subseteq\mathcal W,
\qquad
\mathcal W\mathcal B\subseteq\mathcal W.
\end{equation}
The reachable-module construction only requires the differences $\ii(H_kA-AH_k)$ to remain in the space. It therefore imposes a weaker and more directly problem-specific closure condition.

For example, consider one bosonic mode with $
H=\frac12\left(\hat x^2+\hat p^2\right)$, and 
$O=\hat x$.
The canonical commutation relation gives
$
\mathrm{ad}_H(\hat x)=\hat p,
$ and $
\mathrm{ad}_H(\hat p)=-\hat x,
$
and hence
\begin{equation}
\mathcal V(\hat x)
=
\spann_{\mathbb R}\{\hat x,\hat p\}.
\end{equation}
This two-dimensional space is invariant under the dynamics and yields
\begin{equation}
e^{t\,\mathrm{ad}_H}(\hat x)
=
\cos(t)\hat x+\sin(t)\hat p.
\end{equation}
Nevertheless, it differs from the one-dimensional DLA $\spann_{\mathbb R}\{H\}$, is not a Lie algebra because $\ii[\hat x,\hat p]=-\mathds{1}\notin\mathcal V(\hat x)$, and is not a bimodule because products such as $H\hat x$ do not remain linear in $\hat x$ and $\hat p$.

The reachable-module criterion therefore asks whether the dynamics--observable pair generates a tractable finite-dimensional representation.

\section{Proofs}
\subsection{Proof of \cref{thm:reach}}
\label{prf:reach}

\begin{proof}
Let $\{B_\alpha\}_{\alpha=1}^{D}$ be the chosen basis of $\mathcal V(O)$ given by \cref{def:reach}. For every spanning element
\begin{equation}
X=
\mathrm{ad}_{H_{k_1}}\cdots
\mathrm{ad}_{H_{k_\ell}}(O)
\end{equation}
and every generator $H_j$, we have
\begin{equation}
\mathrm{ad}_{H_j}(X)
=
\mathrm{ad}_{H_j}
\mathrm{ad}_{H_{k_1}}\cdots
\mathrm{ad}_{H_{k_\ell}}(O)
\in\mathcal V(O).
\end{equation}
Hence each $\mathrm{ad}_{H_j}$ restricts to a linear endomorphism of $\mathcal V(O)$. Its exponential therefore also preserves the module.
Define the action matrices $\mathrm{ad}_{H_k}(B_\alpha)$ as in \cref{eq:module-action-matrices}.
For $X=\sum_\alpha c_\alpha B_\alpha$, let $\bm c^\top$ denote its row vector of coefficients. The Heisenberg orbit generated by one layer satisfies
\begin{equation}
\frac{\dd}{\dd t}X(t)
=
\mathrm{ad}_{H_k}(X(t)),
\qquad
X(0)=X.
\end{equation}
Substituting the basis expansion and using \cref{eq:module-action-matrices} gives
\begin{equation}
\frac{\dd}{\dd t}\bm c^\top(t)
=
\bm c^\top(t)A_k,
\end{equation}
and therefore
\begin{equation}
\bm c^\top(t)
=
\bm c^\top(0)e^{tA_k}.
\label{eq:proof-single-layer}
\end{equation}
Thus the restricted matrix exponential represents the exact Heisenberg action of the layer on $\mathcal V(O)$.
Write $M_\ell:=e^{\theta_\ell A_{k_\ell}}$. With the circuit ordering of \cref{eq:circuit-unitary}, the layers act on the observable coefficients in the order
\begin{equation}
\bm w^\top
\longmapsto
\bm w^\top M_L\cdots M_1.
\end{equation}
Consequently,
\begin{equation}
U(\bm\theta)^\dagger O\,U(\bm\theta)
=
\sum_{\beta=1}^{D}
\left(
\bm w^\top M_L\cdots M_1
\right)_\beta B_\beta.
\end{equation}
Contracting this expansion with the input state and using
$
e_\beta^{\mathrm{in}}
=
\Tr[B_\beta\rho_{\mathrm{in}}]
$
yields
\begin{equation}
\langle O(\bm\theta)\rangle
=
\bm w^\top M_L\cdots M_1\bm e^{\mathrm{in}},
\end{equation}
which is \cref{eq:module-eval}. This reduction is exact because every Heisenberg-evolved observable remains in the invariant module; no operator-space truncation has been introduced.
\end{proof}

\subsection{Proof of \cref{cor:module-correlators}}
\label{prf:module-correlators}

\begin{proof}
For each $j$, the reachable-module expansion of \cref{eq:reach-mod-expansion} is exact.
Expanding the ordered operator product by multilinearity gives
\begin{equation}
\begin{aligned}
O_1^{(1)}\cdots O_m^{(m)}
=
\left(
\sum_{\alpha_1=1}^{D_1}
\widetilde w_{\alpha_1}^{(1)}
B_{\alpha_1}^{(1)}
\right)
\cdots
\left(
\sum_{\alpha_m=1}^{D_m}
\widetilde w_{\alpha_m}^{(m)}
B_{\alpha_m}^{(m)}
\right)
=
\sum_{\alpha_1=1}^{D_1}
\cdots
\sum_{\alpha_m=1}^{D_m}
\left(
\prod_{j=1}^{m}
\widetilde w_{\alpha_j}^{(j)}
\right)
B_{\alpha_1}^{(1)}
\cdots
B_{\alpha_m}^{(m)}.
\end{aligned}
\label{eq:proof-correlator-expansion}
\end{equation}
Only the scalar coefficients have been collected in the product. The order of the basis operators remains unchanged, so \cref{eq:proof-correlator-expansion} also applies when the factors correspond to unequal circuit times or form an out-of-time-ordered product.
Taking the expectation value and using linearity of the trace yields
\begin{equation}
\begin{aligned}
\left\langle
O_1^{(1)}\cdots O_m^{(m)}
\right\rangle
=
\Tr\!\left[
O_1^{(1)}\cdots O_m^{(m)}
\rho_{\mathrm{in}}
\right]
&=
\sum_{\alpha_1=1}^{D_1}
\cdots
\sum_{\alpha_m=1}^{D_m}
\left(
\prod_{j=1}^{m}
\widetilde w_{\alpha_j}^{(j)}
\right)
\Tr\!\left[
B_{\alpha_1}^{(1)}
\cdots
B_{\alpha_m}^{(m)}
\rho_{\mathrm{in}}
\right]
\\
&=
\sum_{\alpha_1,\ldots,\alpha_m}
\left(
\prod_{j=1}^{m}
\widetilde w_{\alpha_j}^{(j)}
\right)
E^{\mathrm{in}}_{\alpha_1\cdots\alpha_m}.
\end{aligned}
\end{equation}
This proves the stated identity without assuming that the basis operators commute or that the input moments factorize.
The contraction contains $\prod_{j=1}^{m}D_j$ terms. If each $D_j$ is polynomial in the system size and $m=O(1)$, then
\begin{equation}
\prod_{j=1}^{m}D_j
=
\operatorname{poly}(n).
\end{equation}
The coefficient vectors are obtained by the individual module propagations, and the remaining quantities are precisely the ordered input moments. Under the stated accessibility assumptions, evaluating all required coefficients and carrying out the contraction therefore takes polynomially many arithmetic operations.
\end{proof}

\subsection{Proof of \cref{prop:moment-tower}}
\label{prf:moment-tower}
\begin{proof}
Let $\widetilde R_0=\mathds{1}$ denote the identity component of the augmented quadrature vector. The affine symplectic evolution of \cref{eq:augmentedS} reads componentwise
\begin{equation}
U(t)^\dagger\widetilde R_\alpha U(t)
=
\sum_{\beta=0}^{2n}
\widetilde S_{\alpha\beta}(t)\widetilde R_\beta.
\label{eq:proof-augmented-component}
\end{equation}
Since unitary conjugation preserves ordered products,
\begin{equation}
U(t)^\dagger
\left(
\widetilde R_{\alpha_1}\cdots
\widetilde R_{\alpha_m}
\right)
U(t)
=
\prod_{\ell=1}^{m}
\left(
U(t)^\dagger
\widetilde R_{\alpha_\ell}
U(t)
\right),
\end{equation}
where the factors retain their original order. Substituting \cref{eq:proof-augmented-component} into every factor gives
\begin{equation}
\begin{aligned}
U(t)^\dagger
\left(
\widetilde R_{\alpha_1}\cdots
\widetilde R_{\alpha_m}
\right)
U(t)
=
\prod_{\ell=1}^{m}
\left(
\sum_{\beta_\ell=0}^{2n}
\widetilde S_{\alpha_\ell\beta_\ell}(t)
\widetilde R_{\beta_\ell}
\right)
=
\sum_{\beta_1,\ldots,\beta_m}
\left(
\prod_{\ell=1}^{m}
\widetilde S_{\alpha_\ell\beta_\ell}(t)
\right)
\widetilde R_{\beta_1}\cdots
\widetilde R_{\beta_m}.
\end{aligned}
\end{equation}
Consequently,
\begin{equation}
\begin{aligned}
U(t)^\dagger O\,U(t)
=
\sum_{\alpha_1,\ldots,\alpha_m}
\sum_{\beta_1,\ldots,\beta_m}
w_{\alpha_1\cdots\alpha_m}
\left(
\prod_{\ell=1}^{m}
\widetilde S_{\alpha_\ell\beta_\ell}(t)
\right)
\widetilde R_{\beta_1}\cdots
\widetilde R_{\beta_m}.
\end{aligned}
\end{equation}
Taking the trace against $\rho_{\mathrm{in}}$ and using the definition of the ordered input moments yields \cref{eq:moment-tower}.

It remains to establish the invariant polynomial space and its dimension. Choose an ordering of the $2n$ quadratures. The canonical commutation relations allow any product to be brought into this order by repeated use of
\begin{equation}
R_iR_j
=
R_jR_i+\ii\Omega_{ij}\mathds{1}.
\label{eq:proof-ccr-reordering}
\end{equation}
Each exchange in \cref{eq:proof-ccr-reordering} produces an ordered term of the same degree and a correction whose degree is lower by two. The ordered monomials
$R_1^{\alpha_1}\cdots R_{2n}^{\alpha_{2n}}$, with $\bm\alpha\in\mathbb N_0^{2n}$,
therefore form the standard Poincaré--Birkhoff--Witt basis of the polynomial Weyl algebra.
For a fixed multi-index $\bm\alpha$, reordering every term in the Weyl-symmetrized monomial $\operatorname{Sym}(R^{\bm\alpha})$ gives
\begin{equation}
\operatorname{Sym}(R^{\bm\alpha})
=
R_1^{\alpha_1}\cdots R_{2n}^{\alpha_{2n}}
+
\text{terms of degree at most }|\bm\alpha|-2.
\label{eq:proof-weyl-triangular}
\end{equation}
Thus the change of basis from ordered monomials to Weyl-symmetrized monomials is triangular with unit diagonal when the monomials are ordered by degree. It is therefore invertible. The operators
\begin{equation}
\left\{
\operatorname{Sym}(R^{\bm\alpha})
:
|\bm\alpha|\leq m
\right\}
\end{equation}
form a basis of the degree-at-most-$m$ polynomial space. Since the quadratures are Hermitian, these symmetrized monomials are Hermitian, and their real span is precisely $\mathcal P_{\leq m}$.
Equation~\eqref{eq:proof-augmented-component} maps each quadrature to a real affine-linear combination of the identity and the quadratures. The conjugate of a product of $q\leq m$ quadratures is therefore a sum of products containing at most $q$ quadratures. Hence
\begin{equation}
U(t)^\dagger\mathcal P_{\leq m}U(t)
\subseteq
\mathcal P_{\leq m}.
\end{equation}
Equivalently, differentiating this finite-dimensional action at $t=0$ shows that the adjoint action of every at-most-quadratic generator preserves $\mathcal P_{\leq m}$. Since $O\in\mathcal P_{\leq m}$, all nested adjoint actions generated from $O$ remain in this space, and therefore $\mathcal V(O)\subseteq\mathcal P_{\leq m}$.
Finally, the number of multi-indices in $2n$ variables with total degree $q$ is
\begin{equation}
\#\left\{
\bm\alpha\in\mathbb N_0^{2n}:|\bm\alpha|=q
\right\}
=
\binom{2n+q-1}{q}.
\end{equation}
Summing over $q=0,\ldots,m$ and applying the hockey-stick identity gives
\begin{equation}
\dim_{\mathbb R}\mathcal P_{\leq m}
=
\sum_{q=0}^{m}
\binom{2n+q-1}{q}
=
\binom{2n+m}{m}.
\end{equation}
The claimed bound on $\dim_{\mathbb R}\mathcal V(O)$ follows from the inclusion above.
\end{proof}

\subsection{Proof of \cref{cor:1-2-moments}}
\label{prf:1-2-moments}

\begin{proof}
Let $\rho(t)=U(t)\rho(0)U(t)^\dagger$  and write the affine Heisenberg evolution componentwise as
\begin{equation}
R_i(t)
:=
U(t)^\dagger R_iU(t)
=
\sum_{k=1}^{2n}S_{ik}(t)R_k-d_i(t)\mathds{1}.
\label{eq:proof-affine-quadratures}
\end{equation}
Using cyclicity of the trace and $\Tr[\rho(0)]=1$, we obtain
\begin{equation}
\begin{aligned}
\overline R_i(t)
=
\Tr[R_i\rho(t)]
=
\Tr[R_i(t)\rho(0)]
=
\sum_{k=1}^{2n}
S_{ik}(t)\Tr[R_k\rho(0)]
-d_i(t)
=
\sum_{k=1}^{2n}
S_{ik}(t)\overline R_k(0)-d_i(t).
\end{aligned}
\end{equation}
Collecting the components gives $\overline{\bm R}(t)$ of \cref{eq:1-2-moments}.
Define the centered Heisenberg quadratures by
\begin{equation}
\delta R_i(t)
:=
R_i(t)-\overline R_i(t)\mathds{1}.
\end{equation}
Substituting the affine transformation and the mean update yields
\begin{equation}
\begin{aligned}
\delta R_i(t)
&=
\sum_{k=1}^{2n}S_{ik}(t)R_k
-d_i(t)\mathds{1}
-
\left[
\sum_{k=1}^{2n}S_{ik}(t)\overline R_k(0)
-d_i(t)
\right]\mathds{1}
\\
&=
\sum_{k=1}^{2n}
S_{ik}(t)
\left(
R_k-\overline R_k(0)\mathds{1}
\right)
=
\sum_{k=1}^{2n}
S_{ik}(t)\delta R_k(0).
\end{aligned}
\label{eq:proof-centered-quadratures}
\end{equation}
The displacement therefore cancels after centering.
Using cyclicity of the trace once more, the covariance at time $t$ can be evaluated in the initial state:
\begin{equation}
\sigma_{ij}(t)
=
\frac12
\Tr\!\left[
\left\{
\delta R_i(t),\delta R_j(t)
\right\}
\rho(0)
\right].
\end{equation}
Substitution of \cref{eq:proof-centered-quadratures} gives
\begin{equation}
\begin{aligned}
\sigma_{ij}(t)
&=
\frac12
\sum_{k,l=1}^{2n}
S_{ik}(t)S_{jl}(t)
\Tr\!\left[
\left\{
\delta R_k(0),\delta R_l(0)
\right\}
\rho(0)
\right]
=
\sum_{k,l=1}^{2n}
S_{ik}(t)\sigma_{kl}(0)S_{jl}(t)
=
\left[
S(t)\sigma(0)S(t)^\top
\right]_{ij},
\end{aligned}
\end{equation}
and hence $\sigma(t)$ of \cref{eq:1-2-moments}.
The derivation uses only the affine Heisenberg action and the existence of the first and second moments; it does not require $\rho(0)$ to be Gaussian.
\end{proof}

\subsection{Proof of \cref{lem:finite-sector}}
\label{prf:finite-sector}

\begin{proof}
Let $|\psi\rangle\in\mathcal H_N$, so that $\hat N|\psi\rangle=N|\psi\rangle$.
Using $[\hat H,\hat N]=0$, we obtain
\begin{equation}
\hat N\hat H|\psi\rangle
=
\hat H\hat N|\psi\rangle
=
N\hat H|\psi\rangle.
\end{equation}
Hence $\hat H|\psi\rangle$ is again an eigenvector of $\hat N$ with eigenvalue $N$, and therefore $\hat H\mathcal H_N\subseteq\mathcal H_N$.
Thus the restriction $\hat H_N:=\hat H|_{\mathcal H_N}$ is well defined.
The occupation-number states
\begin{equation}
|\bm n\rangle
=
|n_1,\ldots,n_n\rangle,
\qquad
n_k\in\mathbb Z_{\geq0},
\qquad
\sum_{k=1}^{n}n_k=N,
\end{equation}
form an orthonormal basis of $\mathcal H_N$. Their number equals the number of non-negative integer solutions of
\begin{equation}
n_1+\cdots+n_n=N.
\end{equation}
By the stars-and-bars argument, such a solution is represented by a sequence of $N$ identical stars and $n-1$ separators. Choosing the separator positions among the $N+n-1$ available positions gives
\begin{equation}
d_{n,N}
=
\binom{N+n-1}{n-1}
=
\binom{N+n-1}{N}.
\end{equation}
Consequently, $\mathcal H_N$ is finite-dimensional and $\hat H_N$ is represented in the occupation-number basis by a $d_{n,N}\times d_{n,N}$ matrix.
The restriction remains Hermitian. Indeed, for any $|\phi\rangle,|\psi\rangle\in\mathcal H_N$,
\begin{equation}
\begin{aligned}
\langle\phi|\hat H_N|\psi\rangle
=
\langle\phi|\hat H|\psi\rangle=
\langle\hat H\phi|\psi\rangle
= \langle\hat H_N\phi|\psi\rangle,
\end{aligned}
\end{equation}
where invariance of $\mathcal H_N$ ensures that $\hat H|\phi\rangle=\hat H_N|\phi\rangle$ remains in the sector. Hence $\hat H_N^\dagger=\hat H_N$.
Because $\hat H=\bigoplus_N\hat H_N$ acts as a bounded Hermitian block on every invariant sector, it is essentially self-adjoint on the finite-particle domain, and $e^{-\ii\hat Ht}$ denotes the unitary group generated by its closure. Since every power of $\hat H$ preserves $\mathcal H_N$, this group restricts on the finite-dimensional sector to the convergent exponential series,
\begin{equation}
\left.e^{-\ii\hat Ht}\right|_{\mathcal H_N}
=
\sum_{r=0}^{\infty}
\frac{(-\ii t)^r}{r!}
\left.\hat H^r\right|_{\mathcal H_N}
=
\sum_{r=0}^{\infty}
\frac{(-\ii t)^r}{r!}
\hat H_N^r
=
e^{-\ii\hat H_Nt}.
\end{equation}
Finally, Hermiticity of $\hat H_N$ gives
\begin{equation}
\left(e^{-\ii\hat H_Nt}\right)^\dagger
e^{-\ii\hat H_Nt}
=
e^{\ii\hat H_Nt}
e^{-\ii\hat H_Nt}
=
\mathds{1}_{\mathcal H_N},
\end{equation}
so $e^{-\ii\hat H_Nt}\in U(d_{n,N})$.
\end{proof}

\subsection{Proof of \cref{thm:boundedN}}
\label{prf:boundedN}

\begin{proof}
Let $P_{\mathcal S}$ be the projector onto $\mathcal H_{\mathcal S}$ of \cref{eq:mixedcarrier}. Number conservation implies $[\hat H_k,\Pi_N]=0$ for every $N$ and $k$, and therefore 
\begin{equation}
[\hat H_k,P_{\mathcal S}]=0.
\label{eq:proof-sector-projector-commutation}
\end{equation}
Consequently, every circuit layer commutes with $P_{\mathcal S}$, i.e., $\left[
e^{-\ii\theta\hat H_k},
P_{\mathcal S}
\right]
=
0$,
and so does any circuit $U$ generated by these layers.

Define the evolved state
\begin{equation}
\rho_U
:=
U\rho_{\mathrm{in}}U^\dagger.
\end{equation}
Using the support condition
$\rho_{\mathrm{in}}=P_{\mathcal S}\rho_{\mathrm{in}}P_{\mathcal S}$
and $[U,P_{\mathcal S}]=0$, we obtain
\begin{equation}
\begin{aligned}
\rho_U
=
UP_{\mathcal S}
\rho_{\mathrm{in}}
P_{\mathcal S}U^\dagger
=
P_{\mathcal S}
U\rho_{\mathrm{in}}U^\dagger
P_{\mathcal S}
=
P_{\mathcal S}\rho_U P_{\mathcal S}.
\end{aligned}
\label{eq:proof-evolved-sector-support}
\end{equation}
Cyclicity of the trace then gives
\begin{equation}
\begin{aligned}
\Tr[O\rho_U]
=
\Tr[
O P_{\mathcal S}\rho_U P_{\mathcal S}
]
=
\Tr[
P_{\mathcal S}OP_{\mathcal S}\rho_U
],
\end{aligned}
\end{equation}
which proves \cref{eq:exact-sector-compression}.

We next show that the compressed operator space is invariant under every generator action. Let $B=P_{\mathcal S}BP_{\mathcal S}$.
Using \cref{eq:proof-sector-projector-commutation}, we find
\begin{equation}
\begin{aligned}
\mathrm{ad}_{\hat H_k}(B)
=
\ii
\left(
\hat H_kP_{\mathcal S}BP_{\mathcal S}
-
P_{\mathcal S}BP_{\mathcal S}\hat H_k
\right)
=
P_{\mathcal S}
\ii[\hat H_k,B]
P_{\mathcal S}.
\end{aligned}
\label{eq:proof-compressed-adjoint}
\end{equation}
Hence $\mathrm{ad}_{\hat H_k}(B)$ is again supported on $\mathcal H_{\mathcal S}$. More explicitly, for the sector block $B_{MN}:=\Pi_MB\Pi_N$, with $M,N\in\mathcal S$, we have
\begin{equation}
\begin{aligned}
\Pi_M\mathrm{ad}_{\hat H_k}(B)\Pi_N
=
\ii\Pi_M
\left(
\hat H_kB-B\hat H_k
\right)
\Pi_N
=
\ii
\left(
\hat H_{k,M}B_{MN}
-
B_{MN}\hat H_{k,N}
\right),
\end{aligned}
\label{eq:sectorpairaction}
\end{equation}
where $\hat H_{k,N}:=\Pi_N\hat H_k\Pi_N$.
Thus each sector-pair space
$\operatorname{Hom}(\mathcal H_N,\mathcal H_M)$ is invariant under the adjoint action.
Because $O_{\mathcal S}=P_{\mathcal S}OP_{\mathcal S}$ is supported on $\mathcal H_{\mathcal S}$, repeated application of \cref{eq:proof-compressed-adjoint} shows that every element generated in its reachable module has the same support. Moreover, since $O_{\mathcal S}$ and the generators are Hermitian,
\begin{equation}
\left(
\ii[\hat H_k,B]
\right)^\dagger
=
\ii[\hat H_k,B]
\end{equation}
whenever $B$ is Hermitian. Therefore
$
\mathcal V(O_{\mathcal S})
\subseteq
\operatorname{Herm}(\mathcal H_{\mathcal S})
\subseteq
\operatorname{End}(\mathcal H_{\mathcal S})$.
Let
\begin{equation}
d_{n,\mathcal S}
=
\dim\mathcal H_{\mathcal S}
=
\sum_{N\in\mathcal S}d_{n,N}.
\end{equation}
The real vector space of Hermitian
$d_{n,\mathcal S}\times d_{n,\mathcal S}$ matrices has dimension
$d_{n,\mathcal S}^2$: it has $d_{n,\mathcal S}$ real diagonal entries and two real parameters for each of its
$\binom{d_{n,\mathcal S}}{2}$ off-diagonal pairs. Hence
\begin{equation}
\dim_{\mathbb R}\mathcal V(O_{\mathcal S})
\leq
d_{n,\mathcal S}^2.
\label{eq:proof-bounded-module-dimension}
\end{equation}

Finally,
\begin{equation}
\begin{aligned}
d_{n,\mathcal S}
=
\sum_{N\in\mathcal S}
\binom{N+n-1}{N}
\leq
\sum_{N=0}^{N_{\max}}
\binom{N+n-1}{N}
=
\binom{n+N_{\max}}{N_{\max}},
\end{aligned}
\end{equation}
where the last equality is the hockey-stick identity. Since
$N_{\max}=\mathcal O(1)$,

$\tbinom{n+N_{\max}}{N_{\max}}
=
\mathcal O(n^{N_{\max}})$,
and therefore
\begin{equation}
\dim_{\mathbb R}\mathcal V(O_{\mathcal S})
=
\mathcal O(n^{2N_{\max}}).
\end{equation}
The restricted generator matrices act on this polynomial-dimensional space, and the reduction is exact because the dynamics never leaves $\mathcal H_{\mathcal S}$. Whenever the compressed observable and input data satisfy the accessibility conditions of \cref{thm:reach}, that theorem gives the corresponding polynomial-time expectation-value simulation.
\end{proof}

\subsection{Proof of \cref{prop:cph}}
\label{prf:cph}

\begin{proof}
Let $\mathfrak g_{\mathrm{cph}}^{(q)}$ be defined as in \cref{eq:cph}.
Since the position operators commute, $\ii[f(\hat{\bm x}),g(\hat{\bm x})]=0$ for all $f,g\in\mathcal P_{\leq q}(\hat{\bm x})$. Likewise, $\ii[\hat p_j,\hat p_k]=0$.
The only nonzero brackets between the two summands follow from the canonical commutation relations:
\begin{equation}
\ii[\hat p_j,f(\hat{\bm x})]
=
\partial_j f(\hat{\bm x}).
\label{eq:proof-cph-derivative}
\end{equation}
Indeed, for a position monomial
$\hat{\bm x}^{\bm\alpha}
=\prod_{k=1}^{n}\hat x_k^{\alpha_k}$,
the Leibniz rule gives
\begin{equation}
\begin{aligned}
\ii[\hat p_j,\hat{\bm x}^{\bm\alpha}]
=
\ii
\left(
\prod_{k\neq j}\hat x_k^{\alpha_k}
\right)
[\hat p_j,\hat x_j^{\alpha_j}]
=
\alpha_j
\hat{\bm x}^{\bm\alpha-\bm e_j}.
\end{aligned}
\end{equation}
Linearity then proves \cref{eq:proof-cph-derivative} for every position polynomial. Since differentiation lowers total degree by one,
\begin{equation}
\partial_j f\in\mathcal P_{\leq q-1}(\hat{\bm x})
\subseteq
\mathcal P_{\leq q}(\hat{\bm x}).
\end{equation}
Thus $\mathfrak g_{\mathrm{cph}}^{(q)}$ is closed under the Hermitian bracket. Bilinearity, antisymmetry, and the Jacobi identity are inherited from the operator commutator, so it is a real Lie algebra.
To establish nilpotency, let
\begin{equation}
\gamma_1
=
\mathfrak g_{\mathrm{cph}}^{(q)},
\qquad
\gamma_{r+1}
=
[\mathfrak g_{\mathrm{cph}}^{(q)},\gamma_r]_{\mathrm H}
\end{equation}
denote its lower central series, where
$[X,Y]_{\mathrm H}:=\ii[X,Y]$. The commutation relations above imply $\gamma_2
=
\mathcal P_{\leq q-1}(\hat{\bm x})$.
The equality holds because differentiating degree-at-most-$q$ polynomials produces every monomial of degree at most $q-1$: for any $\bm\alpha$ with $|\bm\alpha|\leq q-1$,
\begin{equation}
\hat{\bm x}^{\bm\alpha}
=
\frac{1}{\alpha_j+1}
\partial_j
\left(
\hat{\bm x}^{\bm\alpha+\bm e_j}
\right)
\end{equation}
for any chosen $j$. Repeated use of \cref{eq:proof-cph-derivative} therefore gives
\begin{equation}
\gamma_r
=
\mathcal P_{\leq q-r+1}(\hat{\bm x}),
\qquad
2\leq r\leq q+1.
\end{equation}
In particular, $\gamma_{q+1}
=
\mathcal P_{\leq0}(\hat{\bm x})$, and $\gamma_{q+2}=0$, 
because constant operators commute with every element. Hence
$\mathfrak g_{\mathrm{cph}}^{(q)}$ is nilpotent, with nilpotency class at most $q+1$.

It remains to count its dimension. The number of monomials in $n$ commuting variables with total degree exactly $r$ is $\tbinom{n+r-1}{r}$.
Consequently,
\begin{equation}
\begin{aligned}
\dim\mathcal P_{\leq q}(\hat{\bm x})
=
\sum_{r=0}^{q}
\binom{n+r-1}{r}
=
\binom{n+q}{q}.
\end{aligned}
\end{equation}
The momentum operators are linearly independent of the position-polynomial sector, so
\begin{equation}
\dim\mathfrak g_{\mathrm{cph}}^{(q)}
=
n+\binom{n+q}{q}.
\end{equation}
For fixed $q$, the binomial coefficient scales as $\mathcal O(n^q)$. When $q\geq3$, the position-polynomial summand contains the cubic generators $\hat x_k^3$ and $\hat x_k\hat x_l\hat x_m$.
\end{proof}

\subsection{Proof of \cref{prop:cph-reach}}
\label{prf:cph-reach}

\begin{proof}
For multi-indices
$\bm\alpha,\bm\beta\in\mathbb N_0^n$, write the normally ordered monomials as
$
M_{\bm\alpha,\bm\beta}
:=
\hat{\bm x}^{\bm\alpha}
\hat{\bm p}^{\bm\beta}$.
Introduce the weighted degree
\begin{equation}
\mathrm{wt}_q(\bm\alpha,\bm\beta)
:=
|\bm\alpha|+(q-1)|\bm\beta|.
\end{equation}
By definition,
\begin{equation}
\mathcal F_{m,q}
=
\spann_{\mathbb C}
\left\{
M_{\bm\alpha,\bm\beta}:
\mathrm{wt}_q(\bm\alpha,\bm\beta)
\leq m(q-1)
\right\}
\end{equation}
from \cref{eq:cph-filtration}.
We prove that this space is invariant under the two types of generators spanning $\mathfrak g_{\mathrm{cph}}^{(q)}$.
For a momentum generator, the canonical commutation relations and the Leibniz rule give
\begin{equation}
\begin{aligned}
\mathrm{ad}_{\hat p_j}
\left(
M_{\bm\alpha,\bm\beta}
\right)
=
\ii
\left[
\hat p_j,
\hat{\bm x}^{\bm\alpha}
\hat{\bm p}^{\bm\beta}
\right]
=
\alpha_j
\hat{\bm x}^{\bm\alpha-\bm e_j}
\hat{\bm p}^{\bm\beta}.
\end{aligned}
\label{eq:proof-cph-momentum-action}
\end{equation}
Whenever this expression is nonzero, its weighted degree is
\begin{equation}
\mathrm{wt}_q(\bm\alpha-\bm e_j,\bm\beta)
=
\mathrm{wt}_q(\bm\alpha,\bm\beta)-1.
\end{equation}
Thus momentum generators preserve $\mathcal F_{m,q}$.

Now let $V(\hat{\bm x})\in\mathcal P_{\leq q}(\hat{\bm x})$. Repeatedly using
\begin{equation}
\hat p_j f(\hat{\bm x})
=
f(\hat{\bm x})\hat p_j
-\ii\,\partial_jf(\hat{\bm x})
\end{equation}
gives the multi-index normal-ordering identity
\begin{equation}
\hat{\bm p}^{\bm\beta}V(\hat{\bm x})
=
\sum_{\bm 0\leq\bm\gamma\leq\bm\beta}
\binom{\bm\beta}{\bm\gamma}
(-\ii)^{|\bm\gamma|}
\bigl(\partial^{\bm\gamma}V\bigr)(\hat{\bm x})
\hat{\bm p}^{\bm\beta-\bm\gamma}.
\label{eq:proof-cph-normal-ordering}
\end{equation}
Here $\tbinom{\bm\beta}{\bm\gamma}
=
\prod_{j=1}^{n}
\binom{\beta_j}{\gamma_j}$.
\cref{eq:proof-cph-normal-ordering} follows by induction on $|\bm\beta|$: the case $|\bm\beta|=0$ is immediate, and multiplying by one additional $\hat p_j$ produces the two terms required by the multi-index Pascal identity.
Since $V(\hat{\bm x})$ commutes with $\hat{\bm x}^{\bm\alpha}$, subtracting the $\bm\gamma=\bm0$ term in \cref{eq:proof-cph-normal-ordering} yields
\begin{equation}
\begin{aligned}
\mathrm{ad}_{V(\hat{\bm x})}
\left(
M_{\bm\alpha,\bm\beta}
\right)
=
-\ii
\sum_{\substack{
\bm0<\bm\gamma\leq\bm\beta
}}
\binom{\bm\beta}{\bm\gamma}
(-\ii)^{|\bm\gamma|}
\hat{\bm x}^{\bm\alpha}
\bigl(\partial^{\bm\gamma}V\bigr)(\hat{\bm x})
\hat{\bm p}^{\bm\beta-\bm\gamma}.
\end{aligned}
\label{eq:cph-normal-commutator}
\end{equation}
Because $\deg V\leq q$,
\begin{equation}
\deg\bigl(\partial^{\bm\gamma}V\bigr)
\leq
q-|\bm\gamma|
\end{equation}
for every nonzero derivative. Each summand in
\cref{eq:cph-normal-commutator} therefore has weighted degree at most
\begin{equation}
\begin{aligned}
|\bm\alpha|
+q-|\bm\gamma|
+(q-1)
\left(
|\bm\beta|-|\bm\gamma|
\right)
=
\mathrm{wt}_q(\bm\alpha,\bm\beta)
+q-q|\bm\gamma|
\leq
\mathrm{wt}_q(\bm\alpha,\bm\beta),
\end{aligned}
\end{equation}
because $|\bm\gamma|\geq1$. Position-polynomial generators therefore also preserve $\mathcal F_{m,q}$. By linearity, the same holds for every generator in $\mathfrak g_{\mathrm{cph}}^{(q)}$.

We next verify that the initial observable lies in this filtration. Normal ordering uses
\begin{equation}
\hat p_j\hat x_k
=
\hat x_k\hat p_j-\ii\delta_{jk}\mathds{1}.
\end{equation}
Each exchange either preserves the ordinary total degree or produces a correction of degree lower by two. Hence normal ordering a polynomial of total degree at most $m$ produces only monomials satisfying $|\bm\alpha|+|\bm\beta|\leq m$.
For $q\geq2$,
\begin{equation}
\begin{aligned}
\mathrm{wt}_q(\bm\alpha,\bm\beta)
=
|\bm\alpha|+(q-1)|\bm\beta|
\leq
(q-1)
\left(
|\bm\alpha|+|\bm\beta|
\right)
\leq
m(q-1).
\end{aligned}
\end{equation}
Thus $O\in\mathcal F_{m,q}$. Since this space is invariant under every generator action, all nested adjoint actions generated from $O$ remain in it, proving $
\mathcal V_{\mathbb C}(O)
\subseteq
\mathcal F_{m,q}$, with $\mathcal V_{\mathbb C}(O):=\mathcal V(O)\otimes_{\mathbb R}\mathbb C$.

Finally, the normally ordered monomials are linearly independent. Fixing
$b=|\bm\beta|$, there are
$
\tbinom{n+b-1}{b}
$
possible momentum multi-indices. The weighted-degree constraint then requires
\begin{equation}
|\bm\alpha|
\leq
(m-b)(q-1),
\end{equation}
which allows
\begin{equation}
\binom{n+(m-b)(q-1)}{(m-b)(q-1)}
\end{equation}
position multi-indices. Hence
\begin{equation}
\dim_{\mathbb C}\mathcal F_{m,q}
=
\sum_{b=0}^{m}
\binom{n+b-1}{b}
\binom{n+(m-b)(q-1)}{(m-b)(q-1)}.
\label{eq:cph-filtration-dimension}
\end{equation}
For fixed $m$ and $q$, the summand indexed by $b$ scales as
\begin{equation}
\mathcal O\!\left(
n^{m(q-1)-b(q-2)}
\right).
\end{equation}
Since $q\geq3$, the leading term is $b=0$, and therefore
$
\dim_{\mathbb C}\mathcal F_{m,q}
=
\Theta\!\left(n^{m(q-1)}\right)$.
Together with
$\dim_{\mathbb R}\mathcal V(O)
=
\dim_{\mathbb C}\mathcal V_{\mathbb C}(O)$,
this proves \cref{eq:cph-reachable-dimension}.
\end{proof}

\section{Product observables and gradients for reachable modules}
\label{app:module-products-gradients}

We spell out the fixed-order correlator and gradient formulas used in \cref{sec:finiteopertormodules}. The equal-time product-observable and reverse-mode constructions below follow the invariant-subspace formulation of Lie-algebraic simulation in Ref.~\cite{Goh.2025-LiealgebraicClassicalSimulations}. We restate them in reachable-module notation and extend the product construction to operator insertions at unequal circuit times.

\subsection{Equal-time product observables}
For fixed $m=\mathcal O(1)$, let $\mathcal V_j=\mathcal V(O_j)$ have basis $\{B^{(j)}_{\alpha}\}_{\alpha=1}^{D_j}$ and module-action matrices $A^{(j)}_k$ defined by
\begin{equation}
  \mathrm{ad}_{H_k}(B^{(j)}_{\alpha})
  =
  \sum_{\beta=1}^{D_j}
  (A^{(j)}_k)_{\alpha\beta}B^{(j)}_{\beta}.
\end{equation}
Writing $O_j=\sum_\alpha w^{(j)}_\alpha B^{(j)}_\alpha$, the Heisenberg coefficients after the circuit are
\begin{equation}
\widetilde{\bm w}^{(j)\top}
=
\bm w^{(j)\top}
e^{\theta_LA^{(j)}_{k_L}}\cdots
e^{\theta_1A^{(j)}_{k_1}}.
\end{equation}
Therefore
\begin{equation}
  \langle O_1(\bm\theta)\cdots O_m(\bm\theta)\rangle
  =
  \sum_{\alpha_1,\ldots,\alpha_m}
  \widetilde w^{(1)}_{\alpha_1}\cdots
  \widetilde w^{(m)}_{\alpha_m}
  E^{\mathrm{in}}_{\alpha_1\cdots\alpha_m},
  \label{eq:module-product-eval}
\end{equation}
where
\begin{equation}
  E^{\mathrm{in}}_{\alpha_1\cdots\alpha_m}
  =
  \Tr[
    B^{(1)}_{\alpha_1}\cdots B^{(m)}_{\alpha_m}
    \rho_{\mathrm{in}}
  ] .
\end{equation}
For fixed $m$, this evaluation is polynomial whenever all $D_j=\poly(n)$, the individual module propagations are efficient, and the required contraction with $E^{\mathrm{in}}$ can be performed in polynomial time.

\subsection{Unequal-time correlation functions}
The same construction applies when the factors are inserted at different circuit times. Let $U(t_j)$ denote the circuit prefix up to $t_j$ and define
\begin{equation}
O_j(t_j)=U(t_j)^\dagger O_jU(t_j)
=\sum_{\alpha_j}\widetilde w_{\alpha_j}^{(j)}(t_j)B_{\alpha_j}^{(j)}.
\end{equation}
Each coefficient vector is propagated only through the corresponding circuit prefix. Consequently,
\begin{equation}
\left\langle O_1(t_1)\cdots O_m(t_m)\right\rangle
=\sum_{\alpha_1,\ldots,\alpha_m}
\prod_{j=1}^m\widetilde w_{\alpha_j}^{(j)}(t_j)
E_{\alpha_1\cdots\alpha_m}^{\mathrm{in}}.
\end{equation}
The operator order is retained in the input tensor, so the same formula covers out-of-time-ordered products without chronological reordering.

\subsection{Reverse-mode gradients}
The following recursion is the reachable-module form of the adjoint or reverse-mode differentiation used in Refs.~\cite{Goh.2025-LiealgebraicClassicalSimulations,Khaneja.2005-OptimalControlCoupled}.

For gradients, consider first a single observable with module propagators
\begin{equation}
  R_\ell(\theta_\ell)=e^{\theta_\ell A_{k_\ell}},
  \qquad
  F(\bm\theta)
  =
  \bm w^{\top}R_L\cdots R_1\bm e^{\mathrm{in}} .
\end{equation}
Define forward vectors and backward covectors by
\begin{equation}
  \bm v_0=\bm e^{\mathrm{in}},
  \qquad
  \bm v_\ell=R_\ell\bm v_{\ell-1},
\end{equation}
and
\begin{equation}
  \bm \lambda_L=\bm w,
  \qquad
  \bm \lambda_{\ell-1}=R_\ell^{\top}\bm\lambda_\ell .
\end{equation}
Assume that $\bm w$ and $\bm e^{\mathrm{in}}$ are independent of $\bm\theta$.
Then
\begin{equation}
  \partial_{\theta_\ell}F
  =
  \bm\lambda_\ell^{\top}
  A_{k_\ell}
  \bm v_\ell,
  \label{eq:module-gradient}
\end{equation}
using $\partial_{\theta_\ell}R_\ell=A_{k_\ell}R_\ell$. If the same parameter occurs in several circuit layers, its derivative is the sum of the corresponding layer contributions. Thus all first derivatives are obtained by one forward pass and one backward pass through the same finite-dimensional module propagators. For a product correlator, the derivative is the sum of the contributions obtained by inserting $A_{k_\ell}^{(j)}$ on each active tensor leg; equivalently, the generator on the product space is the corresponding Kronecker sum. For an unequal-time product, differentiation with respect to a circuit parameter acts only on those tensor legs whose circuit prefix contains the corresponding layer. Thus a parameter appearing after the insertion time of a factor does not contribute to the derivative of that factor.

\section{Details for the Gaussian degree-truncation family}
\label{app:gaussian-degree-details}

\subsection{Standard Gaussian generators}
\label{app:gaussian-gates}

The following standard Gaussian generators fix our phase and squeezing conventions~\cite{Weedbrook2012} and provide elementary tests for \cref{sec:heisenbergweyl}. 
For a single mode, the phase rotation
\begin{equation}
  H=\omega\,\hat a^\dagger\hat a
\end{equation}
gives
\begin{equation}
  \hat a(t)=e^{-\ii\omega t}\hat a,
  \qquad
  \hat a^\dagger(t)=e^{\ii\omega t}\hat a^\dagger,
\end{equation}
corresponding to a compact $U(1)$ block of the Bogoliubov matrix. Single-mode squeezing,
\begin{equation}
  H=\frac{\ii}{2}\left(\xi^*\hat a^2-\xi\hat a^{\dagger 2}\right),
  \qquad
  \xi=|\xi|e^{\ii\phi},
\end{equation}
gives
\begin{equation}
  \hat a(t)
  =
  \cosh(|\xi|t)\hat a
  -
  e^{\ii\phi}\sinh(|\xi|t)\hat a^\dagger,
\end{equation}
a non-compact hyperbolic symplectic transformation. The beamsplitter Hamiltonian
\begin{equation}
H_{\mathrm{BS}}
=
g\left(\hat a^\dagger\hat b+\hat a\hat b^\dagger\right)
\end{equation}
generates passive mode mixing. Together with phase rotations, such generators span the passive subgroup
$\Sp(2n,\mathbb R)\cap O(2n)\cong U(n)$.

\subsection{Cumulants, product inputs, and fixed-order cost}
\label{app:gaussian-cumulants}

The moment tensor $E^{\mathrm{in}}$ in \eqref{eq:input-moment-tensor} is computed once for the input state. For product inputs $\rho_{\mathrm{in}}=\bigotimes_k\rho_k$, these moments factor over modes. For example, for a single-mode Fock state $\ket m$,
\begin{equation}
  \langle m|\hat a^{\dagger p}\hat a^q|m\rangle
  =
  \delta_{pq}\frac{m!}{(m-p)!},
  \qquad
  p\le m,
  \label{eq:fock-moments}
\end{equation}
where the displayed value applies for $p=q\leq m$, and the moment vanishes otherwise. Coherent, thermal, squeezed, cat, and finite-energy approximate GKP inputs can be treated analogously whenever their required fixed-order quadrature moments can be computed in polynomial time.

For contraction, it is often preferable not to store the full tensor. With a fixed operator ordering, decompose moments into connected cumulants by the partition relation
\begin{equation}
  E^{\mathrm{in}}_{i_1\dots i_m}
  =
  \sum_{\pi\in\mathfrak P_m}
  \prod_{b\in\pi}
  \kappa^{\mathrm{in}}_{|b|}
  (i_j:j\in b),
  \label{eq:moment-cumulant}
\end{equation}
where $\mathfrak P_m$ is the set of partitions of $\{1,\dots,m\}$. Because each tensor leg is propagated by the same linear map, the cumulants propagate blockwise:
\begin{equation}
  \kappa^{(t)}_{|b|}(j_\ell:\ell\in b)
  =
  \sum_{i_\ell:\ell\in b}
  \prod_{\ell\in b}
  \widetilde S_{j_\ell i_\ell}(t)
  \kappa^{\mathrm{in}}_{|b|}(i_\ell:\ell\in b).
  \label{eq:cumulant-propagation}
\end{equation}
For a block $b=\{\ell_1<\cdots<\ell_r\}$, the arguments of
$\kappa_r$ retain their original operator order. Thus
\eqref{eq:moment-tower} can be written as
\begin{equation}
\langle O(t)\rangle
=
\sum_{i_1,\ldots,i_m}
w_{i_1\cdots i_m}
\sum_{\pi\in\mathfrak P_m}
\prod_{b\in\pi}
\kappa^{(t)}_{|b|}
(i_\ell:\ell\in b).
\label{eq:linked}
\end{equation}

For Gaussian input states all cumulants of order $p\ge3$ vanish, so \eqref{eq:linked} reduces to Wick's theorem with propagated mean and covariance. Higher cumulants therefore capture deviations from Gaussian statistics. A fixed-order calculation need not distinguish every non-Gaussian state, but any nonzero cumulant of order $p\geq3$ certifies non-Gaussianity.
For fixed $m$, the number of partition terms is the Bell number $B_m=\mathcal O(1)$, and product structure or sparsity in $w$ avoids storing the full $(2n+1)^m$ tensor.

\subsection{Multi-time correlators and gradients}
\label{app:gaussian-multitime-gradients}

The same tensor-leg formula gives unequal-time correlators. For linear factors,
\begin{equation}
  \left\langle
    \widetilde{\Rvec}_{i_1}(t_1)\cdots
    \widetilde{\Rvec}_{i_m}(t_m)
  \right\rangle
  =
  \sum_{j_1,\dots,j_m}
  \prod_{\ell=1}^m
  \widetilde S_{i_\ell j_\ell}(t_\ell)
  E^{\mathrm{in}}_{j_1\dots j_m}.
  \label{eq:multitime}
\end{equation}
Equal-time moments are the special case $t_1=\cdots=t_m$. Out-of-time-order correlators are handled in the same way, provided the operator order in the input tensor is kept fixed.

For variational Gaussian circuits, let
\begin{equation}
  \widetilde S(\bm\theta)
  =
  \widetilde S_L(\theta_L)\cdots \widetilde S_1(\theta_1).
\end{equation}
Then a parameter derivative of \cref{eq:moment-tower} inserts one differentiated transfer matrix into the corresponding tensor product:
\begin{equation}
  \partial_{\theta_r}\langle O\rangle
  =
  \left\langle
    w,
    \widetilde S_L^{\otimes m}\cdots
    \partial_{\theta_r}\!\left(\widetilde S_r^{\otimes m}\right)
    \cdots
    \widetilde S_1^{\otimes m}
    E^{\mathrm{in}}
  \right\rangle,
  \label{eq:gradient-chain}
\end{equation}
with
\begin{equation}
  \partial_\theta\left(\widetilde S^{\otimes m}\right)
  =
  \sum_{\ell=1}^m
  \widetilde S^{\otimes(\ell-1)}
  \otimes
  \partial_\theta\widetilde S
  \otimes
  \widetilde S^{\otimes(m-\ell)}.
  \label{eq:tensor-gradient}
\end{equation}
Here, $\langle\cdot,\cdot\rangle$ denotes the bilinear contraction over all tensor indices.
For fixed $m$, reverse accumulation evaluates all layer derivatives with the same asymptotic scaling as the forward contraction, provided the differentiated transfer-matrix actions are efficiently computable.

\subsection{Bounded-order correlators and the boson-sampling boundary}
\label{app:bounded-order-boundary}

For passive linear optics, write the output modes as
\begin{equation}
  \hat b_i=\sum_k W_{ik}\hat a_k,
  \qquad
  W\in U(n).
\end{equation}
On a product number input $\ket{\bm n}$, the first-order photon-number moment is
\begin{equation}
  \langle \hat n_i\rangle
  =
  \sum_k n_k |W_{ik}|^2 .
  \label{eq:first-number-moment}
\end{equation}
For $i\ne j$, the second-order correlator is
\begin{equation}
\begin{aligned}
  \langle \hat n_i\hat n_j\rangle
  =
  \left(\sum_k n_k|W_{ik}|^2\right)
  \left(\sum_l n_l|W_{jl}|^2\right)
  +
  \left|\sum_k n_k W_{ik}^*W_{jk}\right|^2
  -
  \sum_k n_k(n_k+1)|W_{ik}|^2|W_{jk}|^2 .
\end{aligned}
\label{eq:ninj-general}
\end{equation}
For the single-photon-per-mode input $\ket{1}^{\otimes n}$ this reduces to
\begin{equation}
  \langle \hat n_i\hat n_j\rangle
  =
  \left(\sum_k |W_{ik}|^2\right)
  \left(\sum_l |W_{jl}|^2\right)
  +
  \left|\sum_k W_{ik}^*W_{jk}\right|^2
  -
  2\sum_k |W_{ik}|^2|W_{jk}|^2 .
  \label{eq:ninj}
\end{equation}
For a two-mode beam splitter of angle $\theta$, \eqref{eq:ninj} gives the Hong--Ou--Mandel coincidence~\cite{Hong.1987}
\begin{equation}
  \langle \hat n_1\hat n_2\rangle
  =
  \cos^2(2\theta),
\end{equation}
which vanishes at a $50/50$ beam splitter. Higher fixed-order photon-number correlators are obtained by the same moment or cumulant contraction and remain polynomial in $n$ for fixed order.

For collision-free input and output patterns $S$ and $T$ containing the same number of photons, the transition probability is
\begin{equation}
P(T|S)
=
\left|\operatorname{Per}(W_{T,S})\right|^2,
\end{equation}
where $W_{T,S}$ is the corresponding interferometer submatrix. For repeated input or output occupations, rows or columns are repeated and the probability includes the usual product of input and output factorial normalizations. For fixed photon number $N=\mathcal O(1)$, these permanents have fixed size and remain efficiently computable. The boson-sampling boundary arises when $N$, and hence the relevant correlation order, grows with system size. Fixed-order correlation evaluation and growing-$N$ output-probability or sampling are therefore distinct problems~\cite{Aaronson2013}.

\section{Symmetry-adapted Gaussian reductions}
\label{app:gaussian-symmetry}

The reductions below are the bosonic analogs of the translation- and permutation-adapted representations developed for practical $\gsim$ preprocessing in Ref.~\cite{Barligea.2026-EnablingLieAlgebraicClassical}. We write the homogeneous holomorphic evolution as $\frac{\dd}{\dd t}\rvec=\mathcal M\rvec$, so that $e^{t\mathcal M}$ is the Bogoliubov transfer matrix of \cref{eq:bogoliubov}.

\subsection{Translation-invariant Gaussian systems}
\label{app:translation-gaussian}

For a translation-invariant quadratic Hamiltonian on a ring, the passive and pairing blocks in the holomorphic generator are circulant. Let $F$ denote the discrete Fourier transform which diagonalizes the passive block. The pairing block then transforms by congruence and therefore couples momentum $q$ only to $-q$. After ordering the holomorphic basis into these momentum pairs, the Bogoliubov generator decomposes into independent momentum sectors, pairing $q$ only with $-q$ when squeezing terms are present. Equivalently,
\begin{equation}
  \mathcal M
  \cong
  \bigoplus_q \mathcal M_q,
  \label{eq:momentum-blocks}
\end{equation}
where each $\mathcal M_q$ acts on the two-dimensional span of $(\hat a_q,\hat a_{-q}^\dagger)$, where momenta are understood modulo $n$. At the self-conjugate momenta $q=0$ and, for even $n$, $q=n/2$, the block acts on $(\hat a_q,\hat a_q^\dagger)$ and remains two dimensional when squeezing is present. It splits into one-dimensional blocks only when the pairing term vanishes. Hence, the full transfer matrix is obtained by exponentiating these independent low-dimensional blocks and transforming the annihilation and creation components back to the position basis.

This is the bosonic analog of using momentum sectors in translation-invariant free-fermion or Pauli-orbit simulations. It is a representation-adapted speedup of the same Gaussian module, not a separate simulability principle.

\subsection{Permutation-invariant Gaussian systems}
\label{app:permutation-gaussian}

For permutation-invariant Gaussian generators, every passive block has the form
\begin{equation}
  A=\alpha\,\mathds{1}_n+\beta(\bm 1\bm 1^\top-\mathds{1}_n),
\end{equation}
and every pairing block has the analogous form
\begin{equation}
  B=\gamma\,\mathds{1}_n+\delta(\bm 1\bm 1^\top-\mathds{1}_n).
\end{equation}
An orthogonal change of basis separates the collective mode
\begin{equation}
  \hat a_{\mathrm{col}}
  =
  \frac1{\sqrt n}
  \sum_{k=1}^n \hat a_k
\end{equation}
from the $(n-1)$-dimensional orthogonal subspace. The collective eigenvalues of $A$ and $B$ are
$\alpha+(n-1)\beta$ and $\gamma+(n-1)\delta$, respectively, whereas every relative mode has eigenvalues $\alpha-\beta$ and $\gamma-\delta$. The Bogoliubov generator therefore splits into one collective Gaussian block and one degenerate relative block:
\begin{equation}
  \mathcal M
  \cong
  \mathcal M_{\mathrm{col}}
  \oplus
  \left(
    \mathcal M_{\mathrm{rel}}
  \right)^{\oplus(n-1)} .
  \label{eq:permutation-blocks}
\end{equation}
The one-particle transfer is therefore determined by one collective and one relative Bogoliubov block and fixed-order propagation uses tensor products of these blocks. When the input moments and observable are also permutation invariant, equivalent tensor components can be grouped by their permutation-orbit multiplicities. This is the bosonic counterpart of orbit-compressed permutation-equivariant simulation~\cite{Barligea.2026-EnablingLieAlgebraicClassical} and uses the same adjoint-space propagation in a symmetry-adapted basis.

\section{Bounded-sector modules and representation}
\label{app:bounded-sector-proofs}

\subsection{Sector-pair decomposition}

Let $\Pi_N$ denote the projector onto $\mathcal H_N$. For any operator $B$, define
\begin{equation}
  B_{MN}:=\Pi_MB\,\Pi_N
  \in\operatorname{Hom}(\mathcal H_N,\mathcal H_M).
\end{equation}
If $[H,\hat N]=0$, then $H=\bigoplus_NH_N$, where $H_N=\Pi_NH\,\Pi_N$.

\begin{lemma}[Sector-pair invariance]
\label{lem:sectorpairs}
For every number-conserving Hamiltonian $H$,
\begin{equation}
  \mathrm{ad}_{H}(B_{MN})
  =
  \ii\left(H_MB_{MN}-B_{MN}H_N\right)
  \in\operatorname{Hom}(\mathcal H_N,\mathcal H_M).
  \label{eq:sectorpairappendix}
\end{equation}
\end{lemma}
\begin{proof}
Number conservation implies $H\Pi_N=\Pi_NH=\Pi_NH\Pi_N$. Hence
\begin{equation}
\ii[H,B_{MN}]
=
\ii\left(H_MB_{MN}-B_{MN}H_N\right).
\end{equation}
Both terms map $\mathcal H_N$ into $\mathcal H_M$, so every sector-pair projection other than $(M,N)$ vanishes.
\end{proof}

For an observable $O$, define its sector-pair support by
\begin{equation}
  \operatorname{supp}_{\mathrm{sec}}(O)
  :=
  \left\{
    (M,N):
    \Pi_MO\Pi_N\neq0
  \right\}.
\end{equation}
For an observable compressed to a finite sector union, $O_{\mathcal S}=P_{\mathcal S}OP_{\mathcal S}$, this support is contained in $\mathcal S\times\mathcal S$.
For the dimension count, let
$\mathcal V_{\mathbb C}(O)
=
\mathcal V(O)\otimes_{\mathbb R}\mathbb C$
denote the complexified reachable module. Repeated adjoint actions preserve each sector pair, so
\begin{equation}
  \mathcal V_{\mathbb C}(O)
  \subseteq
  \bigoplus_{(M,N)\in\operatorname{supp}_{\mathrm{sec}}(O)}
  \operatorname{Hom}(\mathcal H_N,\mathcal H_M),
  \label{eq:sectorsupportmodule}
\end{equation}
and therefore
\begin{equation}
  \dim_{\mathbb R}\mathcal V(O)
  =
  \dim_{\mathbb C}\mathcal V_{\mathbb C}(O)
  \leq
  \sum_{(M,N)\in\operatorname{supp}_{\mathrm{sec}}(O)}
  d_{n,M}d_{n,N},
  \label{eq:sectorpairdim}
\end{equation}
yielding~\eqref{eq:diagonalmodule}.

On a finite sector union $\mathcal H_{\mathcal S}$, number-conserving evolution has the form
$U_{\mathcal S}=\bigoplus_{N\in\mathcal S}U_N$, and
\begin{equation}
  O_{MN}(t)=U_M^\dagger O_{MN}U_N.
\end{equation}
Writing $\rho_{MN}=\Pi_M\rho \Pi_N$, the readout becomes
\begin{equation}
  \Tr[O(t)\rho]
  =
  \sum_{M,N\in\mathcal S}
  \Tr_{\mathcal H_M}
  \left[
    O_{MN}(t)\rho_{NM}
  \right].
  \label{eq:mixedsectorreadout}
\end{equation}
Equation~\eqref{eq:mixedsectorreadout} shows that number-conserving observables depend only on $\rho_{NN}$, whereas number-changing observables can access coherent off-diagonal blocks.

\subsection{Matrix-unit and MGGM basis}
We adapt the MGGM basis introduced for bounded-Hamming-weight subspaces in Ref.~\cite{Barligea.2026-EnablingLieAlgebraicClassical} to bosonic photon-number sectors, translated here from the skew-Hermitian convention to the Hermitian bracket convention used throughout this paper.

Enumerate the occupation vectors in $\mathcal H_N$ as $\{\bm n^{(1)},\ldots,\bm n^{(d)}\}$, where $d=d_{n,N}$, and write $\ket a=\ket{\bm n^{(a)}}$. The matrix units
\begin{equation}
  \hat E^{ab}:=\ket a\!\bra b
\end{equation}
satisfy
\begin{equation}
  [\hat E^{ab},\hat E^{cd}]
  =
  \delta_{bc}\hat E^{ad}
  -
  \delta_{da}\hat E^{cb}.
  \label{eq:matrixunitbracket}
\end{equation}
An orthonormal real basis of $\operatorname{Herm}(\mathcal H_N)$, with respect to the Hilbert--Schmidt inner product, is
\begin{equation}
\begin{aligned}
  \Lambda_{\mathrm S}^{ab}
  &=
  \frac{1}{\sqrt2}
  \left(
    \hat E^{ab}+\hat E^{ba}
  \right),\\
  \Lambda_{\mathrm A}^{ab}
  &=
  \frac{\ii}{\sqrt2}
  \left(
    \hat E^{ab}-\hat E^{ba}
  \right),
  \qquad a<b,\\
  \Lambda_{\mathrm D}^{a}
  &=\hat E^{aa}.
\end{aligned}
  \label{eq:mggm}
\end{equation}
For each unordered pair of distinct sectors, fix the orientation $M<N$ and define
\begin{equation}
E_{MN}^{ab}:=\ket{a,M}\!\bra{b,N}, \quad 
1\leq a\leq d_{n,M},\quad 1\leq b\leq d_{n,N},
\qquad
(E_{MN}^{ab})^\dagger=E_{NM}^{ba}.
\end{equation}
An orthonormal real basis, with respect to the Hilbert--Schmidt inner product, of the Hermitian operators supported on the paired off-diagonal blocks between $\mathcal H_M$ and $\mathcal H_N$ is
\begin{equation}
X_{MN}^{ab}
=
\frac{E_{MN}^{ab}+E_{NM}^{ba}}{\sqrt2},
\qquad
Y_{MN}^{ab}
=
\frac{\ii(E_{MN}^{ab}-E_{NM}^{ba})}{\sqrt2}.
\end{equation}
Equivalently, each oriented complex sector-pair space $\operatorname{Hom}(\mathcal H_N,\mathcal H_M)$ may be represented directly by the rectangular matrix units $E_{MN}^{ab}$.
These matrix-unit identities reduce every restricted adjoint action to left and right multiplication by the finite generator matrices, whose entries contain the model-dependent bosonic amplitudes.

\section{Controlled photon-sector leakage under squeezing}
\label{app:squeezing}

\subsection{Regularity and coefficient matching}

The coefficient-matching argument of \cref{ssec:squeezing} expands the dynamics in the overall squeezing strength and associates each power of this parameter with one pair-creation or pair-annihilation process. Since bosonic creation and annihilation operators are unbounded, we must first state conditions under which this perturbative expansion is well defined on a finite time interval.

As in \cref{eq:squeezedH}, we decompose the Hamiltonian as $\hat H(r,t)=\hat H_0(t)+r\hat V(t)$, where $\hat H_0(t)$ conserves the total photon number and $r$ controls the overall squeezing strength. We take $\hat V(t)$ to be the general finite-mode quadratic squeezing interaction,
\begin{equation}
\hat V(t)
=
\frac12\sum_{j,l=1}^n
\left[
\xi_{jl}(t)\hat a_j^\dagger\hat a_l^\dagger
+
\xi_{jl}(t)^*\hat a_j\hat a_l
\right],
\qquad
\xi(t)=\xi(t)^\top.
\label{eq:general-squeezing-generator}
\end{equation}
Here, $\xi(t)=[\xi_{jl}(t)]_{j,l=1}^n$ is the complex symmetric matrix of time-dependent squeezing amplitudes. Its diagonal entries $\xi_{jj}(t)$ generate single-mode squeezing, while the off-diagonal entries $\xi_{jl}(t)$ couple modes $j$ and $l$ through two-mode squeezing. The magnitude and phase of each entry determine the strength and phase of the corresponding pair process. The symmetry of $\xi(t)$ reflects the symmetry of the bosonic pair operators, while the factor $1/2$ prevents off-diagonal mode pairs from being counted twice.

Fix a finite time interval $t\in [0,T]$. We assume that $\hat H(r,t)$ generates a unitary propagator $U_r(t,s)$ for every real $r$ in a neighborhood of the origin and all $0\leq s\leq t\leq T$, with
\begin{equation}
\ii\partial_tU_r(t,s)
=
\hat H(r,t)U_r(t,s),
\qquad
U_r(s,s)=\mathds{1}.
\end{equation}
At $r=0$, this propagator reduces to the number-conserving evolution $U_0(t,s)$ generated by $\hat H_0(t)$.

To separate the squeezing events from this number-conserving background evolution, we pass to the interaction picture and define
\begin{equation}
\hat V_{\mathrm I}(t)
:=
U_0(t,0)^\dagger\hat V(t)U_0(t,0).
\label{eq:squeezing-interaction-picture}
\end{equation}
Because $U_0(t,0)$ preserves every photon-number sector, this change of picture does not alter the defining selection rule of the squeezing interaction, as $\hat V_{\mathrm I}(t)$ still changes the total photon number only by $\pm2$.

Let
\begin{equation}
\mathcal D_{\mathrm{fin}}
:=
\bigcup_{N_{\max}<\infty}
\bigoplus_{N=0}^{N_{\max}}\mathcal H_N
\end{equation}
denote the dense domain of states with finite photon-number support. To make the time-ordered integrals below well defined, we assume that $\hat V_{\mathrm I}(t)$ is strongly measurable on $\mathcal D_{\mathrm{fin}}$, meaning that $t\mapsto\hat V_{\mathrm I}(t)|\psi\rangle$ is measurable for every $|\psi\rangle\in\mathcal D_{\mathrm{fin}}$. We further require, uniformly for $t\in[0,T]$,
\begin{equation}
\Pi_M\hat V_{\mathrm I}(t)\Pi_N=0
\quad\text{unless}\quad |M-N|=2,
\qquad
\|\hat V_{\mathrm I}(t)\Pi_N\|
\leq C_T(N+1),
\label{eq:squeezing-regularity}
\end{equation}
where $C_T<\infty$ may depend on the time interval and the squeezing amplitudes, but not on $N$. The first condition states that every squeezing insertion changes the photon number by exactly two. The second bounds the corresponding bosonic enhancement: the norm of a quadratic pair process grows at most linearly with the photon number. Together, these conditions ensure that every finite-order perturbative contribution is well defined and allow us to control the growth of the Dyson coefficients.
We also assume that $\rho_{\mathrm{in}}$ has bounded photon-number support and require that the number-conserving $O$ grows at most polynomially across photon-number sectors,
\begin{equation}
\|\Pi_NO\Pi_N\|
\leq C_O(N+1)^{p_O}
\label{eq:squeezing-observable-growth}
\end{equation}
for fixed $C_O<\infty$ and $p_O\geq0$. Physically, this excludes readouts whose sensitivity grows so rapidly with photon number that an arbitrarily small population in a high-photon sector could produce an uncontrolled contribution. Bounded observables satisfy this condition with $p_O=0$, while fixed-degree number-conserving polynomials in the mode operators satisfy it for some finite $p_O$.

The preceding conditions hold for the standard finite-mode squeezing interaction in \cref{eq:general-squeezing-generator} whenever its time-dependent coefficients remain bounded on $[0,T]$. Pair creation and annihilation change the total photon number by two and satisfy $\|\hat a_j^\dagger\hat a_l^\dagger\Pi_N\|\leq N+2$ and $\|\hat a_j\hat a_l\Pi_N\|\leq N$. The finite sum over modes and the uniform bound on $\xi_{jl}(t)$ can therefore be absorbed into the constant $C_T$ in \cref{eq:squeezing-regularity}. Moreover, $U_0(t,0)$ acts unitarily within each photon-number sector. Conjugating by $U_0(t,0)$ consequently preserves both the $\pm2$ selection rule and the norm of every sector block, so the same estimates hold for $\hat V_{\mathrm I}(t)$.

These two properties also organize the perturbative expansion, meaning every insertion of $\hat V_{\mathrm I}(t)$ moves the state between neighboring sectors of the same parity, while the norm bound controls the amplitude accumulated along such a sequence. To make this structure explicit, define the interaction-picture propagator
\begin{equation}
W_r(t)
:=
U_0(t,0)^\dagger U_r(t,0),
\end{equation}
satisfying
\begin{equation}
\ii\partial_tW_r(t)
=
r\hat V_{\mathrm I}(t)W_r(t),
\qquad
W_r(0)=\mathds{1}.
\end{equation}

Iterating the corresponding integral equation gives the Dyson expansion~\cite{Dyson1949}. The estimates below show that, on bounded-photon input sectors and for sufficiently small $|r|$, this series converges and represents $W_r(t)$, 
\begin{equation}
W_r(t)
=
\sum_{\ell=0}^{\infty}
(-\ii r)^\ell
\int_{0\leq t_\ell\leq\cdots\leq t_1\leq t}
\hat V_{\mathrm I}(t_1)\cdots
\hat V_{\mathrm I}(t_\ell)
\,\dd t_1\cdots\dd t_\ell .
\label{eq:squeezing-dyson-series}
\end{equation}

To bound the Dyson coefficients, we split $\hat V_{\mathrm I}(t)=\hat V_+(t)+\hat V_-(t)$ into its pair-raising and pair-lowering parts, 
\begin{equation}
\hat V_+(t)
:=
\sum_{N\geq0}
\Pi_{N+2}\hat V_{\mathrm I}(t)\Pi_N,
\qquad
\hat V_-(t)
:=
\sum_{N\geq2}
\Pi_{N-2}\hat V_{\mathrm I}(t)\Pi_N,
\end{equation}
each of which inherits the sector-norm bound of \cref{eq:squeezing-regularity}. Since the two components are obtained by projecting onto orthogonal output sectors, each satisfies $\|\hat V_\pm(t)\Pi_N\|\leq C_T(N+1)$.

Expanding a product
of $\ell$ interactions produces
at most $2^\ell$ sequences of raising and lowering operations.
With initial sectors below $N_{\max}$, 
any such sequence will have a support of at most $N_{\max}+2j$ after $j$ operations. Some sequences vanish because they attempt to lower the photon number below zero, but summing over all $2^\ell$ possibilities gives the convenient upper bound
\begin{equation}
\left\|
\hat V_{\mathrm I}(t_1)\cdots
\hat V_{\mathrm I}(t_\ell)P_{\mathcal S}
\right\|
\leq
(2C_T)^\ell
\prod_{j=0}^{\ell-1}(N_{\max}+2j+1).
\label{eq:squeezing-dyson-bound}
\end{equation}
The time-ordered integration region in \cref{eq:squeezing-dyson-series} has volume $t^\ell/\ell!\leq T^\ell/\ell!$, such that the norm of its $\ell$th term on the initial sector space is bounded by
\begin{equation}
\frac{(|r|T)^\ell}{\ell!}
(2C_T)^\ell
\prod_{j=0}^{\ell-1}(N_{\max}+2j+1).
\label{eq:squeezing-time-ordered-bound}
\end{equation}
Writing $a=(N_{\max}+1)/2$, the remaining product becomes
\begin{equation}
\begin{aligned}
\prod_{j=0}^{\ell-1}(N_{\max}+2j+1)
=
2^\ell\prod_{j=0}^{\ell-1}(a+j)
=
2^\ell\frac{\Gamma(\ell+a)}{\Gamma(a)}.
\end{aligned}
\end{equation}
Substituting this identity into \cref{eq:squeezing-time-ordered-bound} gives
\begin{equation}
\frac{\Gamma(\ell+a)}{\Gamma(a)\,\ell!}
(4|r|C_TT)^\ell.
\label{eq:squeezing-convergence-bound}
\end{equation}

The gamma-function ratio grows only polynomially in $\ell$, so the series converges uniformly for $4|r|C_TT<1$. The polynomial sector growth in \cref{eq:squeezing-observable-growth} adds only another polynomial factor. Hence $F(r,t)$ is analytic in $r$ uniformly on $[0,T]$ within a nonzero neighborhood of the origin. The projected evolution acts on a finite-dimensional space $P_k\mathcal H$ and is therefore analytic in $r$ as well.

We now compare the coefficients of the two series. A contribution of total order $q$ contains $q_+$ insertions of $\hat V_{\mathrm I}$ on the forward branch and $q_-$ on the backward branch, with $q=q_++q_-$. Each insertion moves between adjacent parity-compatible photon sectors. Because $\rho_{\mathrm{in}}$ begins and ends in $\mathcal S$ and $O$ is number conserving, concatenating the forward branch with the reverse backward branch produces a sector path whose endpoints lie in $\mathcal S$. If this contribution differs from the projected one, the path must leave $\mathcal S^{(k)}$. It then requires at least $k+1$ steps to reach an omitted sector and at least $k+1$ further steps to return to $\mathcal S$. No coefficient below total order $2(k+1)$ can therefore depend on an omitted sector, which proves \cref{eq:coefficientmatching}.

The difference $F(r,t)-F_k(r,t)$ is analytic and has a zero of order at least $2(k+1)$ at $r=0$. It can thus be written as $r^{2(k+1)}G_k(r,t)$, where $G_k$ is continuous on every compact subset of the common analytic domain. Taking its supremum over $|r|\leq r_\ast$ and $t\in[0,T]$ proves \cref{eq:squeezeerror}. For a number-changing observable, the two branches need not end in the same sector. An omitted sector can then contribute after only $k+1$ insertions, so the generic guarantee weakens to $\mathcal O(r^{k+1})$.

Finally, $\mathcal S^{(k)}\subseteq\{0,\ldots,N_{\max}+2k\}$ gives
\begin{equation}
\dim(P_k\mathcal H)
\leq
\binom{n+N_{\max}+2k}{N_{\max}+2k}
=
\mathcal O(n^{N_{\max}+2k}),
\end{equation}
and squaring this bound yields the stated dimension of $\operatorname{End}(P_k\mathcal H)$.

\subsection{Failure of exact finite closure}

The finite-band approximation does not imply an exact finite reachable module for Kerr-plus-squeezing dynamics. For one mode, consider the Hermitian generators
\begin{equation}
\hat H_{\mathrm{ph}}=\hat n,
\qquad
\hat H_{\mathrm K}=\hat n^2,
\qquad
\hat H_{\mathrm S}
=
\frac12(\hat a^2+\hat a^{\dagger2})
\end{equation}
on the common invariant domain spanned by finite Fock states. Starting from $\hat n$, the squeezing action gives
\begin{equation}
\mathrm{ad}_{\hat H_{\mathrm S}}(\hat n)
=
\ii(\hat a^2-\hat a^{\dagger2}).
\end{equation}
Commuting this operator with $\hat H_{\mathrm{ph}}$ shows that the complexified reachable module contains $\hat a^2$. Moreover,
\begin{equation}
[\hat n^2,\hat a^2]
=
\hat a^2\!\left[(\hat n-2)^2-\hat n^2\right]
=
\hat a^2(-4\hat n+4),
\end{equation}
and hence
\begin{equation}
\mathrm{ad}_{\hat H_{\mathrm K}}^j(\hat a^2)
=
\ii^j\hat a^2(-4\hat n+4)^j,
\qquad
j\geq0.
\label{eq:infiniteorbit}
\end{equation}
These operators are linearly independent. Indeed, applying a finite linear combination of them to every $\ket m$ with $m\geq2$ would force a polynomial in $-4m+4$ to vanish at infinitely many points. All its coefficients must therefore vanish. Thus the complexified, and consequently the real, reachable module of $\hat n$ is infinite-dimensional. This argument establishes an infinite observable orbit directly; an infinite DLA alone would not suffice.

\subsection{Exact squeezed-vacuum tails}

The preceding estimate controls generic interleaved dynamics locally in the perturbative coupling $r$. For exactly solvable squeezed vacua, photon-number tails give nonperturbative cutoff bounds. We denote the integrated squeeze parameter in these examples by $s$ to distinguish it from $r$ in \cref{eq:squeezedH}.

For the two-mode convention
\begin{equation}
S_2(s)
=
\exp\!\left[
s(\hat a^\dagger\hat b^\dagger-\hat a\hat b)
\right],
\end{equation}
the two-mode squeezed vacuum is~\cite{Weedbrook2012}
\begin{equation}
S_2(s)\ket{0,0}
=
\operatorname{sech}s
\sum_{m=0}^{\infty}
(\tanh s)^m\ket{m,m},
\end{equation}
up to an inessential phase convention. Truncating after $k$ generated pairs discards the exact probability
\begin{equation}
\tau_k^{(2)}(s)
=
(\tanh s)^{2(k+1)}.
\label{eq:twomodetail}
\end{equation}
For the single-mode convention
\begin{equation}
S_1(s)
=
\exp\!\left[
\frac{s}{2}(\hat a^2-\hat a^{\dagger2})
\right],
\end{equation}
we have~\cite{Weedbrook2012}
\begin{equation}
S_1(s)\ket0
=
\frac{1}{\sqrt{\cosh s}}
\sum_{m=0}^{\infty}
(-\tanh s)^m
\frac{\sqrt{(2m)!}}{2^m m!}\ket{2m}.
\end{equation}
The probability above $2k$ photons obeys
\begin{equation}
\begin{aligned}
\tau_k^{(1)}(s)
=
\frac1{\cosh s}
\sum_{m=k+1}^{\infty}
\frac{(2m)!}{2^{2m}(m!)^2}
\tanh^{2m}s
\leq
\min\!\left\{
1,\,
\cosh s\,(\tanh s)^{2(k+1)}
\right\},
\end{aligned}
\label{eq:singlemodetail}
\end{equation}
where we used $\binom{2m}{m}/4^m\leq1$.

Let $P$ be a photon-number cutoff and let $\tau=\langle\psi|(\mathds{1}-P)|\psi\rangle$. If $O$ is bounded and number conserving, then $[O,P]=0$ and
\begin{equation}
\left|
\langle\psi|O|\psi\rangle
-
\langle\psi|POP|\psi\rangle
\right|
\leq
\|O\|\tau.
\label{eq:tailobservable}
\end{equation}
For the normalized projected state $\ket{\psi_P}=P\ket\psi/\sqrt{1-\tau}$, the corresponding bound is $2\|O\|\tau$. Number-changing observables can produce cross terms of order $\sqrt{\tau}$. Unbounded observables such as $\hat N$ instead require weighted-tail bounds, for example on $\sum_{N>N_{\mathrm{cut}}}Np_N$.
For $s\neq0$, the two-mode condition $\tau_k^{(2)}(s)\leq\epsilon$ is equivalent to
\begin{equation}
k+1
\geq
\frac{\log(1/\epsilon)}
{2\log(\coth|s|)}.
\label{eq:accuracyband}
\end{equation}
For these squeezed-vacuum examples, fixed $s$ and fixed accuracy therefore require a band depth independent of $n$. By contrast, $\epsilon=n^{-c}$ gives $k=\Theta(\log n)$ and makes the retained-sector dimension $\dim(P_k\mathcal H)$ quasi-polynomial in $n$,
\begin{equation}
n^{N_{\max}+2k}
=
\exp\!\left[\mathcal O((\log n)^2)\right].
\end{equation}
General bosonic leakage bounds control truncation errors by restricting how rapidly number-changing matrix elements grow with occupation~\cite{Tong.2022}. Their sublinear-growth assumptions do not directly cover quadratic pair creation, whose sector-block norm grows linearly with $N$. The local estimate in \cref{eq:squeezeerror} instead follows from the $\Delta N=\pm2$ selection rule, the explicit linear block-growth bound in \cref{eq:squeezing-regularity}, and perturbative path counting. A nonperturbative, time-uniform leakage bound for generic interleaved Kerr-plus-squeezing dynamics remains open.

\section{Boundary calculations for non-Gaussian bosonic generators}
\label{app:boundary-calculations}

\subsection{Projective cubic flow}
\label{app:projective-cubic-flow}

We justify \cref{ex:projective-cubic} in the Hermitian convention used in the main text. Let $K_+=\frac13(\hat x^2\hat p+\hat x\hat p\hat x+\hat p\hat x^2)$~\cite{Turbiner.1988,GonzalezLopez.1991-QuasiExactly}. 
All identities below are understood on the finite-particle domain
\begin{equation}
\mathcal D_{\mathrm{fin}}
=
\spann_{\mathbb C}\{\ket n:n\geq0\},
\end{equation}
which is invariant under polynomials in $\hat a$ and $\hat a^\dagger$, and hence under $\hat x$, $\hat p$, and $K_+$. Since $[\hat p,\hat x]=-\ii$, each ordered term has commutator $[\hat x^2\hat p,\hat x]=[\hat x\hat p\hat x,\hat x]=[\hat p\hat x^2,\hat x]=-\ii\hat x^2$, and therefore
\begin{equation}
  \mathrm{ad}_{K_+}(\hat x)=\ii[K_+,\hat x]=\hat x^2.
\end{equation}
More generally, for every polynomial $f(\hat x)$,
\begin{equation}
\mathrm{ad}_{K_+}\!\left(f(\hat x)\right)
=
\hat x^2f'(\hat x).
\end{equation}
Iterating from $f(\hat x)=\hat x$ therefore gives
\begin{equation}
\mathrm{ad}_{K_+}^{\,r}(\hat x)
=
r!\,\hat x^{r+1},
\qquad r\geq0.
\end{equation}
The operators $\{\hat x^{r+1}\}_{r\geq0}$ are linearly independent on $\mathcal D_{\mathrm{fin}}$. Indeed, $\hat x^m\ket0$ has a nonzero component along $\ket m$, whereas $\hat x^\ell\ket0$ has no such component for $\ell<m$.

\subsection{Kerr orbit on the full Fock space}
\label{app:kerr-full-fock-orbit}

Let $\hat n=\hat a^\dagger\hat a$ and $H=\chi\hat n^2$ with $\chi\neq0$. The generator algebra $\spann_{\mathbb R}\{H\}$ is one-dimensional. Since multiplication of the generator by a nonzero scalar does not change the dimension of its adjoint orbit, we first work with $\hat n^2$ on the finite-particle domain $\mathcal D_{\mathrm{fin}}=\spann\{\ket n:n\geq0\}$. Using $[\hat n,\hat a]=-\hat a$, we find
\begin{equation}
  [\hat n^2,\hat a]=-\hat a(2\hat n-1),
\end{equation}
and hence
\begin{equation}
  \mathrm{ad}_{\hat n^2}(\hat a)=-\ii\hat a(2\hat n-1).
\end{equation}
Because every polynomial in $\hat n$ commutes with $\hat n^2$, iteration gives the exact right-ordered identity
\begin{equation}
  \mathrm{ad}_{\hat n^2}^m(\hat a)
  =
  (-\ii)^m\hat a(2\hat n-1)^m.
\end{equation}
These operators are linearly independent. Indeed, taking the matrix element $\langle k-1|\cdot|k\rangle$ of a finite linear relation reduces it to a polynomial in $2k-1$ that vanishes for every $k\geq1$, and hence has all coefficients equal to zero.

The same conclusion holds for the Hermitian quadrature $\hat x=(\hat a+\hat a^\dagger)/\sqrt2$, since
\begin{equation}
\mathrm{ad}_{\hat n^2}^m(\hat x)
=
\frac{1}{\sqrt2}
\left[
(-\ii)^m\hat a(2\hat n-1)^m
+
\ii^m\hat a^\dagger(2\hat n+1)^m
\right].
\end{equation}
Taking the lowering matrix elements $\langle k-1|\cdot|k\rangle$ proves that these Hermitian operators are linearly independent. Therefore $\mathcal V(\hat x)$ is infinite dimensional.

Equivalently, the exact Heisenberg evolution is
\begin{equation}
  e^{\ii t\chi\hat n^2}\hat a e^{-\ii t\chi\hat n^2}=\hat a e^{-\ii t\chi(2\hat n-1)}.
\end{equation}
For $\chi t\notin\pi\mathbb Z$, the phase $e^{-\ii t\chi(2\hat n-1)}$ is not a polynomial in $\hat n$ on the nonnegative-integer spectrum. The evolved operator is therefore not a finite linear combination of $\hat a\hat n^m$.

\section{Additional numerical demonstrations}

\subsection{Gaussian- and non-Gaussian-input validation}
\label{app:num-gaussian}

Here, we first validate the sign and normalization conventions of the Bogoliubov propagation~\cite{Weedbrook2012} (see \cref{ssec:Gaussianpropagation}) using the single-mode squeezing Hamiltonian
\begin{equation}
\hat H_{\mathrm{sq}}
=
\frac{\ii}{2}
\left(
\xi^*\hat a^2-\xi\hat a^{\dagger2}
\right).
\end{equation}
For a vacuum input and real $\xi=r$, the exact mean photon number is $\langle\hat n(t)\rangle=\sinh^2(rt)$. We set $r=0.6$ and sample $200$ times over $0\leq t\leq1.6$. At eight times, an independent 100-dimensional truncated-Fock calculation, retaining occupations $0,\ldots,99$, agrees with the analytic result to $9.02\times10^{-14}$, see \cref{fig:num-gaussian}(a). Meanwhile the $2\times2$ Bogoliubov transfer agrees with the analytic curve to $6.66\times10^{-16}$.

In \cref{fig:num-gaussian}(b), we then benchmark the complete number-correlation matrix
$
M_{ij}
=
\left\langle
\hat n_i\hat n_j
\right\rangle
$
for the number-state input $\ket{1,1,\ldots,1}$ after a Haar-random passive interferometer~\cite{mezzadri2007generaterandommatricesclassical}. We use mode numbers $n=64,128,256,512,1024$. The measured times depend on the workstation's dense-linear-algebra backend and serve only as practical benchmarks. The implemented contraction evaluates three dense $n\times n$ matrix products and therefore has arithmetic complexity $\Theta(n^3)$, independent of the hardware and linear-algebra backend.

The input contains $N=n$ photons, so its ambient number sector has dimension
$
d_{n,n}
=
\tbinom{2n-1}{n}$~\cite{Aaronson2013}.
The calculation does not propagate a state or a generic operator in this exponentially large sector. Instead, it evaluates the complete matrix of fixed-order moments directly from the interferometer and the input occupations.

\begin{figure}[t]
\centering
\includegraphics[width=0.8\textwidth]{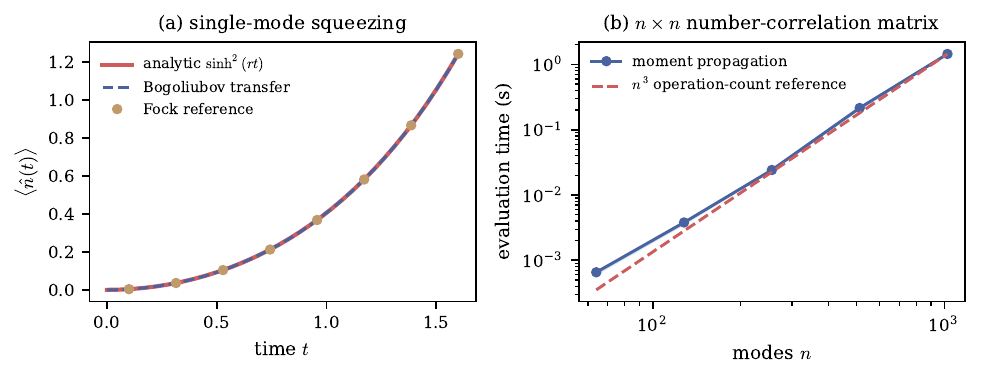}
\caption{Gaussian propagation and fixed-order correlation benchmark. (a) Mean photon number under single-mode squeezing with $r=0.6$, evaluated analytically, with the $2\times2$ Bogoliubov transfer, and with a 100-dimensional truncated-Fock reference. (b) Evaluation time for the complete $n\times n$ number-correlation matrix. Points and shaded regions show medians and interquartile ranges over seven repetitions after warm-up. The dashed line shows the $n^3$ operation-count reference and is not a fit to the hardware timings.}
\label{fig:num-gaussian}
\end{figure}

We test non-Gaussian number-state inputs using the beamsplitter
\begin{equation}
W(\theta)
=
\begin{pmatrix}
\cos\theta&\sin\theta\\
-\sin\theta&\cos\theta
\end{pmatrix}
\end{equation}
and the Fock input $|1,1\rangle$. The coincidence between the two output modes is
$
\left\langle
\hat n_1\hat n_2
\right\rangle
=
\cos^2(2\theta)$.
Over $200$ angles in $0\leq\theta\leq\pi/2$, moment propagation agrees with this expression to $4.44\times10^{-16}$. An independent three-level local Fock calculation at $\theta=0,\pi/8,\pi/4,3\pi/8,\pi/2$ agrees to $2.22\times10^{-16}$; see \cref{fig:num-nongaussian-input}(a).
For \cref{fig:num-nongaussian-input}(b), we generate twelve Haar-random four-mode interferometers~\cite{mezzadri2007generaterandommatricesclassical} and input occupation $(1,1,0,0)$. The resulting comparison contains all $16$ entries of the number-correlation matrix for each interferometer, giving $192$ entrywise tests. Moment propagation agrees with direct three-level Fock calculations to a maximum error of $4.44\times10^{-16}$. These calculations validate fixed-order correlation functions for non-Gaussian inputs.
\begin{figure}[t]
\centering
\includegraphics[width=0.8\textwidth]{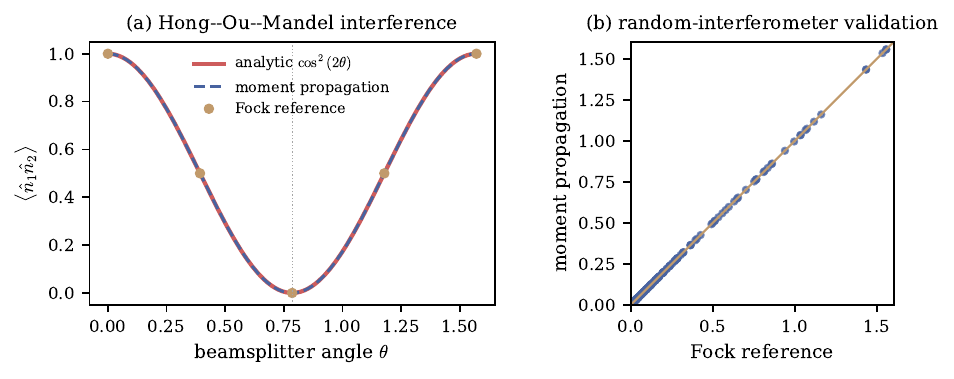}
\caption{Number correlations for non-Gaussian number-state inputs. (a) Hong--Ou--Mandel interference~\cite{Hong.1987} for $\ket{1,1}$, compared with the analytic expression and a three-level Fock reference. (b) Entrywise comparison with direct Fock calculations for all $192$ number-correlation values obtained from twelve seeded four-mode interferometers.}
\label{fig:num-nongaussian-input}
\end{figure}

\subsection{Gradient validation and nonlinear control}
\label{app:num-gradients}

We consider an open $L\times L$ Hofstadter lattice with Landau-gauge hopping
\begin{equation}
\hat H_{\mathrm{hop}}
=
-J\sum_{x,y}
\left(
\hat a_{x+1,y}^\dagger\hat a_{x,y}
+\mathrm{h.c.}
\right)
-J\sum_{x,y}
\left(
e^{-\ii2\pi\phi x}
\hat a_{x,y+1}^\dagger\hat a_{x,y}
+\mathrm{h.c.}
\right).
\label{eq:num-hofstadter-hopping}
\end{equation}
Writing $M=L^2$, each control layer has the form
\begin{equation}
\hat U_\ell
=
\left(
\prod_{j=0}^{M-1}
e^{-\ii\varphi_{\ell j}\hat n_j}
\right)
e^{-\ii\kappa_\ell\hat H_{\mathrm{Kerr}}}
e^{-\ii\tau_\ell\hat H_{\mathrm{hop}}/J},
\qquad
\hat H_{\mathrm{Kerr}}
=
\sum_{j=0}^{M-1}\hat n_j^2.
\label{eq:num-control-layer}
\end{equation}
The full circuit is $\hat U=\hat U_p\cdots\hat U_1$ at depth $p$. Within the two-photon sector, replacing $\hat n_j^2$ by $\hat n_j(\hat n_j-1)$ changes only a global phase. We compare the Kerr circuit with the nested passive ansatz obtained by setting every $\kappa_\ell$ to zero.

Two photons start at adjacent sites $(0,0)$ and $(0,1)$, corresponding to the flattened indices $0$ and $1$. We optimize their probability of forming a doublon at the nearby target $(1,1)$, whose flattened index is $j_\star=L+1$. We define the pair-density profile and loss by
\begin{equation}
p_j(\bm\theta)
=
\frac12
\left\langle
\hat n_j(\hat n_j-1)
\right\rangle_{\bm\theta},\qquad
\mathcal L(\bm\theta)
=
\sum_j
\left[
p_j(\bm\theta)-\delta_{j,L+1}
\right]^2.
\label{eq:num-control-objective}
\end{equation}
For two photons entering distinct input modes $a$ and $b$ of a passive interferometer $W$, the probability that both leave through mode $j$ satisfies
\begin{equation}
p_j^{\mathrm{passive}}
=
2|W_{ja}W_{jb}|^2
\leq
\frac12.
\label{eq:passive-doublon-bound}
\end{equation}
Indeed, row unitarity gives $|W_{ja}|^2+|W_{jb}|^2\leq1$, while $2xy\leq(x+y)^2/2$. A target probability of $p_{j_\star}>1/2$ is therefore inaccessible to any passive linear-optical circuit acting on this input. Notice also that the comparison does not depend on matching the number of trainable parameters.

We use flux $\phi=1/4$ and lattice sizes $L=5,7,9$, corresponding to $M=25,49,81$ modes and two-photon sector dimensions $325$, $1225$, and $3321$. The $L=5,7$ calculations use depths one through eight, while the $L=9$ calculations use depths four, six, and eight. For each size, depth, and ansatz, we run Adam~\cite{Kingma.2014-AdamMethodStochastic} for $250$ steps with learning rate $0.06$ from $30$ Gaussian initializations with standard deviation $0.15$.
Panel~\ref{fig:num-gradients}(a) compares reverse-mode gradients with centered finite differences of step $10^{-6}$ for the $L=5$, depth five circuit. Across all $135$ parameters, the maximum absolute discrepancy is $1.8\times10^{-10}$ at a point where the gradient norm is of order $10^{-1}$.

The depth-three results for $L=5$ separate into two groups. $12$ of the $30$ Kerr runs exceed the passive bound, reaching probabilities between $0.503$ and $0.793$, while the remaining runs converge near zero. The success fraction increases with depth and reaches unity at depth six. At depth eight, the median Kerr probabilities are $0.9856$, $0.9830$, and $0.9857$ for $L=5,7,9$, respectively, and every run exceeds the passive bound; see \cref{fig:num-gradients}(b). Panel (c) shows one independently optimized depth five solution with $p_{j_\star}=0.9748$.

These calculations validate the reverse-mode gradient implementation and show that the Kerr-augmented circuit reaches target probabilities not accessible to passive linear optics. Because the target remains local as $L$ increases, the size sweep tests optimization in progressively larger two-photon sectors rather than control over increasing distances. 

\begin{figure}[t]
\centering
\includegraphics[width=\textwidth]{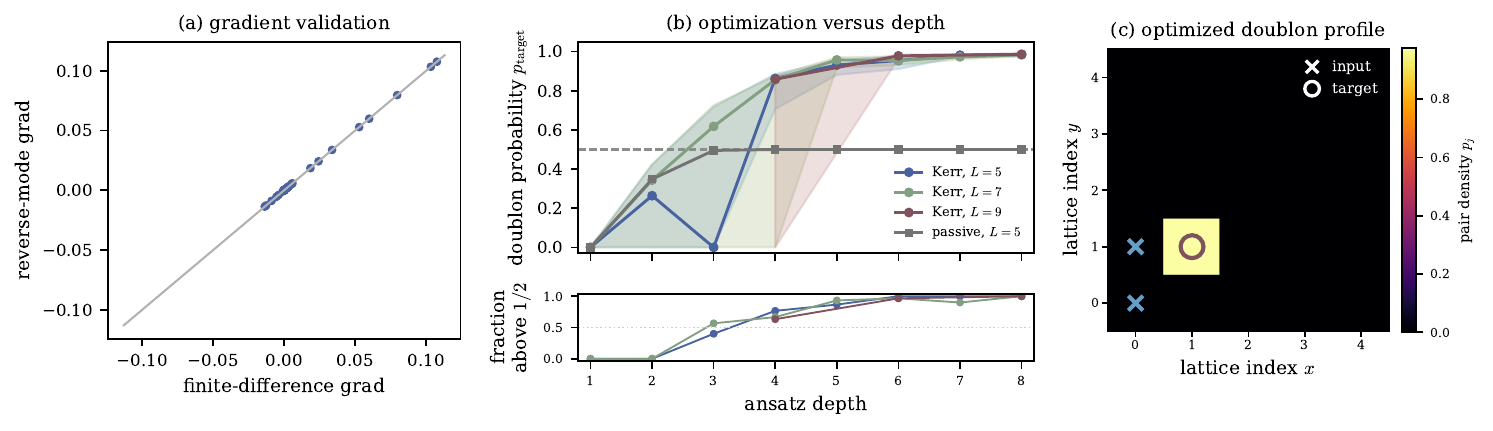}
\caption{Reverse-mode gradient validation and nonlinear two-photon control. (a) Reverse-mode derivatives compared with centered finite differences for the $L=5$, depth-five circuit. The diagonal denotes equality, and the maximum absolute discrepancy is $1.8\times10^{-10}$. (b) Median target-doublon probability and interquartile range over $30$ initializations (upper panel), together with the fraction of runs exceeding the universal passive bound $p_{\mathrm{target}}=1/2$ (lower panel). The passive curve shown corresponds to $L=5$. (c) Pair-density profile of one separately optimized depth-five Kerr circuit with $p_{\mathrm{target}}=0.9748$. Crosses mark the occupied input sites and the circle marks the target.}
\label{fig:num-gradients}
\end{figure}

\subsection{Exact nilpotent cubic-phase dynamics}
\label{app:num-cph}

We next test our second exact non-Gaussian mechanism identified in \cref{sec:boundaries}, namely bounded-degree position phases interspersed with translations. Here we consider the depth-$L$ circuit
\begin{equation}
U_L=\prod_{\ell=1}^{L}e^{-\ii V_\ell(\hat{\bm x})}e^{-\ii\bm s_\ell^{\top}\hat{\bm p}},
\label{eq:num-cph-circuit}
\end{equation}
where each $V_\ell$ has degree at most three. Exact polynomial composition reduces propagation in the Heisenberg picture to
\begin{equation}
U_L^\dagger\hat{\bm x}\,U_L=\hat{\bm x}+\bm c_L,
\quad
U_L^\dagger\hat{\bm p}\,U_L=\hat{\bm p}-\nabla W_L(\hat{\bm x}),
\label{eq:num-cph-shear}
\end{equation}
where $W_L$ has the same bounded degree as the layer phases.
We therefore propagate polynomial coefficients and contract products of the transformed quadratures with input moments, without introducing a Fock-space cutoff.

Panel~\ref{fig:num-cph}(a) uses a single-mode vacuum input and eight cubic layers, each consisting of a position translation $s_\ell$ followed by the cubic phase $V_\ell(\hat x)=\gamma_\ell\hat x^3$. For the experiments, we choose a fixed depth-eight circuit whose translations and phase strengths include both signs and satisfy $|s_\ell|\leq0.22$ and $|\gamma_\ell|\leq0.055$; the complete layer-by-layer parameters are provided with the accompanying numerical data. We evaluate the fourth momentum cumulant
\begin{equation}
\kappa_4(\hat p)
=
\langle(\Delta\hat p)^4\rangle
-3\langle(\Delta\hat p)^2\rangle^2,
\label{eq:num-cph-cumulant}
\end{equation}
which vanishes for Gaussian quadrature statistics~\cite{Weedbrook2012}. After eight layers, the fourth cumulant reaches the value $\kappa_4(\hat p)=0.483323$, giving a direct witness of the output state being non-Gaussian. The coefficient propagation is exact within the finite reachable module and introduces no Fock-space or spatial-grid truncation, such that only floating-point roundoff remains. 
We verify the result independently by representing the wavefunction on $[-10,10)$, applying translations and momentum operators through fast Fourier transforms (FFTs), and applying the cubic phases pointwise in position space. With $2048$ grid points, the two calculations agree at every circuit depth to within $2.66\times10^{-15}$. Doubling the grid resolution to $4096$ points changes the FFT reference by at most $4.44\times10^{-15}$. Thus, the quoted value is stable well within double precision.

Panel~\ref{fig:num-cph}(b) tests a correlated two-mode phase operator
\begin{equation}
\hat V_\ell
=
\gamma_\ell
\left(
\hat x_0^2\hat x_1
-0.55\hat x_0\hat x_1^2
+0.20\hat x_0^3
-0.15\hat x_1^3
\right)
\end{equation}
in a circuit of depth $L=6$. As in the single-mode calculation above, we use one fixed deterministic instance, with translation vectors and phase strengths including both signs within $|s_\ell|\leq0.16$ and $|\gamma_\ell|\leq0.04$.
For a two-mode vacuum input, the final connected correlator produced by the circuit is 
\begin{equation}
\langle\hat p_0\hat p_1\rangle_c
:=
\langle\hat p_0\hat p_1\rangle
-\langle\hat p_0\rangle\langle\hat p_1\rangle
=
-2.29042\times10^{-3}.
\end{equation}

Independent single-mode translations and phase gates would preserve the product structure of the input and leave this connected correlator zero. The value being nonzero therefore isolates the effect of the cross-mode monomials in $\hat V_\ell$. The coefficient calculation agrees with an independent FFT wavefunction calculation on $[-8,8)^2$ to within $8.89\times10^{-18}$ at every depth, while increasing the grid from $128^2$ to $192^2$ points changes the reference by at most $2.43\times10^{-17}$. 

Panel~\ref{fig:num-cph}(c) benchmarks depth-three circuits with homogeneous phase degree $m=3,4,5$. Translation amplitudes scale as $0.03/\sqrt n$, while deterministic polynomial coefficients scale as $0.01/\sqrt n$. At the largest sizes, the dimensions are ${d_{768,3}=76\,089\,473}$, ${d_{64,4}=814\,449}$, and ${d_{32,5}=435\,929}$, respectively.
The timings follow the $O(n^m)$ fixed-degree dimension guides and establish practical coefficient propagation beyond the cubic example. The local and HPC curves were measured on different hardware, so their absolute prefactors should not be compared.

\begin{figure}[htbp]
\centering
\includegraphics[width=\textwidth]{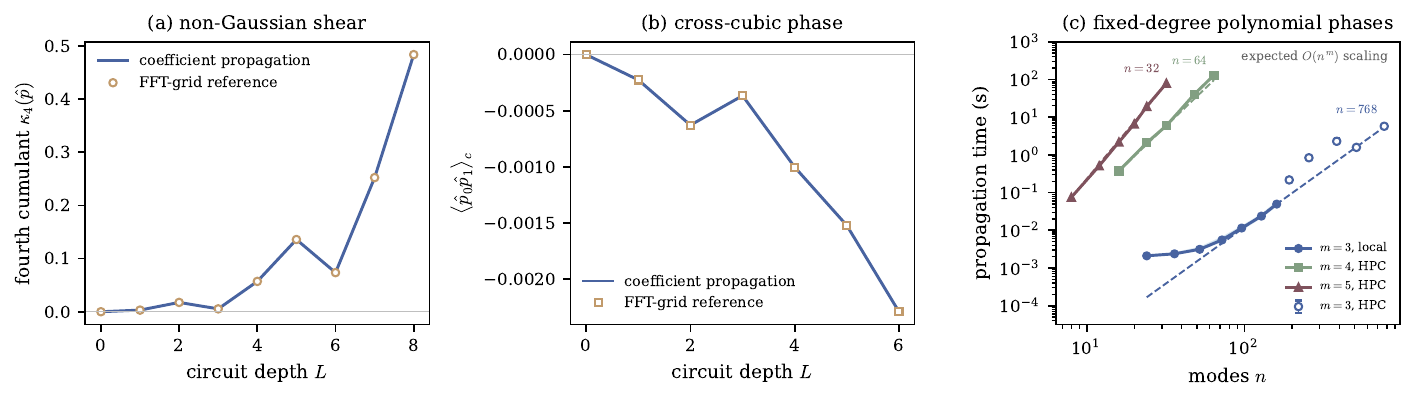}
\caption{Exact nilpotent polynomial-phase dynamics. (a) Fourth momentum cumulant generated by a single-mode cubic-phase circuit, evaluated by coefficient propagation and an independently grid-refined FFT reference. Its nonzero value witnesses the non-Gaussian output. (b) Connected momentum correlator generated by a two-mode cross-cubic phase and validated against a two-dimensional FFT-grid calculation. (c) Median coefficient-propagation time and interquartile ranges over seven repetitions at depth $L=3$. Filled connected circles show the local benchmark for $m=3$, while open circles show its HPC-cluster extension; the $m=4,5$ curves are HPC-cluster benchmarks. Dashed lines indicate the expected $\mathcal O(n^m)$ scaling; their vertical offsets are guides only because runtime prefactors across different hardware are not directly comparable.}
\label{fig:num-cph}
\end{figure}

\twocolumngrid

\bibliography{literature,new_literature,numerics_citations,topology_citations}

\end{document}